\documentclass[reqno,11pt]{amsart}
\usepackage{amsmath, latexsym, amsfonts, amssymb, amsthm, amscd}
\usepackage{graphics,epsf,psfrag}
\usepackage{graphics,epsf,psfrag}
\usepackage{stmaryrd}
\newcommand{\ext}{\mathrm{ext}}
\newcommand{\sm}{\mathrm{small}}
\renewcommand{\bar}{\overline}

\newcommand{\lint}{\llbracket}
\newcommand{\rint}{\rrbracket}

\newcommand{\bg}{\mathbf{g}}
\numberwithin{equation}{section}

\newtheorem{theorema}{Theorem}

\newtheorem{theorem}{Theorem}[section]
\newtheorem{lemma}[theorem]{Lemma}
\newtheorem{proposition}[theorem]{Proposition}
\newtheorem{cor}[theorem]{Corollary}
\newtheorem{rem}[theorem]{Remark}
\newtheorem{definition}[theorem]{Definition}

\newcommand{\Diam}{\mathrm{Diam}}

\newcommand{\ind}{\mathbf{1}}
\newcommand{\restrict}{\!\!\upharpoonright}
\newcommand{\Supp}{\mathrm{Supp}}
\newcommand{\larg}{\mathrm{large}}

\renewcommand{\tilde}{\widetilde}
\renewcommand{\hat}{\widehat}
\newcommand{\cc}{\complement}

\newcommand{\cQ}{{\ensuremath{\mathcal Q}} }

\newcommand{\cA}{{\ensuremath{\mathcal A}} }
\newcommand{\cB}{{\ensuremath{\mathcal B}} }

\newcommand{\cP}{{\ensuremath{\mathcal P}} }

\newcommand{\cH}{{\ensuremath{\mathcal H}} }
\newcommand{\cC}{{\ensuremath{\mathcal C}} }

\newcommand{\cL}{{\ensuremath{\mathcal L}} }

\newcommand{\cZ}{{\ensuremath{\mathcal Z}} }

\newcommand{\cK}{{\ensuremath{\mathcal K}} }

\newcommand{\bP}{{\ensuremath{\mathbf P}} }

\newcommand{\bE}{{\ensuremath{\mathbf E}} }

\newcommand{\bC}{{\ensuremath{\mathbf C}} }

\newcommand{\bL}{{\ensuremath{\mathbf L}} }

\newcommand{\bB}{{\ensuremath{\mathbf B}} }
\newcommand{\cR}{{\ensuremath{\mathcal R}} }

\DeclareMathSymbol{\leqslant}{\mathalpha}{AMSa}{"36} 
\DeclareMathSymbol{\geqslant}{\mathalpha}{AMSa}{"3E} 
\DeclareMathSymbol{\eset}{\mathalpha}{AMSb}{"3F}     
\DeclareMathOperator*{\union}{\bigcup}
\newcommand{\sumtwo}[2]{\sum_{\substack{#1 \\ #2}}} 
\newcommand{\limtwo}[2]{\lim_{\substack{#1 \\ #2}}}     

\newcommand{\bbC}{{\ensuremath{\mathbb C}} }

\newcommand{\bbN}{{\ensuremath{\mathbb N}} }

\newcommand{\bbP}{{\ensuremath{\mathbb P}} }

\newcommand{\bbR}{{\ensuremath{\mathbb R}} }

\newcommand{\bbZ}{{\ensuremath{\mathbb Z}} }

\newcommand{\gb}{\beta}
\newcommand{\gga}{\gamma}            

\newcommand{\gep}{\varepsilon}       

\newcommand{\gG}{\Gamma}

\newcommand{\go}{\omega}

\newcommand{\gO}{\Omega}
\newcommand{\gl}{\lambda}
\newcommand{\gL}{\Lambda}

\newcommand{\tr}{\mathrm{tr}}

\makeatletter
\def\captionfont@{\footnotesize}
\def\captionheadfont@{\scshape}

\long\def\@makecaption#1#2{%
  \vspace{2mm}
  \setbox\@tempboxa\vbox{\color@setgroup
    \advance\hsize-6pc\noindent
    \captionfont@\captionheadfont@#1\@xp\@ifnotempty\@xp
        {\@cdr#2\@nil}{.\captionfont@\upshape\enspace#2}%
    \unskip\kern-6pc\par
    \global\setbox\@ne\lastbox\color@endgroup}%
  \ifhbox\@ne 
    \setbox\@ne\hbox{\unhbox\@ne\unskip\unskip\unpenalty\unkern}%
  \fi
  \ifdim\wd\@tempboxa=\z@ 
    \setbox\@ne\hbox to\columnwidth{\hss\kern-6pc\box\@ne\hss}%
  \else 
    \setbox\@ne\vbox{\unvbox\@tempboxa\parskip\z@skip
        \noindent\unhbox\@ne\advance\hsize-6pc\par}%
\fi
  \ifnum\@tempcnta<64 
    \addvspace\abovecaptionskip
    \moveright 3pc\box\@ne
  \else 
    \moveright 3pc\box\@ne
    \nobreak
    \vskip\belowcaptionskip
  \fi
\relax
}
\makeatother
\def\writefig#1 #2 #3 {\rlap{\kern #1 truecm
\raise #2 truecm \hbox{#3}}}

\newcommand{\ldom}{\preccurlyeq}
\newcommand{\tf}{\textsc{f}}

\title[Wetting and Layering for SOS II]{Wetting and layering for Solid-on-Solid II: Layering transitions, Gibbs states, and regularity
of the free energy}

\author{Hubert Lacoin}
\address{
  IMPA, Institudo de Matem\'atica Pura e Aplicada, Estrada Dona Castorina 110
Rio de Janeiro, CEP-22460-320, Brasil. 
}

\begin{document}

\begin{abstract}
We consider the Solid-On-Solid model interacting with  a wall, which is  the statistical mechanics model associated with 
the integer-valued field $(\phi(x))_{x\in \bbZ^2}$,
and the energy functional
$$V(\phi)=\gb \sum_{x\sim y}|\phi(x)-\phi(y)|-\sum_{x}\left( h\ind_{\{\phi(x)=0\}}-\infty\ind_{\{\phi(x)<0\}} \right).$$
We prove that for $\gb$ sufficiently large, there exists a decreasing sequence $(h^*_n(\gb))_{n\ge 0}$,
satisfying $\lim_{n\to\infty}h^*_n(\gb)=h_w(\gb),$ and such that:
$(A)$
 The free energy associated with the system
is infinitely differentiable on $\bbR \setminus \left(\{h^*_n\}_{n\ge 1}\cup h_w(\gb)\right)$, and not differentiable on $\{h^*_n\}_{n\ge 1}$.
$(B)$ For each $n\ge 0$ within the interval $(h^*_{n+1},h^*_n)$ (with the convention $h^*_0=\infty$),
there exists a unique translation invariant Gibbs state which is localized around height $n$, while at a point of non-differentiability, 
at least two ergodic Gibbs state coexist. The respective typical heights of these two Gibbs states are $n-1$ and $n$.
The value $h^*_n$ corresponds thus to a first order layering transition from level $n$ to level $n-1$.
These results combined with those obtained in \cite{cf:part1} provide a complete description of the wetting and layering transition for SOS.
\\[10pt]
\\[10pt]
\end{abstract}

\maketitle

\tableofcontents

\newpage
\section{Introduction}

The Solid-On-Solid model (SOS) introduced in \cite{cf:BCF, cf:T} provides a simplified framework 
to study the behavior of two dimensional interfaces in three dimensional systems which display phase coexistence, 
such as  the Ising model with mixed boundary condition \cite{cf:Dob}. 
The  SOS interfaces have a simpler description than the ones that appear in most three dimensional lattice models: they are graphs of functions from a subset of 
$\bbZ^2$ to $\bbZ$ and thus have the simplest possible topological structure. The Gibbs weight associated with each possible interface realization also
have a simple expression.
This makes the SOS model considerably easier to analyze than, say, Ising interfaces. On the other hand, as the simplification performed to obtain 
the SOS description starting from a lattice model with phase coexistence, such as the low temperature Ising  or Potts model, are not too drastic, 
it is believed that results obtained for Solid-On-Solid model may have a predictive value for a large class of interfaces \cite{cf:BCF,cf:T,cf:WGL}.
For this reason, a particular attention has been given to results obtained for SOS concerning the transition from 
to rigid  interfaces at low temperature to rough ones at high temperature  \cite{cf:BM, cf:FS, cf:Swend}, and to the study of layering and wetting transitions
in presence of an interaction with a wall (in a wetting or pre-wetting setup) \cite{cf:ADM, cf:AY, cf:BMF,cf:CM, cf:CM2, cf:chal,cf:DM} .
We refer to the recent review \cite{cf:IV} for a richer introduction to effective interface models as well to the introduction of \cite{cf:part1} 
for additional motivation and references.

Our objective is to give a full  description of the transitions occurring for the wetting problem, 
when an interface interacts with a solid wall which occupies a full half-space.
The problem has been investigated in \cite{cf:chal} where it was shown that the \textit{wetting transition}
occurs for a positive value $h_w(\gb)$
of the intensity of the interaction with the wall: When $h>h_w(\gb)$ the interface is typically localized in a neighborhood of the wall, while for $h<h_w(\gb)$
is is repulsed away from it.
In \cite{cf:AY}, a heuristic analysis of the interface stability yielded the prediction that besides this wetting transition, the system should undergo
countably many \textit{layering transitions} which corresponds to discrete change of the typical height of the interface. 
This analysis also provided a low temperature expansion for the value of first layering critical points.
The first rigorous results concerning these conjectured layering phenomenon were obtained in \cite{cf:ADM}
(results were obtained earlier for the related and more tractable pre-wetting problem, see for instance \cite{cf:CM, cf:DM}): For any given $n\ge 0$, the existence of a 
regime where interfaces are localized at height $n$ was evidenced, analyticity of the free energy and results concerning uniqueness of Gibbs states in that regime 
were also proved.  
The results in \cite{cf:ADM} nonetheless leave some challenging questions open:
\begin{itemize}
 \item [(A)] The existence of a regime with localization at height $n$ is only proved under the assumption that $\gb\ge c \log n$ (with our notation) for some constant $c>0$  
 (cf.\ \cite[Remark (4) pp 528]{cf:ADM}). This limitation obstructs the understanding of how the layering transition accumulate on the right of $h_w(\gb)$ when $n$ tends to infinity.
 \item [(B)] The layering transitions corresponding to the changes of typical height, say from $n$ to $n-1$ cannot be analyzed, 
 as the intervals on which localization is 
 shown to occur are not adjacent. This is because the perturbative approach used in \cite{cf:ADM} does not allow to come close to the layering critical points.
\end{itemize}
The present paper overcomes these limitations and proves that for $\gb$ sufficiently large 
the free energy is infinitely differentiable everywhere except on a countable set which 
corresponds to the layering critical points. On this set, the first derivative of the free energy is shown to be  discontinuous. 
The existence of Gibbs states when the free energy is also proved, together with uniqueness  on intervals 
where the free energy is differentiable, and non-uniqueness at points of non differentiability.
Combined with the results obtained in \cite{cf:part1} concerning  the value critical point $h_w(\gb)$ 
and the sharp asymptotics for the free energy, this yields a complete picture of the systems behavior.

\section{Model and results}

\subsection{The Solid on Solid Model on $\bbZ^d$}

Consider $\gL$ a finite subset of  $\bbZ^2$ (equipped with its usual lattice structure) and let $\partial \gL$ denote its external boundary 
$$\partial \gL:= \{ x \in \bbZ^d \setminus \gL \ : \ \exists y \in \gL, \ x \sim y \}.$$
Given $\psi\in \bbZ^\bbZ$, we define the Hamiltonian for SOS in the domain $\gL$ with boundary condition $\psi$ on the set $\gO_{\gL}:= \bbZ^\gL$ by
\begin{equation}\label{defhamil}
\cH^{\psi}_\gL(\phi):= \frac{1}{2}\sumtwo{x, y \in \gL}{ x \sim y} |\phi(x)-\phi(y)|
+\sumtwo{x\in \gL, y\in \partial\gL}{x\sim y} |\phi(x)-\psi(y)|, \quad \forall \phi \in \gO_{\gL}.
\end{equation}
Given $\gb>0$, we define
the SOS measure with boundary condition $\psi$, $\bP^{\psi}_{\gL,\gb}$ on $\gO_{\gL}$ by
\begin{equation}\label{defSOS}
\bP^{\psi}_{\gL,\gb}(\phi):= \frac{1}{\mathcal Z^{\psi}_{\gL,\gb}}e^{-\gb\cH^{\psi}_\gL(\phi)} \quad \text{ where } 
\mathcal Z^{\psi}_{\gL,\gb}:=\sum_{\phi\in \gO_{\gL}}  e^{-\gb\cH^{\psi}_\gL(\phi)}.
\end{equation}
For most purposes, we only have to consider the constant boundary conditions $\psi \equiv n$ for $n\ge 0$.
In that case we simply write $\bP^{n}_{\gL,\gb}$ and $\cZ^{n}_{\gL,\gb}$.
We drop the superscript $n$ in the notation in the special case $n=0$.
Note that by translation invariance $\mathcal Z^n_{\gL,\gb}=\mathcal Z_{\gL,\gb}$ does not depend on $n$.
We also define the \textit{free energy} (sometimes also referred to as \textit{pressure}) for the SOS model by
\begin{equation}
\tf(\gb):=\limtwo{|\gL|\to \infty}{|\partial\gL|/|\gL|\to 0} \frac{1}{|\gL|}\log  \cZ_{\gL,\gb},
\end{equation}
where the limit can be taken over any sequence of finite sets $(\gL_N)_{N\ge 0}$ such such ratio between the cardinality of 
$\gL_N$ and that of its boundary vanishes. A justification of the existence of the limit is given in the introduction of \cite{cf:part1}. 
We used $|\cdot|$ to denote the cardinality of a set, and we keep this notation in the remainder of the paper.

\medskip

When $\gb$ is sufficiently large, it is known  \cite[Theorem 2]{cf:BM}  that $\bP_{\gL,\gb}$ converges (in the sense of finite dimensional marginal)
to an infinite volume measure $\bP_{\gb}$ or \textit{Gibbs state} (see Definition \ref{defgs}). We introduce a quantitative version of the statement which requires 
the introduction of  some classic terminology.

\medskip

We say that a function $f: \ \gO_{\infty}:=(\bbZ)^{\bbZ^2}\to \bbR$ is local if there exists $(x_1,\dots,x_k)$  and $\tilde f:  (\bbZ)^{k}\to \bbR$ 
such that $f(\phi)=\tilde f(\phi(x_1),\dots,\phi(x_k)).$ The minimal  choice (for the inclusion) 
for the set of indices $\{x_1,\dots, x_k\}$ is called the support of $f$ ($\Supp(f)$).
With some abuse of notation, whenever $\gL$ contains the support of $f$,  we extend $f$ to $\gO_{\gL}$ in the obvious way.
An event is called local if its indicator function is a local function.

\medskip

Given $\bP_{\gL}$ a sequence a measure on $\gO_{\gL}$ and $\bP$ a measure on $\gO_{\infty}$.
We say that $\bP_{\gL}$ \textit{converges locally} to $\bP$ when $\gL$ exhausts $\bbZ^2$ if for any sequence  $\gL_N$ exhausting $\bbZ^2$, 
and any local function $f$ we have
$$\lim_{N\to \infty} \bP_{\gL_N}[f(\phi)]=\bP[f(\phi)].$$
 
\noindent For $A$ and $B$ two finite subsets of $\bbZ^2$ we set 
\begin{equation}\label{disset}
d(A,B):=\min_{x\in A, y\in B} |x-y|,
\end{equation}
where $|\ \cdot \ |$ denote the $\ell_1$ distance. The proofs in \cite{cf:BM} imply that  $\bP_{\gL,\gb}$ 
converges exponentially fast in some sense to some measure on $\gO_{\infty}$. The statement below can also be proved using the 
techniques introduced in Section \ref{clusexp} (see Remark \ref{teorema}).

\begin{theorema}\label{infinitevol}
There exists $\gb_0>1$ and $c$ such that 
for any $\gb>\gb_0$, 
there exists a measure $\bP_{\gb}$ defined  on $\gO_{\infty}$ such that that for every local function 
$f: \gO_{\infty}\to [0,1]$ with  $\Supp(f)=A$, and every $\gL$ which contains $A$
\begin{equation}\label{decayz}
| \bE_{\gL,\gb}[f(\phi)]-\bE_{\gb}[f(\phi)] |\le C_{\gb} |A| e^{-c \gb d(\partial \gL, A)}.
\end{equation}
\end{theorema}

\subsection{The wetting problem for the SOS model}

For $\phi \in \gO_{\gL}$ and $A\subset \bbZ$ (or $\bbR$), we set 
$$\phi^{-1}(A):=\{ x\in \gL \ : \ \phi(x)\in A\}.$$
We write also $\phi^{-1}A$ when more convenient.
Using the notation  $\bbZ_+:=\bbZ\cap[0,\infty)$, we define
\begin{equation}\label{positivity}
 \gO^+_{\gL}:=(\bbZ_+)^{\gL}=\{\phi \in \gO_{\gL} \ : \  x\in \gL,\ \phi(x)\ge 0\}.
\end{equation}
This convention of adding the superscript $+$ to indicate a restriction to the set of positive functions 
is used in other contexts throughout the paper.

\medskip

Given $h\in \bbR$ we consider $\bP^{\psi,h}_{\gL,\gb}$ which is a modification of $\bP^{\psi}_{\gL,\gb}$ 
where the interface $\phi$ is constrained to remain positive 
and gets an energetic reward $h$ for each contact with $0$.
It is defined as follows
\begin{equation}\label{def}
\bP^{\psi,h}_{\gL,\gb}(\phi):= \frac{1}{\mathcal Z^{\psi,h}_{\gL,\gb}}e^{-\gb\cH^\psi_\gL(\phi)+ h|\phi^{-1}\{0\}|}
 \quad \text{where } \mathcal Z^{\psi,h}_{\gL,\gb}:=\sum_{\phi\in \gO^+_{\gL}}  e^{-\gb\cH^\psi_\gL(\phi)+h|\phi^{-1}\{0\}|}.
\end{equation}
In this case also, we replace $\psi$ by $n$ in the notation for the special case $\psi\equiv n$.
The aim of our study is to investigate the localization transition in $h$ for $\bP^{n,h}_{\gL,\gb}$ which appears in the limit when $\gL$ exhausts $\bbZ^2$.
A key quantity to study the phenomenon is the corresponding free energy
\begin{equation}\label{nodepends}
\tf(\gb,h):=\limtwo{|\gL|\to \infty}{|\partial \gL|/|\gL|\to 0} \frac{1}{|\gL|}\log \mathcal Z^{n,h}_{\gL,\gb}.
\end{equation}
The reader can check that as a consequence of the inequality 
$$|\cH^n_\gL(\phi)- \cH^m_\gL(\phi)|\le 4 |m-n| |\partial \gL|$$ the quantity
$\tf(\gb,h)$ indeed does not depend on $n$.

\medskip

To clarify notation, in the remainder of the paper, we often consider the limit along the sequence 
$\gL_N:=\lint 1, N\rint^2$ (using the notation $\lint a,b \rint=[a,b]\cap \bbZ$ ). 
We write $\cZ_{N,\gb}$ for $\cZ_{\gL_N,\gb}$ and adopt a similar convention for other quantities.

\medskip

The function $h\mapsto \tf(\gb,h)$ is non-decreasing and convex in $h$ (as a limit of non-decreasing convex function).
At points where $\tf(\gb,h)$ is differentiable, convexity allows to exchange the positions of limit and derivative
thus  $\partial_h \tf(\gb,h)$ corresponds to the 
asymptotic contact fraction. Thus for every $n\in \bbN$ we have
\begin{equation}\label{difencial}
 \partial_h \tf(\gb,h)=\lim_{N\to \infty} \frac{1}{N^2} \bE^{n,h}_{N,\gb}[|\phi^{-1}(0)|],
\end{equation}
wherever $\partial_h\tf$ is defined (by convexity this is everywhere except possibly on a countable set).

\medskip

We define 
$h_w(\gb)$ to be the value of $h$ which marks the wetting transition between a localized phase (where the asymptotic contact fraction is positive) 
and a delocalized phase 
\begin{equation}\begin{split}
 h_w(\gb)&:= \inf\{ h\in \bbR \ : \ \partial_h \tf(\gb,h) \text{ exists and is positive }\}\\
 &= \sup\{ h\in \bbR : \tf(\gb,h)= \tf(\gb)\}.
\end{split}\end{equation}

\subsection{The asymptotic behavior for the free energy}

In previous work \cite{cf:part1}, we established the value of $h_w(\gb)$ answering a question left open since the pioneering work of Chalker \cite{cf:chal}, and 
we were able to describe the asymptotic behavior of $\tf(\gb,h)$ close to the critical point. To state this result we need to introduce a few quantities.
Letting ${\bf 0}$ and ${\bf 1}$ denote the vertices $(0,0)$ and $(1,0)$ respectively, we define, for $\gb>\gb_0$
\begin{equation}\label{singledouble}
\begin{split}
\alpha_1(\gb)&:= \lim_{n\to \infty} e^{4\gb}  \bP_{\gb} \left[ \phi({\bf 0})\ge  n\right],\\
\alpha_2(\gb)&:= \lim_{n\to \infty} e^{6\gb}  \bP_{\gb} \left[  \min( \phi({\bf 0}), \phi({\bf 1})) \ge n\right].
\end{split}
\end{equation}
For a proof of the existence of these quantity, we refer to \cite[Proposition 4.6]{cf:part1}.
We also set $J:=e^{-2\gb}$ and for $u\in \bbR$
\begin{equation}
 \bar \tf(\gb,u)=\tf\left(\gb,\log \left( \frac {e^{4\gb}}{e^{4\gb}-1}\right)+u \right)-\tf(\gb).
\end{equation}

\begin{theorema}\label{oldmain}
If $\gb>\gb_0$ (of Theorem \ref{infinitevol}), we have 
$$h_w(\gb)=\log\left(\frac{e^{4\gb}}{e^{4\gb}-1}\right).$$
Furthermore
\begin{equation}\label{lequivdelamor}
\bar \tf(\gb,u)\stackrel{u\to 0+}{\sim} F(\gb,u),
\end{equation}
where  
\begin{equation}\label{uds}
F(\gb,u):=\max_{n\in \bbZ_{+}}\left( \alpha_1 J^{2n}u-  \frac{2 \alpha_2(J^3-J^4)}{1-J^3} J^{3n} \right).
\end{equation}
\end{theorema}
\noindent Note that the function $F(\gb,u)$ is piecewise affine on $\bbR_+$, and present angular points at 
\begin{equation}\label{defun}
 u=u_n:=\frac{2\alpha_2}{\alpha_1(1+J)}J^{n+2},
\end{equation}
for $n\ge 1$.
While Theorem \ref{oldmain} does not imply the convergence of $\partial_u \bar \tf(\gb,u)$ and thus, the presence 
of angular points on the free energy curves, convexity implies that for large values of $n$ the contact fraction changes abruptly
around $u_n$.

\medskip

A possible explanation for this phenomenon which is corroborated by the proof of Theorem \ref{oldmain} is that the typical behavior of $\phi$ changes radically 
around $u_n$: when $u\le u_n+o(J^{n})$ the surface $\phi$ tends to localize at height $n$, meaning that $\phi(x)=n$ for a majority of points and that 
connected components  of the set $\{ x \ : \ \phi(x)\ne n\}$ are all  of small diameters,
while when $u\ge u_n+o(J^{n})$ the typical height should be $n-1$.

\medskip

This indicates that there should exist a value $u^*_n$ which delimits a phase transition between these two kinds of behavior. Moreover it should satisfy,  asymptotically
for large values of $n$,
$u^*_n=u_n+o(J^{n})$.
The change of behavior around $u^*_n$ should provoke a discontinuity in the contact fraction, 
so that these sequences of phase transition should be manifested by discontinuities
for $\partial_u \bar \tf(\gb,u)$. This prediction can be interpreted as a refined version 
of the conjecture presented in \cite[Statement p 228]{cf:ADM}. Earlier version of the same conjecture is found in \cite{cf:AY} and \cite[Section 4.3]{cf:BMF}.

\subsection{Main results}

In the present paper, we bring the above stated  conjecture on a rigorous ground by showing the existence of an infinite sequence of point of discontinuity for 
$\partial_u \bar \tf(\gb,u)$. We complement this result by relevant information concerning the regularity of $\tf(\gb,u)$ between these transition points, 
and a statement concerning uniqueness of Gibbs states.

\begin{theorem}\label{main}
For $\gb\ge \gb_1$ sufficiently large,
there exists a decreasing  sequence   $(u^*_n(\gb))_{n\ge 1}$, which satisfies 
$$\frac{1}{200} J^{n+2} \le  u^*_n \le 200 J^{n+2}$$
and  
 \begin{equation}\label{asimptic}
\lim_{n\to \infty} e^{2\gb n} u^*_n=\frac{2\alpha_2J^{2}}{\alpha_1(1+J)}.
\end{equation}
which is such that
\begin{itemize}
 \item [(i)] The function $u\mapsto \bar \tf(\gb,u)$ is infinitely differentiable on $(u^*_{n+1}(\gb),u^*_n(\gb))$   for any $n\ge 1$ and also on $(u^*_{1}(\gb),\infty).$
 \item [(ii)] For any $n\ge 1$, $\bar \tf(\gb,u)$ is not differentiable at $u^*_n(\gb)$, meaning that that   the left and right derivative at  $u^*_n$ do not coincide
 \begin{equation}
       \partial^-_u \bar \tf(\gb,u^*_n)<\partial^+_u \bar \tf(\gb,u^*_n).
 \end{equation}

\end{itemize}
\end{theorem}

\begin{rem}
We believe that the free energy is in fact analytic in $u$ on the domain where it is differentiable. While such a statement could in principle be directly deduced from the convergence of the cluster expansion,
it would require to be able to obtain a convergence result for complex values of $u$, more precisely for each $n$ one should prove convergence of the expansion
on an open subset of $\bbC$ which contains the real interval $(u^*_{n+1},u^*_{n})$.
\end{rem}

\begin{rem}
Note that the main result in \cite{cf:ADM} includes a statement about analyticity. While this is not explicitly stated in the proofs, it appears that 
the cluster expansion considered in \cite{cf:ADM} also converges when the parameter $u$ considered in  \cite[Equation (1.15)]{cf:ADM} is allowed to have a small imaginary part \cite{cf:Al}.
It seems plausible that with some (significant) efforts, our proof of Proposition \ref{truta} could be adapted to handle also small imaginary perturbation
of $u$, which would yield analyticity of $\bar \tf(\gb,u)$ in the intervals of the type $[200 J^{n+2}, \frac{1}{200} J^{n+1}]$, $n\ge 1$ and on $[200 J^2,\infty)$.
However
monotonicity plays a too central role in Proposition \ref{mejor} to extend this kind of argument.
For this reason the proof of analyticity in the neighborhood of $u^*_n$ appears like a more challenging task.
\end{rem}

To state our second  result about convergence for the measure $\bP^{n,h}_{\gL,\gb}$ we need to recall some terminology.

\begin{definition}\label{defgs}
An infinite volume measure or \textit{Gibbs state} for parameter $(\gb,h)$ is a measure $\nu_{\gb,h}$ on $(\bbZ_+)^{\bbZ^2}$
such that for any finite $\gL\subset \bbZ^2$, we have for $\nu_{\gb,h}$-almost all $\psi$,
\begin{equation}\label{DLR}
\nu_{\gb,h}\left[ \phi\restrict_{\gL}\in \cdot \ | \  \phi\restrict_{\gL^{\cc}}=\psi\restrict_{\gL^{\cc}} \right]=  \bP^{\psi,h}_{\gL,\gb}\left[ \phi \in \cdot \right].
\end{equation}
\end{definition}

It is not difficult to check that 
the relation \eqref{DLR}, often referred to as the Dobrushin-Lanford-Ruelle (DLR) Equation, is valid if one replaces $\nu_{\gb,h}$
by a measure  $\bP^{\Psi,h}_{\gb,\gL'}$ defined on a domain  $\gL'$ which includes $\gL$ and with arbitrary boundary condition $\Psi$.
As a consequence, the  measures obtained as local limits of  $\bP^{\psi,h}_{N,\gb}$ for $\psi\in \bbZ_+^{\bbZ^2}$ are Gibbs state. 

\medskip

For technical purpose we define $\star$-connectivity, as the connectivity associated with the network $\bbZ^d$ were diagonal edges of the type
$\{x,x+(1,1)\}$ have been added (but not the other diagonals).

\begin{definition}
A Gibbs state $\nu_{\gb,h}$ is said to :
\begin{itemize}
 \item [(i)] Be \textit{translation invariant} if under $\nu_{\gb,h}$, the distribution of $\phi$ and $\theta_z(\phi):= (\phi(z+x))_{x\in \bbZ^2}$
are the same.
 \item [(ii)] Have \textit{finite mean} if for all $x\in \bbZ^2$
\begin{equation}\label{finitemean}
 \nu_{\gb,h}\left[ \phi(x) \right]<\infty.
\end{equation}
 \item [(iii)] \textit{Percolate at level} $n$ if $\nu_{\gb,h}$ almost surely, $\phi^{-1}(n)$ has unique infinite connected component in $\bbZ^2$
and all connected component of $\phi^{-1}[n+1,\infty)$ and $\phi^{-1}(-\infty,n-1]$ are finite.
\end{itemize}
\end{definition}

Setting $h^*_n=h_w(\gb)+u^*_n$,
our second result  essentially  claims that translation invariant Gibbs states are unique for $h\in (h_w(\gb),\infty)\setminus \{h^*_n\}_{n\ge 1}$
and that multiple  translation invariant Gibbs  states coexist at the layering points $(h^*_n)_{n\ge 1}$.

\begin{theorem}\label{Gibbs}
For $\gb$ sufficiently large, the following holds true.
\begin{itemize}
\item[(i)]For  $h\le h_w(\gb)$, there exists  no Gibbs state for $(\gb,h)$. 
\item[(ii)] When $n\ge 0$, $h\in (h^*_{n-1},h^*_{n})$ then there exists a unique finite mean translation invariant Gibbs state
  which we call $\bP^{n,h}_{\gb}$.
  Moreover $\bP^{n,h}_{\gb}$ percolate at level $n$.

  \item [(iii)] When $h=h^*_n$, $n\ge 1$, then there exists several finite mean translation invariant Gibbs states.
In particular we can identify two extremal states $\bP^{n-1,h^*_n}_{\gb}$ and $\bP^{n,h^*_n}_{\gb}$ which satisfy:
  \begin{itemize}
  \item[(A)] 
    \begin{equation}\begin{split}\label{variousgibbs}
     \bP^{n,h^*_n}_{\gb}[\phi(x)=0]&= \partial^-_h\tf(\gb,h^*_n), \\
     \bP^{n-1,h^*_n}_{\gb}[\phi(x)=0]&= \partial^+_h\tf(\gb,h^*_n).
      \end{split}
  \end{equation}
  \item[(B)]  $\bP^{n-1,h^*_n}_{\gb}$ and $\bP^{n,h^*_n}_{\gb}$ respectively percolate at level $n-1$ and $n$.
  \item[(C)] We have $\bP^{n-1,h^*_n}_{\gb}\preccurlyeq  \bP^{n,h^*_n}_{\gb}$ . Any other finite mean translation invariant Gibbs state $\nu$
  for parameters $(\gb,h^*_n)$  satisfies
  \begin{equation}\label{ssanduba}
 \bP^{n-1,h^*_n}_{\gb}\preccurlyeq \nu \preccurlyeq \bP^{n,h^*_n}_{\gb}.
  \end{equation}
\end{itemize}

\end{itemize}

 \end{theorem}

 \begin{rem}
  We believe that there are no infinite mean translation invariant Gibbs state for $h>h_w(\gb)$ and that the finite mean assumption is present only for technical reasons.
  We would also tend to believe that in analogy with low temperature two dimensional Ising model  (see \cite{cf:A, cf:CoV, cf:H} for results and proofs)
  $\bP^{n-1,h^*_n}_{\gb}$ and $\bP^{n,h^*_n}_{\gb}$ are  in fact the only
  ergodic Gibbs states when $h=h^*_n$, but proving such a statement is out of the scope of this paper.
\end{rem}

Finally we conclude the exposition with a result showing that our Gibbs states exhibit exponential decay of correlation.

\begin{proposition}\label{Gibbset}
For $\gb$ sufficiently large, there exists constants $c$ and $C$ such that
for every $n\ge 0$ and any $h\in [h^*_{n+1},h^*_{n}],$ we have for any pair of local functions $f:\gO_{\infty}\to [0,1]$ and $g:\gO_{\infty}\to [0,1]$ with respective supports
$A$ and $B$ we have
\begin{equation}
 \left|\bE^{n,h}_{\gb}\left[ f(\phi)g(\phi) \right]-  \bE^{n,h}_{\gb}\left[ f(\phi)\right]\bE^{n,h}_{\gb}\left[g(\phi) \right] \right|\le C
 |A| e^{-c \gb d(A,B)},
\end{equation}
and thus in particular using the notation $\delta_x:=\ind_{\{\phi(x)=0\}}$.
\begin{equation}
| \bE^{n,h}_{\gb}\left[ \delta_x\delta_y \right]-\bE^{n,h}_{\gb}\left[ \delta_x\right]\bE^{n,h}_{\gb}\left[ \delta_y \right]|\le C e^{-c \gb |x-y|}.
\end{equation}

\end{proposition}

\subsection{Organization of the proof}

All the results exposed in the previous  section are going to be derived as consequences of the 
convergence of a cluster expansion associated with a certain contour representation of the partition functions (see the introduction of \cite{cf:BFP} and mentioned references 
for a review of cluster expansion techniques). 
To prove the convergence, we need to obtain very fine estimates on finite size partition functions: these are obtained by combining various ingredients such as
asymptotic properties of the SOS model (Proposition \ref{rouxrou}),

Therefore our first task is to introduce the necessary framework for the exposition of this result.
This is the purpose of Section \ref{tools}, in which we introduce various technical tools, including contour representations, 
cluster expansion methods, and FKG inequalities.

\medskip

In Section \ref{zorga}, we introduce the main technical result of the paper Theorem \ref{converteo}, which implies the convergence of a cluster expansion associated with the measures 
$\bP^{n,h}_{\gL,\gb}$
and give the main steps of its proof.
We also explain how Theorem \ref{main} can be deduced from this convergence result.

\medskip

In Section \ref{ponverteo}, which is the technical core of the paper, we perform the proof of Theorem \ref{converteo} in full details.
For better readability, the proof of the more technical estimates presented in Section \ref{ponverteo} are performed separately, in Section \ref{secrestrict}.

\medskip
In Section \ref{dalast}, we explore the consequences of Theorem \ref{converteo} on the measure $\bP^h_{\gL,\gb}$  and prove in particular Proposition \ref{Gibbset}.
Finally in Section \ref{daverylast} we prove the remaining statements of Theorem \ref{Gibbs}.

\section{Technical preliminaries}\label{tools}

\subsection{Contour representation}\label{contour}

We recall briefly how to describe a function $\phi\in \gO_{\gL}$ using only its level lines.
The formalism of this section is identical to the one used in \cite{cf:part1}, and close to the one displayed in e.g.\
\cite{cf:ADM, cf:CM, cf:PLMST}.

\medskip

We let $(\bbZ^2)^*$ denote the dual lattice of $\bbZ^2$ (dual edges cross that of $\bbZ^2$ orthogonally in their midpoints).
Two adjacent edges  $(\bbZ^2)^*$ meeting at $x^*$ of are said to be \textit{linked} if they both 
lie on the same side of the line making an angle $\pi/4$ with the horizontal 
and passing through $x$. 
(see Figure \ref{linked}).

We define a \textit{contour sequence} to be  a finite sequence $(e_1,\dots,e_n)$ of distinct edges of $(\bbZ^2)^*$ which satisfies:
\begin{itemize}
 \item [(i)] For any $i=\lint 1,n-1\rint$, $e_i$ and $e_{i+1}$ have a common end point $(\bbZ^2)^*$, $e_1$ and $e_{|\gamma|}$ also have a common end point.
 \item [(ii)] If for $i\ne j$, if $e_i$, $e_{i+1}$, $e_j$ and $e_{j+1}$ meet at a common end point then  $e_i$, $e_{i+1}$ are linked and so are $e_j$ and $e_{j+1}$ (with the convention that $n+1=1$).
\end{itemize}
A \textit{geometric contour} $\tilde \gamma:=\{e_1,\dots,e_{|\tilde \gamma|}\}$ is a set of edges that forms a contour sequence when displayed in the right order.
The cardinality $|\tilde \gamma|$ of $\tilde\gamma$ is called the length of the contour. 
A \textit{signed contour} or simply \textit{contour} $\gamma=(\tilde \gamma,\gep)$ 
is a pair composed of a geometric contour and a sign $\gep\in \{+1,-1\}$. We let $\gep(\gamma)$ denote the sign associated with a contour $\gamma$, 
while with a small abuse of notation, 
$\tilde \gamma$ will be used for the geometric contour associated to $\gamma$ when needed.
For $x^*\in (\bbZ^2)^*$ we write $x^*\in \gamma$ or $x^*\in \tilde \gamma$ when the point $x^*$ is visited by one edge of the geometric contour.

 \begin{figure}[ht]
\begin{center}
\leavevmode
\epsfysize = 3 cm
\epsfbox{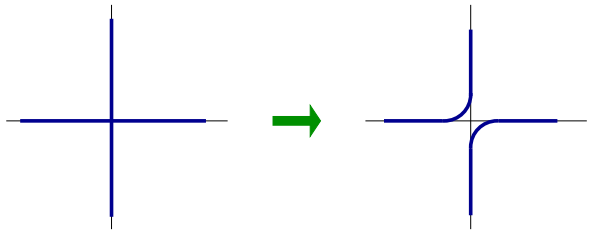}
\end{center}
\caption{\label{linked} 
The rule for splitting a four edges meeting at one points into two pairs of linked edges.
To obtain the set of contours that separates $\{x \ : \ \phi(x)\ge h\}$ from   $\{x \ : \ \phi(x)< h\}$ for $h\in \bbZ$,
we draw all dual edges separating two sites $x$, $y$ such that $\phi(x)\ge h>\phi(y)$ and apply the above graphic rule for every dual vertex where four edges meet.
When several sets of level lines include the same contour, it corresponds to a cylinder of intensity $2$ or more for $\phi$.}
\end{figure}

We let $\bar \gamma$ denote the set of vertices of $\bbZ^2$ enclosed by $\tilde \gamma$.
We refer to $\bar \gamma$ as the \textsl{interior of} $\gamma$ and say that $|\bar \gamma|$ is the volume enclosed in the contour $\gamma$.
We let $\Delta_{\gamma}$, the neighborhood of $\gamma$, be the set of vertices of $\bbZ^2$ located either at a (Euclidean) distance $1/2$ from 
$\tilde\gamma$ 
(when considered as a subset of $\bbR^2$) 
or at a distance $1/\sqrt{2}$ from the meeting point of two non-linked edges.
We split the $\Delta_{\gamma}$ into two disjoint sets, the internal and the external neighborhoods of $\gamma$ 
(see Figure \ref{compa})
$$\Delta^-_{\gamma}:=\Delta_{\gamma}\cap \bar \gamma \quad \text{ and } \quad  \Delta^+_{\gamma}:=\Delta_{\gamma}\cap \bar \gamma^{\cc}.$$

\medskip

Given a finite set $\gL\subset \bbZ^2$ 
a contour $\gamma$ is said to be in $\gL$ is if $\bar \gamma\subset \gL$.
We let $\cC$ denote the set of contours in $\bbZ^2$ and $\cC_{\gL}$ that of contours in $\gL$.

\medskip

Given $\phi\in \gO_{\gL}$, we say that $\gamma \in \cC_{\gL}$  is a contour
for $\phi$ with boundary condition $n$, if there exists $k\ge 1$ such that 
\begin{equation}\label{cont+}
\min_{x\in \Delta^-_{\gamma}} \phi(x)= \max_{x\in \Delta^+_{\gamma}} \phi(x)+k\gep(\gamma).
\end{equation}
where in the above equation by convention we consider  that
$$\phi(x)=n \quad \text{ if } \quad x\in \gL^{\cc}.$$

The quantity $k$ appearing in \eqref{cont+} is called the \textit{intensity} of the contour and the
triplet $(\gamma,k)=(\tilde \gamma, \gep(\gamma), k)$
with $\gamma \in \cC$ and $k\in \bbN$ an intensity, is called a \textit{cylinder}. We say that  $(\gamma,k)$
is a cylinder for $\phi$ (with boundary condition $n$) if $\gamma$ is a contour of intensity $k$.
The cylinder function associated to $(\gamma,k)$ is defined on $\bbZ^2$ by
\begin{equation}\label{cylfunc}
\varphi_{(\gamma,k)}(x)=\gep(\gamma) k \ind_{\bar \gamma}(x).
\end{equation}
We use $\hat \gamma$ to denote a generic cylinder associated with the contour $\gamma$ (we use the notation $k(\hat \gamma)$ to denote its intensity).
We let $\hat \Upsilon_n(\phi)$ denote the set of cylinders for $\phi$ with boundary condition $n$ and $\Upsilon_n(\phi)$ the associated set of contours.\\

\medskip

We say that $\gL$ is a \textit{simply connected} subset of $\bbZ^2$, if it can be expressed as the interior of a contour,
that is, if
\begin{equation} \label{connectedness}
 \exists \gamma_{\gL}\in \cC, \quad \bar\gamma_{\gL}=\gL.
\end{equation}
Note that, when $\gL$ is simply connected, 
an element
$\phi\in \gO_{\gL}$ is uniquely characterized by its cylinders. More precisely, we have  
\begin{equation}\label{cylinder}
\forall x\in \gL, \quad \phi(x):= n+\sum_{\hat \gamma \in \hat \Upsilon_n(\phi)}  \varphi_{\hat \gamma}(x).
\end{equation}
Furthermore, the reader can check that
\begin{equation}\label{express}
 \cH^n_{\gL}(\phi)=\sum_{\hat \gamma \in \hat \Upsilon_n(\phi)} k(\hat \gamma)|\tilde \gamma|.
\end{equation}
Of course not every set of cylinder is of the form $\hat \Upsilon_n(\phi)$ and we must introduce a notion of compatibility
which characterizes the ``right" sets of cylinder.

\medskip

Two cylinders $\hat \gamma$ and  $\hat \gamma'$ are said to be \textsl{compatible} if they are
cylinders for the function $\varphi_{\hat \gamma}+\varphi_{\hat \gamma'}$.
This is equivalent to the three following conditions  being satisfied :
 (see Figure \ref{compa})
\begin{itemize}
 \item [(i)] $\tilde \gamma\ne \tilde\gamma'$ and $\bar \gamma \cap  \bar \gamma' \in\{\emptyset,\bar \gamma, \bar \gamma'\}$.
  \item [(ii)]  If $\gep= \gep'$ and $\bar\gamma\cap \bar \gamma'= \emptyset$, then
 then $\bar \gamma'\cap \Delta^+_{\gamma}=\emptyset$ .
 \item [(iii)] If $\gep\ne \gep'$ and $\bar \gamma'\subset \bar \gamma$ (resp.\ $\bar \gamma\subset \bar \gamma'$) then 
 $\bar \gamma'\cap \Delta^-_{\gamma}=\emptyset$  (resp. $\bar \gamma\cap \Delta^-_{\gamma'}=\emptyset$).
\end{itemize}
This first condition simply states that compatible contours do not cross each-other.
The conditions $\bar \gamma'\cap \Delta^+_{\gamma}=\emptyset$ and  $\bar \gamma'\cap \Delta^-_{\gamma}=\emptyset$  in  $(ii)$ and $(iii)$  
can be reformulated as: $\tilde \gamma$ and $\tilde \gamma'$ do not share edges, and if both $\tilde \gamma$ and $\tilde \gamma'$ possess
two edges adjacent to one vertex 
$x^*\in (\bbZ^2)^*$ 
 then the two edges in $\gamma$ are linked and so are those in $\gamma'$.
 
 \medskip
 
  \begin{figure}[ht]
\begin{center}
\leavevmode
\epsfxsize =5 cm
\epsfbox{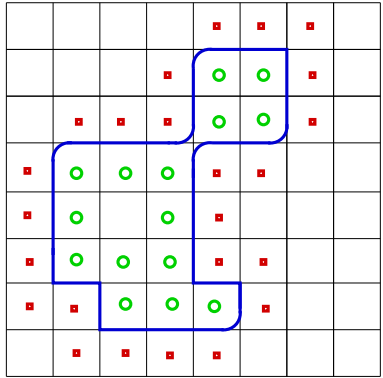}
\end{center}
\caption{\label{compa} 
A contour $\gamma$ represented with its internal (circles) and external (squares) neighborhood.
To be compatible with $\gamma$, a contour $\gamma'$ of the same sign such that $\bar \gamma'\cap \bar \gamma= \emptyset$ cannot enclose any squares.
A compatible contour of opposite sign enclosed in $\gamma$ (such that $\bar \gamma'\subset \bar \gamma$) cannot enclose any circles.
}
\end{figure}

Note that the compatibility of two cylinders does not depend on their respective intensity, so that the notion can naturally be extended to signed contours:
The contours  $\gamma$ and $\gamma'$ are said to be compatible (we write $\gamma \mid \gamma'$) if the cylinders $(\gamma,1)$ and $(\gamma',1)$ are.
Two distinct non-compatible contours are said to be \textit{connected} (we write $\gamma \perp \gamma'$).

\medskip

If $\bC_1$ and $\bC_2$ are two finite collections of contours (compatible or not) we say that $\bC_1$ is \textit{compatible with} $\bC_2$ and write 
$\bC_1 \mid \bC_2$ if
\begin{equation}\label{compatiblesets}
 \forall \gamma_1 \in \bC_1, \forall \gamma_2\in \bC_2, \  \gamma_1 \mid \gamma_2. 
\end{equation}
If \eqref{compatiblesets} does not holds we say that $\bC_1$ and $\bC_2$ are connected and  write  $\bC_1 \perp \bC_2$.
For a contour $\gamma$ and a collection $\bC$ with use the notation $\gamma \perp \bC$ and $\gamma \mid \bC$ for 
 $\{\gamma\} \perp \bC$ and $\{\gamma\} \mid \bC$

\medskip

A (finite or countable) collection of cylinders (or of signed contours) 
is said to be a compatible collection if its elements are pairwise compatible (see Figure \ref{grid}).
The reader can check by inspection that the following result holds. In 
particular it establishes that the set of compatible collections of cylinders is in bijection with $\gO_{\gL}$
(simple connectivity is required to avoid having level lines enclosing holes).

\begin{lemma}\label{toott}
If $\gL$ is simply connected, then for any $\phi\in \gO_{\gL}$, $\hat \Upsilon_n(\phi)$ is a compatible collection of cylinders
and reciprocally, if $\hat \gG \subset \hat \cC_{\gL}$ is a compatible collection of cylinder in $\gL$
then its elements are the cylinders of the function 
$\sum_{\hat \gamma\in \hat \gG} \varphi_{\hat \gamma}.$
\end{lemma}

Using \eqref{express} and the contour representation above, we can rewrite the partition function $\cZ_{\gL,\gb}$ in a new form.
We let $\cK(\gL)$ and $\hat \cK(\gL)$ denote the set of compatible collections of contour and cylinders in $\gL$.
We have 
\begin{equation}
 \cZ_{\gL,\gb}= \sum_{\hat \gG \in \hat \cK(\gL)} \prod_{\hat \gamma \in \hat \gG} e^{-k(\hat \gamma)\gb|\tilde \gamma|}.
\end{equation}
Summing over all the possible intensities, we obtain
\begin{equation}\label{contourdecomp}
 \cZ_{\gL,\gb}=\sum_{\gG \in  \cK(\gL)} \prod_{\gamma \in  \gG} \frac{1}{e^{\gb|\tilde \gamma|}-1}.
\end{equation}
This last representation of the partition function is suitable to apply the cluster expansion techniques which we introduce in the next section.

 \begin{figure}[ht]
\begin{center}
\leavevmode
\epsfysize = 7 cm
\epsfbox{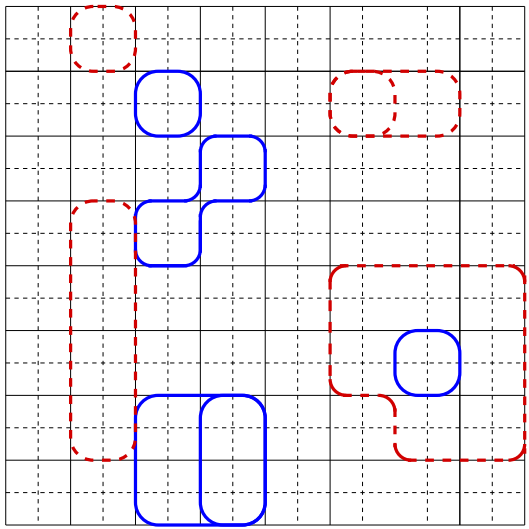}
\end{center}
\caption{\label{grid} 
A compatible collection of contour on the dual lattice (the primal lattice is displayed is dotted lines).
Contours of different signs are displayed in different colors (red-dotted/blue-solid).
The primal lattice is represented in dotted line. 
}
\end{figure}

\medskip

We end this section by introducing a notion which will be of fundamental use in our proofs, and a few notation.
Given $\gG$ a compatible collection of contour and $\gamma \in \gG$, we say that $\gamma$ is an \textit{external contour} in $\gG$
if $\bar \gamma$ is maximal in $\gG$ for the inclusion, that 
is 
\begin{equation}\label{defexternal}
 \forall \gamma'\in \gG, \quad \bar \gamma' \subset \bar\gamma \text{ or } \bar \gamma' \cap \bar \gamma= \emptyset.
\end{equation}

We say that $\gG$ is a \textit{compatible collection of external contours} if it is a compatible collection and  every contour of $\gG$ is external in $\gG$.
Given $\bL$ a finite set of contour, we let $\cK(\bL)$ denote the set of compatible collections of contours included in $\bL$ and
$\cK_{\ext}(\bL)$ denote the set of compatible collection of external contours
(we use the notation $\cK_{\ext}(\gL)$ for $\bL=\cC_{\gL}$).
Given $\phi\in \gO_{\gL}$, we define 
\begin{equation}\label{defuext}
 \Upsilon^{\ext}_n(\phi):=\{ \gamma \in \Upsilon_n(\phi) \ : \ \gamma \text{ is external in }   \Upsilon_n(\phi) \}.
\end{equation}
Obviously $ \Upsilon^{\ext}_n(\phi)\in \cK_{\ext}(\gL)$.
We say that two contours $\gamma_1$ and $\gamma_2$ are \textit{externally compatible}  if they are compatible and 
$\bar \gamma_1\cap \bar  \gamma_2=\emptyset$. We use the notation $\gamma_1\parallel \gamma_2$.
We say that two collections
 $\gG_1,\ \gG_2  \in \cK_{\ext}(\bL)$ are  externally compatible if
 \begin{equation}\label{extcompatibl}
 \forall \gamma_1\in \gG_1, \  \forall \gamma_2\in \gG_2, \ \gamma_1 \parallel \gamma_2,
 \end{equation}
 or equivalently if $\gG_1\cap \gG_2=\emptyset$ and 
  $\gG_1\cup \gG_2\in \cK_{\ext}(\bL).$
 We also use the notation $\gG_1 \parallel   \gG_2$ for external compatibility between contour collections, and also $\gamma \parallel \gG$ for $\{ \gamma \} \parallel \gG$.

\subsection{Cluster expansion}\label{clusexp}
 
Partition functions which can be written as a sum over collections of compatible geometric objects such as \eqref{contourdecomp}
appears in a variety of situation in statistical mechanics. A powerful method called \textit{cluster expansion} has been engineered to analyze the associated systems in the 
low temperature regime (that corresponds to large $\gb$) .
We introduce it here as it appears in \cite{cf:KP}, with a set of  notation adapted to our context.

\medskip
 
Recall that  $\cC$ is the set of contours in $\bbZ^2$ and let  $w: \cC \to \bbR_+$ be an arbitrary function (in full generality $w$ could assume complex values, cf. \cite{cf:KP}).
Recall that $\cK(\bL)$ denote the set of compatible collections of  contours in $\bL$.
Given a finite subset ${\bL}$ of $\cC$, the partition function associated to $w$ and $\bL$,  $Z[\bL,w]$ is defined 
by
\begin{equation}\label{weightpar}
 Z[\bL,w]:= \sum_{\gG\in \cK(\bL)} \prod_{\gamma \in \gG} w(\gamma).
\end{equation}
For $\gL$ a subset of $\bbZ^2$ we write  $Z[\gL,w]$ for $Z[\cC_{\gL},w]$.

\medskip

We consider consider also $\bbP^w_{\bL}$ the probability measure on $\cK(\bL)$ corresponding to $Z[\bL,w]$ and call  $\Upsilon$ the associated random variable.
The distribution $\bbP^w_{\bL}$ has its support in $\cK(\bL)$ and we have for $\gG\in \cK(\bL)$,
\begin{equation}\label{pwl}
 \bbP^w_{\bL}(\Upsilon=\gG):=\frac{1}{Z[\bL,w]}\prod_{\gamma\in \gG}w(\gamma), 
\end{equation}
We use the notation $\bbP^w_{\gL}$ when $\bL:=\cC_{\gL}$ for $\gL$ a finite subset of $\bbZ$.

\medskip

\begin{rem}\label{teorema}
Going back to   \eqref{contourdecomp}, the reader can check that  the distribution of $\Upsilon(\phi)$ under 
$\bP_{\gL,\gb}$, $\gL$ simply connected, is given by $\bP^{w_{\gb}}_\gL$ where
\begin{equation}\label{specialweight}
 w_{\gb}(\gamma):=\frac{1}{e^{\gb|\tilde \gamma|}-1}.
\end{equation}
\end{rem}
\subsubsection{The key result}

The starting point of cluster expansion is the observe that the $\log$ of a partition function  
can be expressed as sum over geometric objects called \textit{clusters}.
A cluster of contour $\bC$ in $\bL$ is a finite non-empty subset of $\bL$ which cannot be split into two compatible parts (recall \eqref{compatiblesets})
or more formally which satisfies 
\begin{equation}\label{noncluster}
\forall \bB \subset \bC, \ \bB \perp (\bC \setminus \bB).
\end{equation}

We let $\cQ(\bL)$ denote the set of clusters in $\bL$ and $\cQ$ the set of all clusters (finite subsets of $\cC$).
The starting point of cluster expansion is the observation that $\log  Z[\bL,w]$ can be written as a sum over clusters
\begin{equation}\label{expansion}
 \log Z[\bL,w]:= \sum_{\bC \in \cQ(\bL)}  w^T(\bC),
\end{equation}
 where the modified weights $w^T$ are given by 
\begin{equation}\label{modifi}
w^T(\bC):= \sum_{\bB\in \cP(\bC)} (-1)^{|\bB|+|\bC|} \log Z[\bB,w],
\end{equation}
where $\cP$ stands for the set of parts. In fact \eqref{expansion} is almost immediate if we consider the sum over all subsets of $\bL$. Then
one can check that $w^T(\bC)=0$ if $\bC$ is not a cluster (we refer to the first lines in \cite[Section 3]{cf:KP} for full details).

\medskip

The reason why the expansion \eqref{expansion} is relevant is that if the original weights $w$ are small in a certain sense, and in particular decay exponentially with the 
length of the contours, then the modified weights $w^T(\bC)$ are also small and decay exponentially fast with the \textit{total length of the cluster}  $L(\bC)$, 
defined as follows
\begin{equation}
 L(\bC)=\sum_{\gamma\in \bC} |\tilde \gamma|.
\end{equation}
The powerful estimate displayed below is the main result of \cite{cf:KP}.

\begin{theorema}\label{superclust}
If there exists two functions $a$ and $d$, $\cC\to \bbR_+$, such that for every $\gamma$ 
 \begin{equation}\label{condit}
  \sum_{\{\gamma'\in \cC \ : \ \gamma' \perp \gamma\}} e^{a(\gamma')+d(\gamma')}w(\gamma')\le a(\gamma)
 \end{equation}
then 
 \begin{equation}\label{dabound}
 \sum_{\{ \bC\in \cQ \ : \  \bC  \perp \gamma \}} |w^T(\bC)|\exp\left(\sum_{\gamma'\in\bC}d(\gamma')\right) \le a(\gamma).
  \end{equation}
\end{theorema}

\begin{rem}\label{laremarc}
 For simplicity we introduced the result for the notion of contour compatibility/connectedness defined in Section \ref{contour}.
 However the result is purely algebraic and is remains valid if compatibility is replaced by another symmetric relation on contours and 
 an the notion of cluster is defined using this other relation.
In the present paper we use the result with compatibility replaced by external compatibility in the proof of Lemma \ref{dzip}.
\end{rem}

For all practical purpose, in the remainder of the paper, we use the criterion \eqref{condit} for a pair of simple functions  $$a_0(\gamma)=|\tilde \gamma| \text{
and  } d_0(\gamma):= \left(\gb-5\right) |\tilde \gamma|,$$
with $\gb>5$.
A simple and practical way of verifying condition \eqref{condit} in that case is to check for every $x^*\in (\bbZ^2)^*$
\begin{equation}\label{usualchoice}
  \sum_{\gamma\in \cC \ : \ x^*\in \gamma} e^{(\gb-4)|\tilde \gamma|} w(\gamma)\le 1.
 \end{equation}
 Given $\bC$ a set of contour, let us use the notation $x^*\in \bC$ for 
 \begin{equation}
 \exists \gamma \in \bC, \ x^*\in\gamma.
 \end{equation}
We let the reader check that
that for $x^*\in (\bbZ^2)^*$, any clusters which  satisfies $x^*\in \bC$ is incompatible with a contour of length $4$ which displays $x^*$ in its top right corner
(the choice for the sign being left open).
Applying \eqref{dabound} for these two contours of length $4$,  provided that \eqref{condit} holds for \eqref{usualchoice} 
we obtain thus that 
 \begin{equation}\label{dabound2}
 \sum_{ \bC\in \cQ \ : x^*\in \bC } |w^T(\bC)|e^{(\gb-5)L(\bC)} \le 8.
  \end{equation}

For the Solid-On-Solid model without constraint, which corresponds to the the weight function \eqref{specialweight},
one can check that 
\eqref{usualchoice}  holds provided $\gb>5$.

 The results mentioned in the rest of the section are classical consequences of Theorem \ref{superclust}, 
 but are sometimes exposed in the literature in a way that does not exactly fit the needs of our paper.
 For the sake of completeness, we prove these corollaries in Appendix \ref{apclus}.

\subsubsection{Free Energy and boundary effects}\label{fffa}
In our analysis we will be only interested in the case of \textit{translation invariant} weight functions $w$, 
meaning that $w(\gamma+x)=w(\gamma)$ for $x\in \bbZ^2$ where $\gamma+x$ is defined as the contour with the same sign as $\gamma$,
with the set of edges obtained by translating every edges of 
$\tilde\gamma$ by $x$.

\medskip

If the partition function of a statistical mechanics model has an expression of the form \eqref{weightpar}, 
with  translation invariant weights, the cluster expansion yields a simple expression for the free energy of the associated model.
Assuming that $w$ is a translation invariant weight function which satisfies \eqref{usualchoice}, 
the following  limit exists 
\begin{equation}\label{lafreen}
\limtwo{|\gL| \to \infty}{|\partial\gL|/|\gL|\to 0} \frac{1}{|\gL|} \log Z[ \gL ,w]=\tf(w).
 \end{equation}
 More precisely we have for $x^*$ an arbitrary point in the dual lattice $(\bbZ^2)^*$
\begin{equation}\label{fraconv}
 \tf(w)=\sum_{\bC \in \cQ(\bL) \ : \ x^*\in \bC}  \frac{1}{|\bC|}w^T(\bC),
\end{equation}
where $|\bC|:=\{y^*\in (\bbZ^2)^*\ : \ y^*\in \bC \}$ is the number of points in the dual lattice which are visited by a contour in $\bC$
(note that $|\bC|\le L(\bC)$ and that the inequality can be strict). The above expression does not depend on $x^*$ by translation invariance.

\medskip

The fact that the sum in \eqref{fraconv} converges is a consequence of \eqref{dabound2}.
This convergence result can simply be obtained by controlling the difference between the expression given for 
$|\gL|\tf(w)$ and $\log Z[ \gL ,w]$  using \eqref{dabound}. This difference can be shown to be proportional to the size of the boundary.

\medskip

We do not prove \eqref{lafreen} but present instead a very similar result for another kind of partition function. Given $\gamma\in \cC$, we let $Z[\gamma,w]$ denote the 
partition function corresponding to the set of contours in the domain $\bar \gamma$ which are compatible with $\gamma$,
\begin{equation}\label{cgammadef}
 \cC_{\gamma}:= \{ \gamma'\in \cC \ : \ \bar \gamma'\subset \bar\gamma \text{ and } \gamma'\mid \gamma\}
\end{equation}

\begin{lemma}\label{finalfrontier}
If $w$ is a translation invariant weight function which satisfies Equation \eqref{usualchoice} for $\gb$ sufficiently large
we have
\begin{equation}
  \Big |\log Z[\gamma ,w]- |\bar\gamma|\tf(w) \Big|\le \frac 1 4 |\tilde \gamma|.
\end{equation}
\end{lemma}
The proof of Lemma \ref{finalfrontier} is displayed in Appendix \ref{pfinal}, and \eqref{lafreen} can be obtained with only minor modifications.

\subsubsection{Correlation decay and infinite volume limits}\label{fffb}

We say that a countable collection of contours $\gG\subset \cC$  is locally finite if
\begin{equation}\label{locallyfinite}
 \forall x\in \bbZ^2, \quad  \#\{ \gamma \in \gG \ : \ x\in \bar \gamma\}<\infty.
\end{equation}
We let $\cK$ denote the set of \textit{locally finite} compatible collection of contours on $\bbZ^2$.
We say that a function $f: \cK\to \bbR$ is a local function if there exists a finite set 
$A \subset \bbZ^2$ such that $f(\gG)$ is entirely determined by $\gG\cap \cC'_A$,
where $$\cC'_A:= \{ \gamma \in \cC \ : \ \bar \gamma \cap A\ne \emptyset \}.$$
This is equivalent to say that there exists 
$\hat f : \cK[\cC'_A]\to \bbR$ such that 
$$f(\gG)=\hat f(\gG\cap \cC'_{A}).$$
Given $\gG\in  \cK[\cC'_{A}]$ we obtain as a consequence of \eqref{expansion} that 
\begin{equation}\label{godstim}
\bbP^w_{\bL}[\Upsilon \cap \cC'_{A}=\gG]=w_{\bL}(\gG)\frac{Z[\bL'_{A,\gG}]}{Z[\bL]}\\
=
w_{\bL}(\gG)\exp\left( -\sum_{\bC\in \cQ(\bL,A,\gG)} w^T(\bC)\right).
\end{equation}
where 
\begin{equation}\begin{split}
w_\bL(\gG)&:=\ind_{\{\gG\subset \bL\}}\prod_{\gamma\in \gG}w(\gamma),\\
\bL'_{A,\gG}&:=\left\{  \gamma \in \bL \setminus \cC'_A \ : \ \gamma \text{ is compatible with } \gG\right\}
\end{split}
\end{equation}
and $\cQ(\bL,A,\gG)$ is the set of clusters that either intersect $A$ or are connected with $\gG$
\begin{equation}
\begin{split}
 \cQ(\bL,A,\gG)&:= \{\bC\in \cQ(\bL) \ :  \bC\cap (\bL'_{A,\gG})^\cc\ne \emptyset  \}\\
 &=\{\bC\in \cQ(\bL) \ : \ \exists \gamma\in \bC, \ \gamma\in \cC'_A  \text{ or } \gamma \perp \gG  \},\\
 \end{split}
\end{equation}
When \eqref{usualchoice} holds, using Equation \eqref{godstim},
we can prove  two important consequences: Firstly,  $\bbP^w_{\bL}$ converges to an infinite volume $\bbP^w$ limit when 
$\bL$ exhaust $\cC$ with the convergence holding in the local sense (the expectation of every local function converges). 
Secondly, the correlation between two local functions decays exponentially with the distance of their support.

 \medskip
 
 The infinite volume limit $\bbP^w$ is defined via its finite dimensional projection (using Kolmogorov extension Theorem). It is the 
 unique probability on $\cK$ which satisfies for every finite subset $A$
 \begin{equation}\label{godstiminf}
\bbP^w[\Upsilon \cap \cC'_{A}=\gG]=
w(\gG)\exp\left( -\sum_{\bC\in \cQ(A,\gG)} w^T(\bC)\right).
\end{equation}
where $w(\gG):=\prod_{\gamma\in \gG}w(\gamma)$ and
 \begin{equation}
 \cQ(A,\gG):=\{\bC\in \cQ \ : \ \exists \gamma\in \bC, \ \gamma\in \cC'_A  \text{ or } \gamma \perp \gG  \}.
\end{equation}
The convergence of the sum in the exponential in \eqref{godstiminf} is ensured by \eqref{dabound2}.
We show that the convergence occurs in an exponential fashion and that spatial correlation decay exponentially. 
To state the result, we need to introduce the following notion of distance between finite subset of $\bbZ^d$
and the complement of a finite set of contours.

\begin{equation}
d(A,\bL^{\cc}):= \min\{  x\in A, \gamma\in \cC\setminus \bL, \ \max_{y\in \bar \gamma}  |x-y|\}.
\end{equation}
Note that when $A$ is fixed, this distance grows to infinity when $\bL$ exhausts $\cC$.

 \begin{proposition}\label{propinfi}
 If $w$ is a translation invariant weight function which satisfies \eqref{usualchoice}, 
Then
for every pair of  local functions  $f$ and $g$, $\cK \to [0,1]$ with respective supports $A$ and $B$  and every $\bL$, $\bL'$ such that  
$ d(A,\bL^\cc) \le d(A,(\bL')^\cc)$ we have 
\begin{equation}\label{convergenz}
\begin{split}
| \bbP^w_{\bL}[f(\gG)]-\bbP^w_{\bL'}[f(\gG)]|\le &|A|e^{-(\gb/100)d(A,\bL^\cc)},\\
|  \bbP^w_{\bL}[f(\gG)]-\bbP^w[f(\gG)]|\le &|A|e^{-(\gb/100)d(A,\bL)},
  \end{split}
\end{equation}
and also
\begin{equation}\label{decayze}
 |\bbP^w_{\bL}[f(\gG)g(\gG)]- \bbP^w_{\bL}[f(\gG)]\bbP^w_{\bL}[g(\gG)]|\le |A|e^{-(\gb/100)d(A,B)}.
\end{equation}

\end{proposition}

\noindent The proof of this result is displayed in Appendix \ref{ppropinf} for completeness. 
 
 \begin{rem}
 Let us remark that the result can be applied to the weights given by \eqref{specialweight} in order to obtain the convergence of the distribution
 of contours associated with the measure $\bP_{\gL,\gb}$ when $\gL$ is simply connected. 
A proof of  Theorem \ref{infinitevol} can then be deduced from this result by noticing that conditioned to the set of contour, the heights of the cylinders are
 independent geometric variables (see e.g.\ \cite[Lemma 4.3]{cf:part1}). The reader can refer to the proof of Proposition \ref{conviark} to see
 how results on the distribution of the field $\phi$ can be deduced from a result about the contour distribution.
 \end{rem}

\subsection{Contour decomposition for the wetting problem} \label{contsec}

We face various obstacles when trying to obtain a decomposition similar to \eqref{contourdecomp}, for $\cZ^{n,h}_{\gL,\gb}$.
First because the function $\phi$ cannot be expressed directly from the contour collection $\Upsilon_n(\phi)$. 
Opting for a representation using cylinders does not fully solve the problem, 
since the quantities $\ind_{\{\phi(x)\ge 0\}}$ and $\ind_{\{ \phi(x)=0\}}$ which appear in the Hamiltonian cannot be fitted in the expansion, 
because they depend on the set of contours in a highly non-local way.

\medskip

The way out is to opt for a more abstract representation, where the contours in the sum do not correspond to the level line of $\phi$.
We obtain one such representation for each choice of boundary condition $n$
\begin{equation}\label{reprezent}
\cZ^{n,h}_{\gL,\gb}= \sum_{\gG\in \cK(\gL)} \prod_{\gamma\in \gG} w^h_n(\gamma).
\end{equation}
Let us stress that using this type of contour decomposition is not a new idea, and that our 
weight function is very similar to the ones used e.g.\ in \cite{cf:ADM, cf:CM}.

\medskip

In order to provide the expression of the weights $w^h_n$ (displayed in \eqref{defw}) we need to introduce a few notation.
Given $\gamma$ a contour, $n\in \bbZ_+$, we let $\gO[\gamma,n]$ and  $\bar \gO[\gamma,n]$ denote the sets of functions in $\bar \gamma$ defined as follows
\begin{equation}\begin{split}\label{defoom}
\gO[\gamma,n]&:=\{\phi,\ \bar \gamma \to \bbZ \ :  \ \forall x\in \Delta^-_{\gamma},\  \gep(\gamma)(\phi(x)-n)\ge 0\},\\
\bar \gO[\gamma,n]&:=\gO[\gamma,n]\setminus \gO[\gamma,n+\gep(\gamma)]\\
&= \{ \phi \in \gO[\gamma,n] \ : \ \exists x\in \Delta^-_{\gamma},\ \phi(x)=n\}.
\end{split}\end{equation}
We define $\gO^+[\gamma,n]$, $\bar \gO^+[\gamma,n]$ as the restrictions of $\gO[\gamma,n]$ and $\bar \gO[\gamma,n]$ to the set of non-negative functions 
(recall the convention adopted in
\eqref{positivity}).
The set $\bar \gO[\gamma,n]$ and  $\gO[\gamma,n]$ can respectively be described as the sets of functions
$\phi$ such that $\Upsilon_n(\phi)$ resp.\ $\Upsilon_n(\phi)\setminus \{\gamma\}$
is compatible with $\gamma$. 

\medskip

Given $\gamma$, $n$ and $h>0$, we define $z^{h}_n(\gamma)$ and $\bar z^h_n(\gamma)$
to be the two partition functions associated with the sets  $\gO^+[\gamma,n]$ and $\bar \gO^+[\gamma,n]$ and the energy functional
$\gb \cH^n_{\bar \gamma}(\phi)-h|\phi^{-1}\{0\}|$ (recall \eqref{defhamil})
\begin{equation}\label{partisioux}\begin{split}
z^h_n(\gamma)&:= \sum_{\phi\in \gO^+[\gamma,n]} e^{-\gb \cH^n_{\bar \gamma}(\phi)+h |\phi^{-1}\{0\}|},\\
\bar z^h_n(\gamma)&:= \sum_{\phi\in \bar \gO^+[\gamma,n]} e^{-\gb \cH^n_{\bar \gamma}(\phi)+h|\phi^{-1}\{0\}|}.
\end{split}\end{equation}
We extend the definition to the case of negative $n$ by setting $z^h_n=\bar z^h_n=0$ for $n<0$.
 The reader can check that $\gO^+[\gamma,n]:=\bigcup_{k\ge 0}\bar\gO^+[\gamma,n+\gep(\gamma)k]$ and thus that
\begin{equation}\label{relatz}
 z^h_n(\gamma)= \sum_{k\ge 0} e^{-k \gb|\tilde \gamma|}\bar z^h_{n+\gep(\gamma)k}(\gamma).
\end{equation}
We are now ready to define our contour weight $w^h_n(\gamma)$ (for $n\ge 0$) as follows
\begin{equation}\label{defw}
w^h_n(\gamma)=  \frac{ e^{-\gb|\tilde \gamma|} z^h_{n+\gep(\gamma)}(\gamma)}{\bar z^h_n(\gamma)}.
\end{equation}
Note that, with our convention, negative contours have weight zero for $n=0$.
This definition turns out to be the most natural to obtain a contour representation for the partition function.

\begin{proposition}\label{trax}
The contour representation \eqref{reprezent} of the partition $\cZ^{n,h}_{\gL,\gb}$ holds true for the weights defined in \eqref{defw} 
when $\gL$ is simply connected (recall \eqref{connectedness}).
\end{proposition}

While it involves some notation, the proof is not conceptually difficult. 
The idea is to process recursively starting with external contours of the field $\phi$ \eqref{defexternal} and iterating the procedure.

\begin{proof}

The starting point of our proof is the observation that the complete description of 
$\phi\in \gO^+_{\gL}$ can be obtained by knowing the set of external contours $\Upsilon^{\ext}_n(\phi)$ together with the associated intensity,
and the value of the restriction $\phi\restrict_{\bar \gamma}$ for every $\gamma\in \Upsilon^{\ext}_n(\phi)$.
When $\gamma$ is an external contour associated with boundary condition $n$ we have
$\phi\restrict_{\bar \gamma}\in \gO^+[\gamma,n+\gep(n)]$ (recall \eqref{defoom}), and this is the only requirement that $\phi\restrict_{\bar \gamma}$
must satisfy.
Hence we obtain directly from \eqref{partisioux}
\begin{equation}\label{tapzero}
Z^{n,h}_{\gL,\gb}
=\sum_{\gG\in \cK_{\ext}(\gL)}\prod_{\gamma\in \gG}  e^{-\gb |\tilde \gamma|} z^h_{n+\gep(\gamma)}(\gamma).
\end{equation}
Using the definition \eqref{defw} we can rewrite the sum as
\begin{equation}\label{tapouz}
Z^{n,h}_{\gL,\gb}
=\sum_{\gG\in \cK_{\ext}(\gL)}\prod_{\gamma\in \gG} w^u_n(\gamma) \bar z^h_n(\gamma).
\end{equation}
Now let us introduce $\bar \cK_{\ext}(\gamma)$ which is the space in which the set of external contours associated to an element of $\bar \gO^+[n,\gamma]$ lies
\begin{equation}\label{daext}
 \bar \cK_{\ext}(\gamma):=\{ \gG \in \cK_{\ext}(\bar \gamma) \ : \ \gG \mid \gamma \}.
\end{equation}
Decomposing according to the external contours of $\phi\restrict_{\bar \gamma}$ we obtain similarly to \eqref{partisioux} that
\begin{equation}\label{tapun}
  \bar z^h_n(\gamma)=  \sum_{\gG_1\in \bar \cK_{\ext}(\gamma)}\prod_{\gamma_1\in \gG_1}  
  w^h_n(\gamma_1) \bar z^h_n(\gamma_1).
\end{equation}
Injecting \eqref{tapun} in \eqref{tapouz} and iterating the procedure, we obtain \eqref{reprezent}.
\end{proof}

\subsection{Rewriting partition functions}

In order to obtain bounds on the contour weights  $w^h_n(\gamma_1)$ 
which are sufficient to prove \eqref{usualchoice}, we have to use alternative expressions for the partition functions in order  to facilitate
the comparison between 
$z^h_{n+\gep(\gamma)}(\gamma)$ and $\bar z^h_n(\gamma)$.
One of the objective is to get rid the positivity constraint for $\phi$. 
Let us define for  $\gG$ a finite subset of $\bbZ^2$,
\begin{equation}\label{ZZplus}
\cZ^+_{\gG}:=\cZ^0_{\gG,\gb}=\sum_{\phi\in \gO^+_{\gG}} \exp\left(-\gb \cH_\gL(\phi)\right),
\end{equation}
and set 
\begin{equation}
\bar H(\gG):= \log \cZ^+_{\gG}-|\gG|\log \left(\frac{e^{4\gb}}{e^{4\gb}-1}\right).
\end{equation}
We introduce the partition functions  $z(\gamma)$,  $\bar z(\gamma)$ which corresponds to the model without positivity constraint or interaction at level zero
\begin{equation}\label{tambem}
 z(\gamma):=\sum_{\phi\in \gO[\gamma,n]}e^{-\gb\cH^n_{\bar \gamma}(\phi)}, \quad \bar z(\gamma):=\sum_{\phi\in \bar \gO[\gamma,n]}e^{-\gb\cH^n_{\bar \gamma}(\phi)},
\end{equation}
which, by translation invariance, do not depend on $n$. We consider $\bP^n_{\gamma}$ and $\bar \bP^n_{\gamma}$ the associated probability distributions on $\bar \gO[\gamma,n]$ and $\bar \gO[\gamma,n]$.

\begin{lemma}\label{rwo}
We have for any $n\ge 1$ 
\begin{equation}\label{swoutch}\begin{split}
z^h_n(\gamma)&=\sum_{\phi\in \gO[\gamma,n]} e^{-\gb \cH^n_{\bar \gamma}(\phi)+ u|\phi^{-1}(\bbZ_-)|-\bar H(\phi^{-1}(\bbZ_-))},\\
\bar z^h_n(\gamma)&=\sum_{\phi\in \bar \gO[\gamma,n]} e^{-\gb \cH^n_{\bar \gamma}(\phi)+ u|\phi^{-1}(\bbZ_-)|-\bar H(\phi^{-1}(\bbZ_-))},
\end{split}\end{equation}
where in the formula above $u=u_h:= h- \log \left(\frac{e^{4\gb}}{e^{4\gb}-1}\right).$
Alternatively we can write 
\begin{equation}\label{dooz}\begin{split}
z^h_n(\gamma)&= z(\gamma)\bE^n_{\gamma}\left[e^{u|\phi^{-1}(\bbZ_-)|-\bar H(\phi^{-1}(\bbZ_-))}\right],\\
\bar z^h_n(\gamma)&= z(\gamma)\bE^n_{\gamma}\left[e^{u|\phi^{-1}(\bbZ_-)|-\bar H(\phi^{-1}(\bbZ_-))}\ind_{\{\exists x\in \Delta^-_{\gamma}, \phi(x)=n\}} \right],\\
&=\bar z(\gamma)\bar \bE^n_{\gamma}\left[e^{u|\phi^{-1}(\bbZ_-)|-\bar H(\phi^{-1}(\bbZ_-))}\right].
\end{split}\end{equation}
\end{lemma}

\begin{proof}
The statement \eqref{swoutch} can be proved in the same manner as \cite[Lemma 3.1 and Lemma 3.2]{cf:part1}, and \eqref{dooz} is an obvious consequence of it.
\end{proof}

It follows from the definition of $\bar H$ that if
$\gG=\bigcup_{i\in \lint 1, m\rint} \gG_i$ is the decomposition of $\gG$ into maximal connected components (in $\bbZ^2$) then
\begin{equation}\label{maxdecomp}
\bar H(\gG)=\sum_{i\in \lint 1, m\rint} \bar H(\gG_i).
\end{equation}
For our purpose we  need in fact to estimate sharply the value of $\bar H$  only for connected components of size one and two, 
and to have a rougher estimate for other connected sets.

\begin{lemma}\label{zoomats}
 
We have for any two neighboring points  $x\sim y$  in $\bbZ^2$
\begin{equation}\label{sitmats}
 \bar H\{x\}=0 \quad \text{ and } \quad \bar H\{x,y\}=\log \left(\frac{1-J^4}{1-J^3}\right).
 \end{equation}
For $\gb$ sufficiently large,
for all $|\gG|$ connected and larger than $2$ we have 
\begin{equation}\label{morseu}
 0\le \bar H(\gG)\le 2J^2|\gG|.
\end{equation}
\end{lemma}

\begin{proof}
 The equalities in \eqref{sitmats} are the result of a direct computation whose details are given in the proof of \cite[Lemma 3.2]{cf:part1}.
 For \eqref{morseu} the lower bound is a consequence of the super-additivity of $\bar H$ (see also \cite[Lemma 3.2]{cf:part1}). For the upper bound we use the expansion \eqref{expansion} to evaluate the partition function of SOS which correspond to weight function $w_{\gb}$ given in \eqref{specialweight}.
 We have 
 \begin{equation}
  \log \cZ^+_{\gG}\le  \log \cZ_{\gG,\gb}:=\sum_{\bC\in \cQ(\gG)} w_{\gb}^T(\bC).
 \end{equation}
 As \eqref{usualchoice} is valid for $w_{\gb}$ if $\gb$ is sufficiently large, \eqref{dabound2} implies that
 \begin{equation}\label{treck}
 \sum_{\bC\in \cQ(\gG)} w_{\gb}^T(\bC)\le |\gG|\left[ 2e^{-4\gb}+O(e^{-6\gb})\right],
 \end{equation}
 which is sufficient to conclude.
\end{proof}

\subsection{Peak probabilities}\label{peaksec}

We recall here a result concerning the asymptotic probability of observing ``peaks'' of a given shape for $\phi$ under the measure $\bP_{\gL,\gb}$.
We provide a result which is slightly more general than the one proposed in \cite[Proposition 4.5]{cf:part1}.
Given $\bL$ a finite set of contour included in $\cL_\gL$, and $\gb>0$, we define $\bP_{\bL,\gL,\gb}$ to be a measure on $\gO_{\gL}$
which can be sampled as follows 
\begin{itemize}
\item [(A)] Sample a set of contour $\Upsilon$ according to the measure $\bP^{w_{\gb}}_{\bL}$ (recall \eqref{specialweight}).
\item [(B)] For each contour $\gamma\in \Upsilon$ sample independently a geometric variable $k(\gamma)$ satisfying 
$\bbP[k(\gamma)=i]=[w_{\gb}(\gamma)]^{-1} e^{-\gb|\tilde \gamma|i}$.
\item [(C)] Set (recall \eqref{cylfunc})
$$\phi:=\sum_{\gamma\in \Upsilon} \varphi_{(\gamma,k(\gamma))}.$$
\end{itemize}
Note that when $\bL=\cC_{\gL}$ we have $\bP_{\bL,\gL,\gb}=\bP_{\gL,\gb}$. The probability distribution $\bP^0_{\gamma}$ and $\bar \bP^0_{\gamma}$ defined below 
Equation \eqref{tambem} are also of the form $\bP_{\bL,\bar \gamma,\gb}$ for adequate choices of $\bL$.
This definition thus allows us to treat measures which include special boundary condition or contour restriction.

\begin{proposition}\label{rouxrou}
If $\gb$ is sufficiently large, then such for any choice $\gL$, $\bL$ and $n$  and any triple of distinct vertices $(x,y,z)\in \gL^3$
such that $x\sim y \sim z$ we have 
\begin{equation}\label{zups}
\begin{split}
 \bP_{\bL,\gL,\gb}[\phi(x)\ge n]&\le 2 e^{-4\gb n},\\
 \bP_{\bL,\gL,\gb}[\min(\phi(x),\phi(y)) \ge n]&\le 2 e^{-6\gb n},\\
  \bP_{\bL,\gL,\gb}[\min(\phi(x),\phi(y),\phi(z))\ge n]&\le 2 ne^{-8\gb n}.
  \end{split}
\end{equation}
If we assume in addition that $\bL$ contains the positive contour of length $4$ enclosing $x$, then
\begin{equation}\label{infin}
 \bP_{\bL,\gL,\gb}[\phi(x)\ge n]\ge \frac{1}{2} e^{-4\gb n}.
\end{equation}
If we assume that $\bL$ contains the positive contour of length $6$ enclosing $x$ and $y$, then
\begin{equation}\label{infde}
 \bP_{\bL,\gL,\gb}[\phi(x)\ge n]\ge \frac{1}{2} e^{-6\gb n}.
 \end{equation}
\end{proposition}

The proof of \eqref{zups} is identical to that of \cite[Proposition 4.5]{cf:part1}. The proofs of \eqref{infin} and \eqref{infde} are detailed in Appendix \ref{lazt}.

\subsection{Monotonicity and the FKG inequality}

The set $\gO_{\gL}$ as well as its variants ($\gO^+_{\gL}$ and others introduced later in the paper) are naturally equipped with an order defined as follows 
$$\phi\le \phi'   \quad \Leftrightarrow \quad \forall x\in \gL,\  \phi(x)\le \phi'(x).$$
Using this order we can define a notion of increasing function ($f$ is increasing if $\phi\le \phi'\Rightarrow f(\phi)\le f(\phi')$) and of increasing event 
($A$ is increasing if the function $\ind_A$ is).
We say that a probability measure $\mu$ on $\gO_{\gL}$ \textit{stochastically dominates} another one  $\mu'$ (we write $\mu \succcurlyeq \mu'$) if for any increasing function $f$
$$\mu(f(\phi))\ge \mu'(f(\phi))$$

\medskip

The FKG inequality  allows to say that if a probability measure $\mu$ supported on a subset of $\gO_{\gL}$  satisfies a certain condition,
increasing functions are positively correlated.
For the inequality to be satisfied \cite{cf:Holley}, we need the support of $\mu$ to be a \textit{distributive lattice}, that is, to
be stable over the operations $\vee$ and $\wedge$ 
defined by 
$$(\phi_1\vee \phi_2)(x):=\max(\phi_1(x),\phi_2(x)), \text{ and } (\phi_1\wedge \phi_2)(x):=\min(\phi_1(x),\phi_2(x)).$$
 Moreover the probability considered needs to verify Holley's condition \cite[Equation (7)]{cf:Holley},
 \begin{equation}\label{hcond}
 \mu(\phi_1\vee \phi_2) \mu(\phi_1\wedge \phi_2)\ge \mu(\phi_1)\mu(\phi_2).
\end{equation}
If this is satisfied then for any pair of increasing functions $f$ and $g$ we have 
\begin{equation}\label{FKeq}
 \mu(f(\phi)g(\phi))\ge \mu(f(\phi)) \mu(g(\phi)).
\end{equation}

\medskip

We obtain as immediate consequences of the FKG inequality, several stochastic domination results.
Given $\psi\in \gO_{\infty}$ and $\gL\subset \bbZ^2$, $\gb>0$ and $h\in \bbR$,
we let $\tilde \bP^{\psi,h}_{\gb,\gL}$ denote a measure defined on a subset $\tilde \gO_\gL \subset \gO_{\gL}$ which is a 
distributive lattice, with the probability of each state proportional to the Gibbs weight
$\exp\left(-\gb\cH^\psi_\gL(\phi)+ h|\phi^{-1}\{0\}|\right)$.

\begin{cor}\label{FKalt}
The following holds,
\begin{itemize}
 \item [(i)] For any increasing event $A$ 
 \begin{equation}\label{FKG}
      \tilde \bP^{\psi,h}_{\gb,\gL}[ \ \cdot \  | \ A] \succcurlyeq     \tilde \bP^{\psi,h}_{\gb,\gL}.
\end{equation}
\item[(ii)] For any $h'>h$
\begin{equation}
      \tilde \bP^{\psi,h'}_{\gb,\gL} \preccurlyeq    \tilde \bP^{\psi,h}_{\gb,\gL}.
      \end{equation}
      \item[(iii)] For any $\psi'\ge \psi$
\begin{equation}
      \tilde \bP^{\psi',h}_{\gb,\gL} \succcurlyeq     \tilde \bP^{\psi,h}_{\gb,\gL}.
      \end{equation}
\end{itemize}
\end{cor}

\begin{proof}
 The first point is immediate, for the other ones we simply have to notice that 
 $\exp((h-h')|\phi^{-1}\{0\}|)$ and  $\exp\left(\gb \left(\cH^\psi_\gL(\phi)-\cH^{\psi'}_\gL(\phi) \right)  \right)$ are increasing functions.
\end{proof}

\section{Organization of the proof of Theorem \ref{main}}\label{zorga}

We start with a small notational remark. As our main result concerns the behavior of the free energy close to $h_w(\gb)$, 
it is  more convenient for us to work as in the statement of Theorem \ref{main} with the parameter $u=h-h_w(\gb)$ than with $h$.
Therefore in most cases we work with all quantities defined as functions of $u$ rather than $h$.
When $h$ appear in a computation, we always assume that 
$$h=h_u=: h_w(\gb)+u.$$

\subsection{Contour stability and consequences}\label{stabcon}

If we want \eqref{reprezent} to yield information about the free energy, we need the contour weights $w^u_n(\gamma)$ to be small, or more precisely 
we want \eqref{usualchoice} to be satisfied. 
We say that a contour is $n$-stable for $u$ if
\begin{equation}\label{stability}
w^u_n(\gamma)\le e^{-(\gb-1) |\tilde \gamma|}.
 \end{equation}

\medskip

The most important part of our proof is to show that we can partition $\bbR_+$ 
into intervals $([u^*_{n+1},u^*_{n}])_{n\ge 0}$ (with the convention than $u^*_0=\infty$) in which all the contours are $n$-stable.
This result also plays a central role in  our proof of Theorem \ref{Gibbs}.

\begin{theorem}\label{converteo}
When $\gb$ is sufficiently large.
 There exists a decreasing sequence  $(u^*_n)_{n\ge 1}$ satisfying 
 $$u^*_n\in \left[\frac{1}{200} J^{n+2}, 200 J^{n+2}\right]$$
 such that all contours are $n$ stable for $u\in[u^*_{n+1},u^*_{n}]$.
\end{theorem}

We give a road-map for the proof of Theorem \ref{converteo} in Section \ref{sketch}, by presenting the main steps.
The detailed proof is then given in Section \ref{ponverteo}.

\medskip

Note that the $n$-stability of all contour implies that  \eqref{usualchoice} is satisfied for $w^u_n$.
Indeed a classical counting argument shows that for $k\in \bbN$ even and $x^*\in (\bbZ^2)^*$
\begin{equation}\label{counting}
\#\{ \gamma\in \cC \ : \ |\tilde \gamma|=k,\  x^*\in \gamma\} \le 8.3^{k-2}\le 3^{k},
\end{equation}
(starting from $x^*$ we have $2$ choices for the sign, $4$ choices for the first step, at most $3$ for the other steps, and the last step is determined by the fact that $\gamma$ is a loop).

Thus combining Theorem \ref{converteo} with the results introduced in Sections \ref{fffa} and \ref{fffb}  
we can derive consequences for the free energy and the measure $\bP^{n,h}_{\gL,\gb}$. 
These consequences are detailed in Section \ref{dalast}. 
We state here two statements which are of interest in the proof of Theorem \ref{main} which are respectively proved in Section \ref{laprov} and \ref{lapriv}.
First, we obtain a result concerning the regularity of the free energy.

 \begin{proposition}\label{infinitz}
 The free energy $u\mapsto \bar \tf(\gb,u)$ is infinitely differentiable on $(u^*_{n+1},u^*_{n})$ for $n\ge 0$
 for all $n\ge 0$. Moreover all derivatives of $\tf(\gb,h)$ are uniformly bounded on $(u^*_{n+1},u^*_{n})$.
 \end{proposition}
 
 Secondly, we obtain a priori bound on the derivative which together with Theorem \ref{oldmain} allows a
 sharp asymptotic estimates on the layering transition points $u^*_n$.
 
 \begin{proposition}\label{convirj}
  Given $\gb>\gb_0$ sufficiently large, there exists a constant such that for every $n$ and every $u\in (u^*_{n+1},u^*_n)$ (where by convention $u^*_0=+\infty$)
  we have 
  \begin{equation}
  \frac{1}{10}J^{2n}\le \partial_u\bar \tf(\gb,u)\le 10 J^{2n}.
  \end{equation}
  In particular $\bar \tf(\gb,u)$ is not differentiable at $u^*_n$.
\end{proposition}
\noindent To conclude the proof of Theorem \ref{main}, we provide a proof of \eqref{asimptic}.

  \begin{proof}[Proof of \eqref{asimptic}]
  
We make use of  Theorem \ref{oldmain}.
  Recalling \eqref{defun}, Equation \eqref{lequivdelamor} implies that for every $u,v\in [u_{n+1}, u_n]$
  \begin{equation}
   \bar \tf(\gb,v)-\bar\tf(\gb,u)= \alpha_1 J^{2n} (v-u)+o(J^{3n}).
  \end{equation}
By convexity, this implies that for $\delta>0$, for all $n$ sufficiently large  
  we have 
  \begin{equation}
   \forall u\in (u_{n+1}(1+\delta), u_{n}(1-\delta)), \quad \left| \partial_u  \bar \tf(\gb,u)- \alpha_1 J^{2n}\right| \le 2\delta J^{2n}.
  \end{equation}
  In view of Proposition \ref{convirj} and the of the fact that $\alpha_1, \alpha_2 \in [1/2,2]$ (cf. Proposition \ref{rouxrou}), we can conclude that 
  for $n$ sufficiently large 
  $$ u^*_n\in [u_n(1-\delta), u_n(1+\delta)].$$
  
\end{proof}

\subsection{Truncated weights and road map to Theorem \ref{converteo}}\label{sketch}

Our first step is to prove $n$-stability in a reduced intervals.
We define for $n\ge 1$,
 \begin{equation}
 u^{+}_n:= 200 J^{n+2}, \quad u^-_n:= \frac{1}{200} J^{n+2}, 
  \end{equation}
 and also $u^{\pm}_0=\infty$. And we prove the following.
 \begin{proposition}\label{truta}
 For all $\gb$ sufficiently large
 every contour is $n$-stable for $u\in [u^{+}_{n+1}, u^{-}_n]$.
 \end{proposition}
Using this partial result, we can obtain a characterization of $u^*_n$. This requires introducing the notion of truncated weights and free energy 
(we follow here ideas which were developed in  \cite{cf:CM}).
We define the truncated weights $ w^{u,\tr}_n$ by 
\begin{equation}\label{trunk}
 w^{u,\tr}_n(\gamma):= \max\left( e^{-(\gb-1)|\tilde \gamma|},  w^u_n(\gamma)\right).
\end{equation}
We define in the same manner
\begin{equation}
\cZ^{n,u,\tr}_{\gL,\gb}:= \sum_{\gG\in \cK(\gL)} \prod_{\gamma\in \gG} w^{\tr,u}_n(\gamma),
\end{equation}
and the corresponding free energy
\begin{equation}
\bar \tf^{\tr}_n(\gb,u):= \limtwo{|\gL|\to \infty}{|\partial\gL|/|\gL|\to 0} \frac{1}{|\gL|}\log \cZ^{n,u,\tr}_{\gL,\gb}-\tf(\gb).
\end{equation}
In view of \eqref{reprezent}, we have for every $n$ and $u$
\begin{equation}\label{zinak}
  \bar \tf^{\tr}_n(\gb,u)\le  \bar\tf(\gb,u),
\end{equation}
and equality is achieved if and only if all contours are stable (the \textit{only if} part of the statement may appear less obvious, but as we 
do not use that fact in our proof, we leave it as an exercise to the interested reader).
 In particular a simple consequence of Proposition \ref{truta} is the following.
 \begin{cor}\label{simplecor}
 For every $n\ge 0$ $u\in [u^{+}_{n+1}, u^-_n]$, $$\tf^{\tr}_{n}(\gb,u)=\tf(\gb,u).$$
 \end{cor}
Another important observation, that as the weights $w^{u,\tr}_n(\gamma)$ are continuous in $u$, so are the weights $w^T(\bC)$ associated to clusters.
Thus  as the convergence \eqref{fraconv} is uniform, the function $u\mapsto \bar \tf^{\tr}_n(\gb,u)$ is continuous for every $n$.
Now from Corollary \ref{simplecor}, we  have  for any $n\ge 1$
\begin{equation}\begin{split}
\tf^{\tr}_{n}(\gb,u^-_n)&=\tf(\gb,u^-_n)\ge \tf^{\tr}_{n-1}(\gb,u^-_n),\\
\tf^{\tr}_{n-1}(\gb,u^+_n)&=\tf(\gb,u^+_n)\ge \tf^{\tr}_{n}(\gb,u^+_n).
\end{split}\end{equation}
Using the continuity of $[\tf^{\tr}_{n-1}-\tf^{\tr}_{n}](\gb,u)$ we define 
 \begin{equation}\label{caraca}
  u^*_n:=\min\left\{v\in [u^-_n,u^+_n] \ : \  \tf^{\tr}_{n-1}(\gb,u)=\tf^{\tr}_{n}(\gb,u) \right\}.
 \end{equation}
 To complete the proof of Theorem \ref{converteo}, we need to extend the stability result to the interval $[u^*_n,u^*_{n+1}]$.
 \begin{proposition}\label{mejor}
   For all $\gb$ sufficiently large, 
   every contour is $n$-stable for $u\in[u^*_{n+1},u^*_n]$.
  \end{proposition}

  \begin{rem}
The characterization of $u^*_n$  as a $\min$ in \eqref{caraca}, is a bit arbitrary in the sense that the only requirements of the proof are $u\in [u^-_n,u^+_n]$ and
$\tf^{\tr}_{n-1}(\gb,u)=\tf^{\tr}_{n}(\gb,u)$. It does not mean however that there is any freedom in the choice of $u^*_n$, 
as further results implies that 
\begin{equation}\label{zuniks}
\left\{v\in [u^-_n,u^+_n] \ : \  \tf^{\tr}_{n-1}(\gb,u)=\tf^{\tr}_{n}(\gb,u) \right\}=\{u^*_n\}.
\end{equation}
While the definition of the truncated potential also offers some degree of freedom, a consequence of latter results is that the value of $u^*_n$ does not depend on the particular choice which is made for truncation.
 \end{rem}

 Proposition \ref{truta} turns out to be the more difficult statement as it requires quantitative estimates which proves to be quite technical.
 Its extension, Proposition \ref{mejor} is proved using softer arguments
 combining a monotonicity statement (Lemma \ref{mozon} below) together with Lemma \ref{finalfrontier}.
 For the proof of both Proposition \ref{truta} and Proposition \ref{mejor}, an important building brick is the following monotonicity consideration.
 \begin{lemma}\label{mozon}
 For any $n\in \bbN$, we have:
 \begin{itemize}
 \item[(i)] For any positive contour, $u\mapsto w^u_n(\gamma)$ is decreasing in $u$. 
\item[(ii)]  For any negative contour $u\mapsto w^u_n(\gamma)$ is increasing in $u$.
 \end{itemize}
 \end{lemma}
 \begin{rem} \label{mouton}
 A consequence the above statement for each contour, the proof of Propositions \ref{truta} and \ref{mejor} reduces to checking 
 stability for one value of $u$ which is chosen at an extremity of the interval (the right extremity for negative contour, the left one for positive contour).
\end{rem}

 \section{Proof of Theorem \ref{converteo}}\label{ponverteo}
 
 In this section, we prove all the statements exposed in Section \ref{stabcon}.
 We prove Lemma \ref{mozon} first in Section \ref{mooz}, while the other subsections are devoted to the proof of Propositions \ref{truta} and \ref{mejor}.
 Our proof for the contour's stability depends on the size of the contour.
This gives a utility to the following definition.
Here and in the remainder of the paper, $\Diam(\gamma)$ denotes the Euclidean diameter of the geometric contour $\tilde \gamma$ considered as a subset of $\bbR^2$.
 
 \begin{definition}
 A contour is said to be $n$-\textit{small} if $\Diam(\gamma)\le \lfloor \max(\gb,\gb n)^2 \rfloor$.
 A contour which is not small is said to be \textit{large}.
 \end{definition}

The stability of small contours can relatively is proved directly ``by hand" in Section \ref{smallcproof}, using directly the estimates we have for the Solid-On-Solid measures.
The stability of large contours is proved in two steps, first we restrict our proof to the interval $[u^{+}_{n+1},u^{-}_n]$ to prove Proposition \ref{truta}, 
this is the most delicate part and it spreads from Section \ref{prezent} do Section \ref{dcimus} with the more technical computation postponed to Section 
\ref{secrestrict}.
The last step of the proof of large contour stability is the extension to the full interval $[u^{*}_{n+1},u^{*}_n]$ to complete the proof of Proposition \ref{mejor}. This is done is Section \ref{mejorplus}.

\medskip

Note that whenever a contour is $n$-small we also have a bound on the enclosed area, which we are to use in most computations
\begin{equation}\label{areabound}
|\bar \gamma| \le \Diam(\gamma)^2\le   \max(\gb,\gb n)^4.
\end{equation}

 \subsection{Proof of Lemma \ref{mozon}}\label{mooz}
Let us assume for simplicity that $\gamma$ is a positive contour (the proof for the negative case being identical).
We let $\bP^{n,u}_{\gamma}$  and $\bar \bP^{n,u}_{\gamma}$ be the respective probability on $\gO^+[\gamma,n]$, and $\bar \gO^+[\gamma,n]$ corresponding to the partition functions
 $z^u_n(\gamma)$ and $\bar  z^u_n(\gamma)$, that is (recall $h=u+h_w(\gb)$)
 \begin{equation}\begin{split}\label{defpnu}
  \bar \bP^{n,u}_{\gamma}(\phi)&:=\frac{1}{z^u_n} e^{-\gb \cH^n_{\bar \gamma}(\phi)+h |\phi^{-1}\{0\}|}, \quad \forall \phi \in \bar \gO^+[\gamma,n],\\
   \bP^{n,u}_{\gamma}(\phi)&:=\frac{1}{z^u_n} e^{-\gb \cH^n_{\bar \gamma}(\phi)+h |\phi^{-1}\{0\}|}, \quad \forall \phi \in \gO^+[\gamma,n]
 \end{split}\end{equation}
Using these definitions, the reader can check that the logarithmic derivative of $w^u_n(\gamma)$ can be expressed in the following manner.
 \begin{equation}\label{dirivatz}
 \partial_u \log w^u_n(\gamma)= \bE^{n+1,u}_{\gamma} [ |\phi^{-1}(0)|]- \bar \bE^{n,u}_{\gamma} [|\phi^{-1}(0)|].
 \end{equation}
 As $|\phi^{-1}(0)|$ is a decreasing function of $\phi$, 
if we show that $\bP^{n+1,u}_{\gamma}$ stochastically dominates $\bar \bP^{n,u}_{\gamma}$, then it implies that the r.h.s.\ of
 \eqref{dirivatz} is negative which concludes the proof.
 Let us introduce the events.
 \begin{equation}
\begin{split}
\cA&:= \{   \exists x\in  \Delta^-_{\gamma},\ \phi(x)\le n\},\\
\cB&:= \{   \forall  x\in  \Delta^-_{\gamma}, \ \phi(x)\ge n+1\}.
 \end{split}
 \end{equation}
Observe that  $\bar \bP^{n,u}_{\gamma}$ and $\bP^{n+1,u}_{\gamma}$ can both be defined as a conditioned variant of  $\bP^{n,u}$. We have
 \begin{equation}
  \bar \bP^{n,u}_{\gamma}=\bP^{n,u}_{\gamma}[  \ \cdot  \ | \ \cA ] \text{ and }  \bP^{n+1,u}_{\gamma}=P^{n,u}_{\gamma}[ \ \cdot \ | \ \cB ].
 \end{equation}
Noting that $\gO^+[\gamma,n]$ is a distributive lattice, that $\cA$ is a decreasing event and  that $\cB$ is an increasing event we deduce from Corollary \ref{FKalt} that 
 \begin{equation}
  \bar \bP^{n,u}_{\gamma} \ldom    \bP^{n,u}_{\gamma}\quad  \text{ and }  \quad  \bP^{n,u}_{\gamma}\ldom  \bP^{n+1,u}_{\gamma},
  \end{equation}
  which is sufficient to conclude.

 \subsection{Stability of $n$-Small contours}\label{smallcproof}

We prove the stability directly on a larger interval for the parameter $u$, so that it can be used for both Propositions \ref{truta} and \ref{mejor}.
More precisely the main statement proved in this section is the following.
 
 \begin{proposition}\label{lesmall}
 For $\gb$ sufficiently large, we have: 
\begin{itemize}
  \item[(i)] Every positive $n$-small contour is $n$-stable for $u=u^{-}_{n+1}$.
\item[(ii)]  Every negative $n$-small contour is $n$-stable for $u=u^{+}_n$.
  \end{itemize}

 \end{proposition}
 This proposition combined with Lemma \ref{mozon} implies stability of positive and negative contours on the intervals $[u^{-}_{n+1},\infty)$ and 
 $(-\infty,u^{-}_{n}]$ respectively. 
 Both intervals include $[u^*_{n+1},u^*_n]$, which is sufficient 
 for  Proposition \ref{mejor} and \textit{a fortiori} for Proposition \ref{truta}.
 
 \begin{proof}
 Here and a few other instances, 
 we have to treat separately the cases where level zero is involved: $n=0$, $\gamma$ positive, and $n=1$, $\gamma$ negative.
 We need to show that when $\Diam(\gamma)\le \gb^2$  and $\gamma$ is positive we have
 \begin{equation}
  \frac{z^{u^-_1}_{1}(\gamma)}{\bar z^{u^-_1}_0(\gamma)}\le e^{|\tilde \gamma|}.
 \end{equation}
 Considering the contribution of the ground state $\phi\equiv 1$, we have $z^{u^-_1}_{0}(\gamma)\ge e^{h|\bar \gamma|}\ge  1$. 
On the other hand, setting $h^-_1=u^-_1+h_w(\gb)$,
 we have for $\gb$ sufficiently large (recall $J=e^{-2\gb}$)
 \begin{multline}\label{truxc}
 z^{u^-_1}_{1}(\gamma)\le \sum_{\phi\in \gO_{\bar \gamma}} e^{-\gb\cH^1_n(\phi)+h^-_1|\phi^{-1}(0)|}\\
 \le
 e^{h^-_1|\bar \gamma|} \cZ_{\bar \gamma,\gb}\le \left(  e^{h^-_1}(1+3J^2) \right)^{|\bar \gamma|}\le 
 e^{J |\bar \gamma|}\le  e^{J\gb^4}\le e.
\end{multline}
The third inequality is obtained by using \eqref{treck}.
The fourth and fifth inequality use \eqref{areabound} and are valid for  $\gb$  sufficiently large (note that $h^-_1$, like $h_w(\gb)$ is of order $J^2$).
Similar computations can be used to prove that  $z^{u^+_1}_{0}(\gamma)\le e $ and $z^{u^+_1}_{1}(\gamma)\ge 1$.

\medskip

\noindent In all other cases ($n\ge 1$, $\gamma$ positive and $n\ge 2$, $\gamma$ negative)
we can rewrite the ratio of partition function, using \eqref{dooz} from Lemma \ref{rwo}. We obtain
  \begin{equation}\label{topi}
  \frac{z^u_{n+\gep(\gamma)}(\gamma)}{ \bar z^u_n(\gamma)}
  = \frac{ \bE^{n+\gep(\gamma)}_{\gamma}\left[e^{u |\phi^{-1} (\bbZ_-)| -\bar H(\phi^{-1}(\bbZ_-))}\right]}{\bE^n_{\gb,\gamma}\left[e^{u |\phi^{-1}(\bbZ_-))-\bar H(\phi^{-1}(\bbZ_-))}\ind_{\{ \exists x\in \Delta^-_{\gamma},\phi(x)=n\}}\right]}.
  \end{equation}
  As $\bar H$ is positive (recall \eqref{morseu}),  using \eqref{areabound}, we see that the numerator of the r.h.s.\  satisfies 
  \begin{equation}\label{gonza}
\bE^{n+\gep(\gamma)}_{\gamma}\left[e^{u |\phi^{-1} (\bbZ_-)| -\bar H(\phi^{-1}(\bbZ_-))}\right] \le e^{u|\bar \gamma|}\le e^{u \gb^4n^4},
  \end{equation}
  where $n$-smallness of $\gamma$ is used in the last inequality.
  Considering $u$ being either equal to $u^-_{n+1}$ or $u^+_n$ (depending on the value of  $\gep(\gamma)$) we conclude that provided $\gb$ is sufficiently large 
  \begin{equation}\label{bip}
  \bE^{n+\gep(\gamma)}_{\gamma}\left[e^{u |\phi^{-1} (\bbZ_-)| -\bar H(\phi^{-1}(\bbZ_-))}\right]\le e.
  \end{equation}
  For the denominator,
 considering only the contribution of the event  
 $$\{\phi \in \gO[n,\gamma] \ : \ \forall x\in \bar \gamma,\ \phi(x)\ge 1\},$$
 we obtain that 
  \begin{multline}\label{conque}
 \bE^n_{\gamma}\left[e^{u |\phi^{-1}(\bbZ_-))-\bar H(\phi^{-1}(\bbZ_-))}\ind_{\{\exists x\in \Delta^-_{\gamma},\phi(x)=n\}}\right]\\
 \le 
 \bP^n_{\gamma}\left[\{\exists x\in \Delta^-_{\gamma},\phi(x)=n\}\cap \{\forall x\in \bar \gamma,\ \phi(x)\ge 1 \right].
 \end{multline}
Using \eqref{zups} we obtain that for any $x_0\in  \Delta^-_{\gamma}$
\begin{equation}\label{loot}
  \bP^n_{\gamma}\left[\exists x\in \Delta^-_{\gamma},\phi(x)=n \right]\ge   \bP^n_{\gamma}\left[\phi(x_0)=n\right]\ge 1-4e^{-4\gb}.
\end{equation}
Using \eqref{infin} we have
\begin{equation}\label{looz}
\bP^n_{\gamma} \left[ \exists  x\in \bar \gamma, \phi(x)\le 0 \right]
 \le \sum_{x\in \bar \gamma}  \bP^n_{\gamma} [ \phi(x)\le 0]\le 2|\bar \gamma| e^{-4n\gb}\le 4n^4\gb^4 e^{-4n\gb}.
 \end{equation}
Combining \eqref{conque},\eqref{loot} and \eqref{looz}, we obtain that for $\gb$ sufficiently large
\begin{equation}\label{bap}
 \bE^n_{\gamma}\left[e^{u |\phi^{-1}(\bbZ_-)|-\bar H(\phi^{-1}(\bbZ_-))}\ind_{\{\exists x\in \Delta^-_{\gamma},\phi(x)=n\}}\right]\ge e^{-1},
\end{equation}
and thus we conclude from \eqref{topi}, \eqref{bip} and \eqref{bap} that
$$\frac{z^u_{n+\gep(\gamma)}(\gamma)}{\bar z^u_n(\gamma)}\le e^2.$$ 

\end{proof}
 
 \subsection{Larger contours: presenting the induction}\label{prezent}
 
 To prove the stability of larger contour, we proceed with a double induction. A first induction based on the inclusion order for contours, 
 and a second one on the level $n$ for which stability is tested. To avoid any  confusion, before going into the details of the proof, 
 we provide the structure of this inductive reasoning.
For $n\ge 0$ and $\gamma$ a contour, we define the property
\begin{equation}
 \cP(n,\gamma):=\begin{cases} 
                 \big[ \text{ the contour } \gamma \text{ is stable at level }  n \text{ for } u\ge u^+_{n+1} \big] \quad &\text{ if } \gep(\gamma)=+,\\
                 \big[ \text{ the contour } \gamma \text{ is stable at level }  n \text{ for } u\le u^-_{n} \big] \quad &\text{ if } \gep(\gamma)=-,
                \end{cases}
 \end{equation}
and 
\begin{equation}
 \cP(\gamma):=\big[ \ \cP(n,\gamma)  \text{ is satisfied for all } n\ge 0 \big].
 \end{equation}
We are going to prove that $\cP(\gamma)$ holds for every contour using an  induction on $\gamma$:
From Proposition \ref{lesmall}, we know that $\cP(\gamma)$ hold true when $\Diam(\gamma)\le \gb^2$.
Thus we only need to perform the induction step, which is proving $\cP(\gamma)$ assuming that 
$\cP(\gamma')$ holds for all contours $\gamma'$ ``included in $\bar \gamma$'' \textit{i.e.}\ such that $\bar \gamma' \subset \bar \gamma$, $\bar \gamma'\ne \bar\gamma$.

\medskip

To prove $\cP(\gamma)$ itself we use an induction on $n$.
The direction of the induction depends on the sign of $\gamma$:
\begin{itemize}
 \item If $\gep(\gamma)=-$ then we prove $\cP(n,\gamma)$ assuming that $\cP(m,\gamma)$ holds for all $m\le n-1$, for all $n\ge 1$,
 \item If $\gep(\gamma)=+$ then we prove $\cP(n,\gamma)$ assuming that $\cP(m,\gamma)$ holds for all $m\ge n+1$, for all $n\ge 0$.
\end{itemize}
The descending induction for positive contours works for positive contours because we already know from Proposition \ref{lesmall} that $\cP(\gamma,n)$ holds for 
$n\ge \gb^{-1} \sqrt{\Diam(\gamma)}$. The ascending induction for negative contours is initiated for $n=1$ (for which the induction hypothesis 
``$\cP(m,\gamma)$ holds for all $m\le 0$'' is empty). 


\medskip

\noindent In the remainder of the proof we always assume that $n\ge 0$ for $\gamma$ positive and $n\ge 1$ for $\gamma$ negative.
For readability we also adopt the following convention within the proof
\begin{equation}\label{uvalu}
 u=u(n,\gamma):=\begin{cases}
                 u^+_{n+1}=200 J^{n+3} \quad & \text{ if } \gep(\gamma)=+,\\
                  u^-_{n}=\frac{1}{200} J^{n+2} \quad & \text{ if } \gep(\gamma)=-.
                \end{cases}
\end{equation}
By Lemma \ref{mozon} it is indeed sufficient to check stability for this value of $u$.
It turns out that the ratio $\bar z^{u}_{n+\gep(\gamma)}/\bar z^{u}_n$ is easier to work with than the quantity $z^{u}_{n+\gep(\gamma)}/\bar z^{u}_n$ which appears in the definition of $w^u_n$.
Thus our first task is to prove the following estimate.

\begin{lemma}\label{reduc}
If $\cP(n+1,\gamma)$ holds then we have, for $u$ defined in \eqref{uvalu} 
 \begin{equation}
 w^u_n(\gamma)\le 2e^{-\gb |\tilde \gamma|}\frac{\bar z^{u}_{n+\gep(\gamma)}(\gamma)}{\bar z^n_u(\gamma)}.
 \end{equation}
 \end{lemma}
A  consequence of this result is that to prove the $n$-stability of $\gamma$ we only need to show that 
\begin{equation}\label{core}
  \frac{\bar z^{u}_{n+\gep(\gamma)}(\gamma)}{\bar z^{u}_n(\gamma)}\le \frac{1}{2}e^{|\tilde \gamma|}.
  \end{equation}

  \begin{proof}[Proof of Lemma \ref{reduc}]
  Let us assume for simplicity that $\gamma$ is a positive contour (the adaptation for the negative case is straight-forward). 
From the definition \eqref{defoom} of $\gO^+[\gamma,n]$ and $\bar \gO^+[\gamma,n]$ we have
$$z^u_{n+1}(\gamma)=\bar z^u_{n+1}(\gamma)+e^{-\gb |\tilde \gamma|}  z^u_{n+2}(\gamma).$$
Thus using $\cP(n+1,\gamma)$ and the fact that $u_{n+1}^+\ge u^+_{n+2}$, we have
 \begin{equation}
w^u_n(\gamma)=e^{-\gb |\tilde \gamma|} \frac{\bar z^{u}_{n+1}}{\bar z^n_u}(\gamma) \left[ 1+ w^u_{n+1}(\gamma)\right]
\le \left(1+e^{-(\gb-1)|\tilde \gamma|}\right)e^{-\gb |\tilde \gamma|}\frac{\bar z^{u}_{n+1}}{\bar z^n_u}(\gamma).
 \end{equation}

\end{proof} 
 \medskip
 
 The strategy to prove \eqref{core} is to decompose $\bar z^u_{n+\gep(\gamma)}(\gamma)$ according to the set 
 of large external contours present in the field $\phi \in \bar \gO(n+\gep(\gamma),\gamma)$.
 Here and in the remainder of the proof, large means $n$-large (diameter larger than $\gb^2n^2$). Let us introduce some notation to perform this decomposition:
We let $\cK^{\larg}_{\ext}(\gamma,n)$  be the set of compatible collections of large external  contours and $\cK^{\sm}_{\ext}(\gamma,n)$ be  
 the set of  compatible  collections of small external contours (recall \eqref{daext})
 \begin{equation}\begin{split}\label{smallandbig}
 \cK^{\larg}_{\ext}(\gamma,n)&:= \{  \gG_1 \in \bar \cK_{\ext}(\gamma) \ : \ \forall \gamma_1\in \gG_1,\ \Diam(\gamma_1)>n^2\gb^2\},\\
 \cK^{\sm}_{\ext}(\gamma,n)&:= \{  \gG_2 \in \bar \cK_{\ext}(\gamma) \ : \ \forall \gamma_2\in \gG_2,\ \Diam(\gamma_2)\le n^2\gb^2\},
 \end{split}
 \end{equation}
  where $n^2\gb^2$ is replaced by $\gb^2$ when $n=0$.
 We let  $\cK^{\larg,+}_{\ext}(\gamma,n)$ and $\cK^{\sm,+}_{\ext}(\gamma,n)$ the subsets of   $\cK^{\larg}_{\ext}(\gamma,n)$ and  $\cK^{\sm}_{\ext}(\gamma,n)$ respectively 
 which contains only positive contours.

 \medskip
 
 For $\gG_1\in \bar \cK_{\ext}(\gamma)$, we let $\bar \gG_1$ and $L(\gG_1)$ denote respectively the the set of $\bbZ^2$ sites 
 enclosed by contours in $\gG_1$ and the total length 
 of the contours in $\gG_1$ 
 \begin{equation}\label{interiors}
 \bar \gG_1:=\bigcup_{\gamma_1\in \gG_1} \bar \gamma_1, \quad \text{ and }\quad L(\gG_1):=\sum_{\gamma_1\in \gG_1} |\tilde \gamma_1|.
 \end{equation}
Finally, we define  $Z^{u,\sm}_{m}[\gamma,\gG_1,n]$ which corresponds to a partition function  on the domain  
 $\bar \gamma \setminus \bar \gG_1$,
 which displays only $n$-small contours which are
 compatible with $\gG_1\cup \{\gamma\}$,
 \begin{equation}\label{defrelou}
Z^{u,\sm}_{m}[\gamma,\gG_1,n]:= \sum_{ \{\gG_2 \in \cK^{\sm}_{\ext}(\gamma,n) \ : \ \gG_2 \parallel  \gG_1\}}  \prod_{\gamma_2\in \gG_2 } e^{-\gb |\tilde\gamma_2|} 
z^u_{m+\gep(\gamma_2)}(\gamma_2).
\end{equation}
Note that when $m=0$, the contribution of $\gG_2$ is non-zero only if $\gG_2 \in \cK^{\sm,+}_{\ext}(\gamma,n)$.
The aim of our decomposition procedure is to prove the following result.

\begin{lemma}\label{decopoz}
Assuming that $\cP(\gamma')$ holds when $\bar \gamma' \subset \bar \gamma$,  $\bar \gamma'\ne \bar \gamma$,
we have,  for $u$ defined in \eqref{uvalu} 
  \begin{equation}\label{bluz}
   \frac{\bar z^u_{n+\gep(\gamma)}(\gamma)}{\bar z^u_{n}(\gamma)}\le \sum_{\gG_1\in \cK^{\larg}_{\ext}(\gamma,n)}  
   e^{-(\gb-2) L(\gG_1) } \frac{Z^{u,\sm}_{n+\gep(\gamma)}[\gamma,\gG_1,n]}{Z^{u,\sm}_{n}[\gamma,\gG_1,n]}.
 \end{equation}

\end{lemma}

To conclude the proof of the result we need two technical estimates to control the sum in the r.h.s.\ of \eqref{bluz}.
The first allows to bounds the ratio $(Z^{u,\sm}_{n+\gep(\gamma)}/ Z^{u,\sm}_{n})[\gamma,\gG_1,n]$ by a simpler quantity for which one can have a 
geometric intuition. It is proved in Section \ref{dcimo}.
 \begin{proposition}\label{difficilissimo}
 For $\gb$ sufficiently large, we have, for any $\gG_1\in \cK^{\larg}_{\ext}(\gamma,n)$ and $u$ defined in \eqref{uvalu} 
 \begin{equation}\label{cookries}
 \frac{Z^{u,\sm}_{n+\gep(\gamma)}[\gamma,\Gamma_1,n]}{Z^{u,\sm}_{n}[\gamma,\Gamma_1,n]}\le \frac{1}{2}
 \exp\left(-J^{3n+3}|\bar \gamma\setminus \bar \gG_1|+(L(\gG_1)+|\tilde \gamma|)\right).
 \end{equation}
 \end{proposition}
The second estimate which allows to conclude is a control of the simplified sum.
 We prove it in Section \ref{dcimus}.
 \begin{proposition}\label{dificilissimus}
We have for $\gb$ sufficiently large for every $n\ge 0$
 \begin{equation}\label{crookies}
 \sum_{\gG_1\in \cK^{\larg}_{\ext}(\gamma,n)}   \exp\left(-J^{3n+3}|\bar \gamma\setminus \bar \gG_1|-(\gb-3) L(\gG_1) \right)\le 1.
\end{equation} 
\end{proposition}
Combining \eqref{bluz}, \eqref{cookries}and \eqref{crookies} , we deduce that 
\begin{equation}
 \frac{\bar z^{u}_{n+\gep(\gamma)}(\gamma)}{\bar z^{u}_n(\gamma)}\le \frac{1}{2}e^{|\tilde \gamma|} 
 \sum_{\gG_1\in \cK^{\larg}_{\ext}(\gamma)}  e^{-J^{3n+3} |\bar \gamma\setminus \bar \Gamma_1|+(\gb-3)L(\gG_1)}\le \frac{1}{2} e^{|\tilde \gamma|},
\end{equation}
which  ends our proof by induction (cf. \eqref{core}).

 \subsection{Proof of Lemma \ref{decopoz}}
 
 We split our reasoning into two lemmas, one providing an upper bound on  $\bar z^{u}_{n+\gep(\gamma)}(\gamma)$ and the other providing a lower bound on 
 $\bar z^{u}_{n}(\gamma)$.
 
 \begin{lemma}\label{plif}
 Assuming that $\cP(\gamma')$ holds whenever $\bar \gamma'\subset \gamma$ and $\bar \gamma'\ne \bar \gamma$,
 for $u$ defined in \eqref{uvalu} we have, when $n+\gep(\gamma)\ge 1$
 \begin{equation}\label{floz}
 \bar z^{u}_{n+\gep(\gamma)}(\gamma)\le \sum_{\gG_1\in \cK^{\larg}_{\ext}(\gamma)} 
 \left( \prod_{\gamma_1\in \gG_1} e^{- (\gb-2) |\tilde \gamma_1|} z^u_{n}(\gamma_1)  \right)
Z^{u,\sm}_{n+\gep(\gamma)}[\gamma, \gG_1,n].
 \end{equation}
 Furthermore for $n=1$, $\gep(\gamma)=-1$ we have 
 \begin{equation}
 \bar z^{u}_{0}(\gamma)= \sum_{\gG_1\in \cK^{\larg,+}_{\ext}(\gamma)} 
 \left( \prod_{\gamma_1\in \gG_1} e^{- \gb |\tilde \gamma_1|} z^u_{n}(\gamma_1)  \right)
Z^{u,\sm}_{n+\gep(\gamma)}[\gamma, \gG_1,n].
\end{equation}
\end{lemma} 
To conclude the proof of \eqref{bluz}  we also need a lower bound for $\bar z^u_n(\gamma)$, which is provided by the following lemma.

\begin{lemma}\label{plof}
For any $\gG_1\in \cK^{\larg}_{\ext}(\gamma)$, we have
\begin{equation}\label{flaz}
  \bar z^{u}_{n}(\gamma)\ge \prod_{\gamma_1\in \gG_1} z^u_{n}(\gamma_1)  
  Z^{u,\sm}_{n}[\gamma, \gG_1,n].
\end{equation}
\end{lemma}
The inequality \eqref{bluz} is obtained combining \eqref{floz} and \eqref{flaz}.

 \begin{proof}[Proof of Lemma \ref{plif}]
We assume for notational simplicity that $\gamma$ is a positive contour.
 Recalling Equation \eqref{tapun} and splitting the set $\gG'$ of external contour between large ($\gG_1$) and small ($\gG_2$) contours we obtain
 
 \begin{multline}
 \bar z^{u}_{n+1}(\gamma)=  \sum_{\gG'\in \bar \cK_{\ext}(\gamma)}\prod_{\gamma'\in \gG'}  e^{-\gb |\tilde \gamma'|} z^h_{n+1+\gep(\gamma')}(\gamma') \\
 =
 \!\!\!\!\!\sum_{\gG_1\in \cK^{\larg}_{\ext}(\gamma,n)} 
 \prod_{\gamma_1\in \gG_1} e^{-\gb |\tilde \gamma_1|} z^u_{n+1+\gep(\gamma_1)}(\gamma_1) \!\!\!\!\!\!\!\!
 \sum_{\{ \gG_2 \in \cK^{\sm}_{\ext}(\gamma,n) \ : \ \gG_2 \parallel \gG_1\}}  \prod_{\gamma_2\in \gG_2 } e^{-\gb |\tilde\gamma_2|} z^u_{n+\gep(\gamma_2)}(\gamma_2)\\
 =  \prod_{\gamma_1\in \gG_1} e^{-\gb |\tilde \gamma_1|} z^u_{n+1+\gep(\gamma_1)}(\gamma_1) Z^{u,\sm}_{n}[\gamma, \gG_1].
 \end{multline} 
 To conclude, we need to check that 
 \begin{equation}
 z^u_{n+1+\gep(\gamma_1)}(\gamma_1)\le e^{2|\tilde \gamma_1|}z^u_{n}(\gamma_1).
 \end{equation}
This is of course obvious when $\gep(\gamma_1)=-1$. For positive contours on the other hand we have
 \begin{multline}\label{lezinc}
 z^u_{n+2}(\gamma_1)= e^{\gb | \tilde \gamma_1|}w^u_{n+1}(\gamma_1) \bar z^u_{n+1}(\gamma_1)\\
 \le
 e^{\gb | \tilde \gamma_1|}w^u_{n+1}(\gamma_1) z^u_{n+1}(\gamma_1)= e^{2\gb | \tilde \gamma_1|}w^u_{n+1}(\gamma_1) w^u_{n}(\gamma_1) z^u_{n}(\gamma_1).
 \end{multline}
 Using the induction hypothesis, or more precisely $\cP(n+1,\gamma_1)$ (recall that $u\ge u^+_{n+2}$) and $\cP(n,\gamma_1)$, we deduce from \eqref{lezinc} that
\begin{equation}
 z^u_{n+2}(\gamma_1)\le e^{2|\tilde \gamma_1|}z^u_{n}(\gamma_1).
 \end{equation}
The same proof goes when $\gep(\gamma)=-1$, if we restrict the sum to the set of positive contours in the special case $n=1$.
\end{proof}

\begin{proof}[Proof of Lemma \ref{plof}]

Instead of proving \eqref{flaz}, we prove a stricter inequality where the contours in $\gG_2$ are not required to be small, and which is valid 
for all $\gG_1\in \bar \cK_{\ext}(\gamma)$ (recall \eqref{daext})
\begin{equation}\label{fliz}
  \bar z^{u}_{n}(\gamma)\ge \prod_{\gamma_1\in \gG_1} z^u_{n}(\gamma_1)  \sum_{\{ \gG_2 \in \bar \cK_{\ext}(\gamma) \ : \ \gG_1 \parallel \gG_2\}} 
  \prod_{\gamma_2\in \gG_2 } e^{-\gb |\tilde\gamma_2|} z^u_{n+\gep(\gamma_2)}(\gamma_2),
  \end{equation}
 (in the case $n=0$ only $\gG_2$ with all contour positive give a contribution to the sum).
 We shall show that the l.h.s. in \eqref{fliz} corresponds to the contribution to the sum \eqref{partisioux} of 
the set of $\phi$s whose external contours are either in $\gG_1$ or compatible with $\gG_1$
$$\cA:=\left\{ \phi\in \bar \gO^+[\gamma,n] \ : \ \forall \gamma' \in \Upsilon^{\ext}_n(\phi),\ \gamma' \mid \gG_1 \text{ or } \gamma'\in \gG_1\right\}.$$
To make our decomposition we use the notation 
\begin{equation}\begin{split}
\gG_{1,1}(\gamma_1)&=\{ \gamma_{1,1}\in \Upsilon^{\ext}_n(\phi)\setminus \{\gamma_1\} \ : \ \bar \gamma_{1,1}\subset \bar \gamma_1\},\\
\gG_2&=\Upsilon^{\ext}_n(\phi) \setminus \bigcup_{\gamma_1\in \gG_1}\gG_{1,1}(\gamma_1).
\end{split}\end{equation}
Note that for $\phi\in \cA$ we have $\gG_{1,1}(\gamma_1)\subset \bar \cK_{\ext}(\gamma_1)$.
In analogy with \eqref{tapun}, we can thus write
  \begin{multline}
  \sum_{\phi\in \cA} e^{-\gb \cH^n(\phi)+h|\phi^{-1}(0)|}\\= 
\prod_{\gamma_1\in \gG_1}\left[ e^{-\gb|\tilde \gamma_{1}|} z^u_{n+\gep(\gamma_1)}(\gamma_1)+ \sum_{\gG_{1,1}\in \bar \cK_{\ext}(\gamma_1)} 
\prod_{\gamma_{1,1}\in \gG_{1,1}} e^{-\gb|\tilde \gamma_{1,1}|}z^u_{n+\gep(\gamma_{1,1})} \right]\\
\times\sum_{\{ \gG_2 \in \bar\cK_{\ext}(\gamma) \ : \ \gG_1 \parallel \gG_2\}}  \prod_{\gamma_2\in \gG_2 } e^{-\gb |\tilde\gamma_2|} z^u_{n+\gep(\gamma_2)}(\gamma_2).
\end{multline}
In each factor of the product over $\gamma_1$, the first term corresponds to the contribution of $\phi$s 
for which $\gamma_1$ is a contour.
Finally recalling Equations \eqref{relatz} and \eqref{tapun} and 
\begin{multline}
e^{-\gb|\tilde \gamma_{1}|} z^u_{n+\gep(\gamma_1)}(\gamma_1)+ \sum_{\gG_1\in\bar\cK^{\ext}(\gamma_1)} \prod_{\gamma_{1,1}\in \gG_1} e^{-\gb|\tilde \gamma_{1,1}|}
z^u_{n+\gep(\gamma_{1,1})}\\
=
 e^{-\gb|\tilde \gamma_{1}|} z^u_{n+\gep(\gamma_1)}(\gamma_1)+ \bar z^u_n(\gamma_1)= z^{u}_n(\gamma_1),
 \end{multline}
 which yields \eqref{fliz}.
 \end{proof}

\subsection{Proof of Propoposition \ref{difficilissimo}}\label{dcimo}

To prove the inequality \eqref{cookries}, we prove separately bounds for the numerator and for the denominator.
As for Proposition \ref{lesmall} 
we have to treat separately the cases $n=1$, $\gamma$ negative and $n=0$, $\gamma$ positive, which we do in Lemma \ref{dzip}. The general case is dealt with using Lemma 
\ref{dalo}.
The proof of these two results is technically involved, and for that reason, postponed to Section \ref{secrestrict},

\begin{lemma}\label{dzip}
There exists a constant $C$ (independent of $\gb$) such that for all $\gb$ sufficiently large for every every $\gamma$ with $\Diam(\gamma)> \gb^2$, and every
$\gG_1\in \cK^{\larg}_{\ext}(\gamma,1)$,
$h\in \left[0,2J^2\right]$, we have
 \begin{equation}\begin{split}\label{lezink}
\log  Z^{u,\sm}_{0}[\gamma,\Gamma_1,1]&\le|\bar \gamma \setminus \bar \Gamma_1|\left(h+J^2+2J^3+C J^4\right),\\
\log  Z^{u,\sm}_{0}[\gamma,\Gamma_1,1]&\ge |\bar \gamma \setminus \bar  \Gamma_1|\left(h+2J^2+2J^3-CJ^4\right)+\frac{1}{2}(|\tilde \gamma|+L(\gG_1)),\\
 \log  Z^{u,\sm}_{1}[\gamma,\Gamma_1,1]&\le|\bar \gamma \setminus \bar  \Gamma_1|\left(2J^2+4J^3-C J^4\right),\\
 \log  Z^{u,\sm}_{1} [\gamma,\Gamma_1,1]&\ge |\bar \gamma \setminus \bar \Gamma_1|  \left(2J^2+4J^3-CJ^4\right) -\frac{1}{2}(|\tilde \gamma|+L(\gG_1)).
\end{split} \end{equation}

\end{lemma}

\begin{rem}
Note that the inequalities of \eqref{lezink} also hold if $[\gamma,\Gamma_1,1]$  is replaced by $[\gamma,\Gamma_1,0]$ because the associated notions of small contour are the same.
The range we have chosen for $h$ is sufficient to treat the case of $u=u^{\pm}_1$ for $\gb$ sufficiently large as for these value we have $h=J^2+O(J^3)$.
\end{rem}

To treat the other cases, we define $\gO^+[m,\gamma,\gG_1,n]$ to be the set of trajectories naturally associated with the partition function
$Z^{u,\sm}_m[\gamma,\gG_1,n]$. 
We define first $\cC[\gamma,\gG_1,n]$ the set of contour which can appear in $\Upsilon_m(\phi)$
\begin{equation}
 \cC[\gamma,\gG_1,n]:=\{ \gamma_2\in \cC \ : \ \bar \gamma_2\subset \bar \gamma\setminus \bar \gG_1,\ \Diam(\gamma_2)\le (n\gb)^2 \text{ and }
 \gamma_2 \mid \left(\{\gamma\} \cup \gG_1\right) \}   
\end{equation}
We set 
\begin{multline}
\gO[m,\gamma,\gG_1,n]:=\Big\{ \phi, \bar \gamma \setminus \bar \gG_1 \to \bbZ \ : \  
 \Upsilon^{\ext}_m(\phi)\subset  \cC[\gamma,\gG_1,n]
   \\
  \text{ and } \phi= m+ \sum_{\hat \gamma_2 \in \hat \Upsilon_m(\phi)} \varphi_{\hat \gamma_2}. \Big\}.
\end{multline}
and as usual 
\begin{equation}
\gO^+[m,\gamma,\gG_1,n]:= \Big\{ \phi\in \gO[m,\gamma,\gG_1,n] \ : \ \forall x\in  \bar \gamma \setminus \bar \gG_1, \phi(x)\ge 0 \Big\}.
\end{equation}
The condition $\phi= m+ \sum_{\hat \gamma_2 \in \hat \Upsilon_m(\phi)} \varphi_{\hat \gamma_2}$ corresponds to \eqref{cylinder}, and is violated when $\phi$ presents some level lines which which surrounds hole in $\bar \gamma \setminus \bar \gG_1$.
With this definition, the reader can check that (recall our convention $h=u+h_w(\gb)$)
\begin{equation}\label{winouze}
Z^{u,\sm}_{m}[\gamma,\Gamma_1,n]= \sum_{\phi \in \gO^+[m,\gamma,\gG_1,n]}e^{-\gb \cH^m_{\bar \gamma \setminus \bar \gG_1}(\phi)+h|\phi^{-1}\{0\}|}.
\end{equation}
We let $\bP^{m,\sm,n}_{\gamma,\gG_1}$ be the SOS measure restricted to $\gO[m,\gamma,\gG_1,n]$
\begin{equation}
\bP^{n,\sm}_{\gamma,\gG_1}(\phi):=\frac{1}{Z^{\sm}[\gamma,\gG_1,n]}e^{-\gb \cH^m_{\bar \gamma \setminus \bar \gG_1}(\phi)} 
\end{equation}
where
\begin{equation}
 Z^{\sm}[\gamma,\gG_1,n]:= \sum_{\phi \in \gO[m,\gamma,\gG_1,n]} e^{-\gb \cH^m_{\bar \gamma \setminus \bar \gG_1}(\phi)},
\end{equation}
(again by translation invariance, the partition function does not depend on the boundary condition $m$).
We state a result which is similar to Lemma \ref{rwo} and is useful in our proofs.

\begin{lemma}\label{zilak}
 For any $m\ge 1$ and any $\gamma\in \cC$  and $\gG_1\in \cK^{\larg}_{\ext}(\gamma)$
 \begin{equation}\label{groot}
  Z^{u,\sm}_{m}[\gamma,\Gamma_1,n]= Z^{\sm}[\gamma,\gG_1,n] \bE^{m,\sm,n}_{\gamma,\gG_1}\left[ e^{u |\phi^{-1}(\bbZ_-)|- \bar H(\phi^{-1})(\bbZ_-)}\right].
\end{equation}
\end{lemma}
\begin{proof}
 We have to show that 
\begin{equation}
  Z^{u,\sm}_{m}[\gamma,\Gamma_1,n]= \sum_{\phi \in \gO^+[m,\gamma,\gG_1,n]}e^{-\gb \cH^m_{\bar \gamma \setminus \bar \gG_1}(\phi)+u|\phi^{-1}(\bbZ_-)|-\bar H(\phi^{-1}(\bbZ_-)).}
\end{equation}
The proof can be adapted from that of \cite[Lemma 3.1]{cf:part1}: The sum over all the possible options for 
the negative parts of $\phi$ cancels the term $\bar H$ and changes $u$ into $h$ so that one recovers \eqref{winouze}. 
The key observation to check that the proof adapts is that
the contour restriction does not bring any constraint on the choice of $\phi_-=\max(0,-\phi)$ once $\phi^{-1}(\bbZ^-)$ is fixed. 
This is the case
because the contour restriction forces the diameter of maximal connected components of $\phi^{-1}(\bbZ_-)$ 
are smaller than $(\gb n)^2$.
\end{proof}

As a consequence of Lemma \ref{zilak},
when neither $n$ nor $n+\gep(\gamma)$ are zero, the $\log$ of the estimated ratio can be rewritten in the following form
\begin{multline}\label{crounch}
 \log \left(\frac{  Z^{u,\sm}_{n+\gep(\gamma)}[\gamma,\Gamma_1,n]}{  Z^{u,\sm}_{n}[\gamma,\Gamma_1,n]}\right)\le 
  \log  \bE^{n+\gep(\gamma),\sm,n}_{\gamma,\gG_1}\left[ e^{u |\phi^{-1}(\bbZ_-)|- \bar H(\phi^{-1})(\bbZ^-)} \right] \\
  - \log 
   \bE^{n+\gep(\gamma),\sm,n}_{\gamma,\gG_1}\left[ e^{u |\phi^{-1}(\bbZ_-)|- \bar H(\phi^{-1})(\bbZ^-)} \right].
\end{multline}
We need the following statements 
\begin{lemma}\label{dalo}
The following estimates hold:
\begin{itemize}
 \item [(i)] For positive $\gamma$ and $u$ as in \eqref{uvalu}
 \begin{equation}
  \log \bE^{n+1,\sm,n}_{\gamma,\gG_1}\left[ e^{u |\phi^{-1}(\bbZ_-)|- \bar H(\phi^{-1})(\bbZ^-)} \right]\le  |\bar \gamma \setminus \bar \Gamma_1| 2u J^{2n+2} .
 \end{equation}
  \item[(ii)]For arbitrary $\gamma$ and $u$ as in \eqref{uvalu}
 \begin{multline}\label{ckonveu}
  \log \bE^{n,\sm, n}_{\gamma,\gG_1}\left[ e^{u |\phi^{-1}(\bbZ_-)|- \bar H(\phi^{-1})(\bbZ^-)} \right]\\
  \ge |\bar \gamma \setminus \bar \Gamma_1| 
  \left( \frac{1}{2}u J^{2n}- 40 J^{3n+3}\right)-\frac{1}{4}\left( |\tilde \gamma|+L(\gG_1) \right). 
 \end{multline}
   \item[(iii)] For negative  $\gamma$ and $u$ as in \eqref{uvalu}
   \begin{multline}
   \log \bE^{n-1,\sm, n}_{\gamma,\gG_1}\left[ e^{u |\phi^{-1}(\bbZ_-)|- \bar H(\phi^{-1})(\bbZ^-)} \right]\\
   \le |\bar \gamma \setminus \bar \Gamma_1| 
   \left(4u J^{2(n-1)}-\frac{1}{4}J^{3n}\right)+\frac{1}{4}\left( |\tilde \gamma|+L(\gG_1) \right). 
   \end{multline}   
  \end{itemize}
\end{lemma}

\begin{proof}[Proof of Proposition \ref{difficilissimo}].
We start with the case of positive contour with $n=0$.
Using Lemma \ref{dzip} we have 
\begin{equation}
 \log \left(\frac{  Z^{u,\sm}_{1}[\gamma,\Gamma_1,1]}{  Z^{u,\sm}_{0}[\gamma,\Gamma_1,1]}\right)
 \le |\bar \gamma \setminus \bar \gG_1|(J^2+2J^3+2CJ^4-h)+ \frac 1 2\left(|\tilde \gamma|+L(\gG_1)\right). 
\end{equation}
Now recall  (cf. \eqref{uvalu}) that we are interested in the case
\begin{equation}
h=h_w(\gb)+u^+_1= \log\left(\frac{ e^{4\gb} }{e^{4\gb}-1}\right)+200J^3,
\end{equation}
hence we obtain
\begin{equation}
 \log \left(\frac{  Z^{u,\sm}_{1}[\gamma,\Gamma_1,1]}{  Z^{u,\sm}_{0}[\gamma,\Gamma_1,1]}\right)
 \le -|\bar \gamma \setminus \bar \gG_1| J^3+\frac 1 2 \left(|\tilde \gamma|+L(\gG_1)\right). 
\end{equation}
We let the reader check that similarly for negative contours and $u=u^-_1$ we have 
\begin{equation}
 \log \left(\frac{  Z^{u,\sm}_{0}[\gamma,\Gamma_1,1]}{  Z^{u,\sm}_{1}[\gamma,\Gamma_1,1]}\right)
 \le -|\bar \gamma \setminus \bar \gG_1| J^3+\frac{1}{2}\left(|\tilde \gamma|+L(\gG_1)\right). 
\end{equation}
Let us now treat the case of positive contour for $n\ge 1$.
Using \eqref{crounch} and Lemma \ref{dalo}, we have for $u=u^+_{n+1}$ (recall\eqref{uvalu})
\begin{multline}
 \log \left(\frac{  Z^{u,\sm}_{n+1}[\gamma,\Gamma_1,n]}{  Z^{u,\sm}_{n}[\gamma,\Gamma_1,n]}\right)
 \le  |\bar \gamma \setminus \Gamma_1|\left[ u\left( 2J^{2n+2}- \frac{1}{2}J^{2n}\right)+40J^{3n+3} \right]+ \frac{1}{2}(|\tilde \gamma|+L(\gG_1))\\
 \le |\bar \gamma \setminus \Gamma_1|(40J^{3n+3}-\frac{1}{4}u J^{2n})+ \frac{1}{2}(|\tilde \gamma|+L(\gG_1)).
 \end{multline}
 Now recalling \eqref{uvalu}, 
we obtain the result by observing that 
 \begin{equation}
40J^{3n+3}-\frac{1}{4}u^+_{n+1} J^{2n}\le -J^{3n+3}.
\end{equation}
In a similar manner in the case of negative contour and $n\ge 2$ we have as a consequence of Lemma \ref{dalo} (ii)-(iii), for $\gb$ sufficiently large
\begin{equation}
  \log \left(\frac{  Z^{u,\sm}_{n-1}[\gamma,\Gamma_1,n]}{  Z^{u,\sm}_{n}[\gamma,\Gamma_1,n]}\right)\\
  \le
  |\bar \gamma \setminus \Gamma_1| \left(4u J^{2(n-1)}-\frac{1}{5}J^{3n}\right)+ \frac{1}{2}(|\tilde \gamma|+L(\gG_1))
\end{equation}
and we conclude by observing that 
 \begin{equation}
4u^-_n J^{2(n-1)}-\frac{1}{5}J^{3n}\le -J^{3n+3},
\end{equation}
so that \eqref{cookries} is satisfied in all cases.
 \end{proof}
  
  \medskip

\subsection{Proof of Proposition \ref{dificilissimus}}\label{dcimus}
 
We can relax for this proof the notion of compatibility, meaning we consider the sum over a superset of $\cK^{\larg}_{\ext}(\gamma,n)$.
We consider in this section only that two contours are  externally compatible if $\bar \gamma \cap \bar \gamma'=\emptyset$.
Adding a factor $2$ to take the sign into account (that is replacing each factor $e^{-(\gb-3)|\tilde \gamma_1|}$ by $2e^{-(\gb-3)|\tilde \gamma_1|}$ in \eqref{crookies}), 
we choose to consider geometric contours instead of signed contours (and use $\bar \gamma_1$ to denote the interior of $\tilde \gamma_1$).

\medskip

Our proof works by induction and leads us to consider sets of external contour in a general domain $\gL\subset \bbZ^2$ 
 which are not necessarily simply connected. 
We use in this section only the notation $\gG$ and $\gamma$ instead of $\gG_1$ and $\gamma_1$.
We maintain that all contours must satisfy $\bar \gamma\subset \gL$ and thus cannot surround holes.

\medskip

We let $\tilde \cK^{\larg}_{\ext}(\gL,n)$ denote the set of   collections of externally compatible $n$-large geometric contours
with the above mentioned notion of compatibility.
The result \eqref{crookies} will follow (provided that $e^{-4(\gb/2)}\ge 2e^{-4(\gb-3)}$) if we can prove that for every $\gL\subset \bbZ^2$

\begin{equation}\label{wookies}
 \sum_{\tilde \gG\in \tilde \cK^{\larg}_{\ext}(\gL)}  e^{-J^{3n+3} |\gL\setminus \bar \gG|-(\gb/2) L(\tilde \Gamma)}\le 1,
 \end{equation}
 where $L(\tilde \gG)$ and $\bar \gG$ are the length and perimeter associated with $\tilde \gG$ 
 defined in analogy with  \eqref{interiors}.
We prove a more general version of the statement.
\begin{proposition}
For any finite domain in $\gL$ and any $\ell\ge 4$ we have
\begin{equation}\label{dunnies}
\sum_{\tilde \gG\in \tilde \cK^{\ell+}_{\ext}(\gL)} e^{-2^{-\ell} |\gL\setminus \bar \gG| }e^{-(\gb/2) L(\tilde\gG)}\le 1.
\end{equation}
where  $\tilde \cK^{\ell+}_{\ext}(\gL)$ denote the set of  collections of externally compatible geometric  contours with length larger than $2\ell$.
 \end{proposition}
If we apply this proposition for $\ell=n^2\gb^2$, \eqref{dunnies} implies \eqref{wookies} provided that 
$2^{-n^2\gb^2}\le J^{3n+3}$, which is valid for every $n\ge 1$ provided that $\gb$ is sufficiently large (we have also $2^{-\gb^2}\le J^{3}$ to cover the case $n=0$).

 \begin{proof}
We prove the result by induction on $\ell$. We let  $\tilde \cK^{\ell}_{\ext}(\gL)$ be
the set of  collection   externally compatible of geometric contours with length equal to $2\ell$.
The key step is proving that 
\begin{equation}\label{industaps}
\sum_{\tilde \gG\in \tilde \cK^{\ell}_{\ext}(\gL)} e^{-2^{-\ell} |\gL\setminus \bar \gG|-(\gb/2) L(\tilde\gG)}\le  e^{-2^{-(\ell+1)} |\gL\setminus \bar \gG|}.
\end{equation}
Let us show how \eqref{dunnies} is deduced from \eqref{industaps}.
First let us observe that \eqref{dunnies} is obviously satisfied when $\ell$ is larger than the total number of edges in $\gL$. Hence we can proceed 
by descending induction, assuming that the 
statement is valid for $\ell+1$ and proving it for $\ell$.

\medskip

Obviously $\tilde \gG \in \tilde \cK^{\ell+}_{\ext}(\gL)$ can be written in the form $\tilde \gG_1\cup \tilde \gG_2$ where $\tilde \gG_1\in \tilde \cK^{(\ell+1)+}_{\ext}(\gL)$
and $\tilde \gG_2\in  \tilde\cK^{\ell}_{\ext}(\gL \setminus \bar \gG_1)$.
Hence 
\begin{multline}
\sum_{\tilde \gG\in \tilde \cK^{\ell+}_{\ext}(\gL)} e^{-2^{-\ell} |\gL\setminus \bar \gG| }e^{-(\gb/2) L(\tilde\gG_1)}\\
 =\sum_{\tilde \gG_1\in \tilde \cK^{(\ell+1)+}_{\ext}(\gL)} e^{-(\gb/2) L(\tilde\gG_1)}\sum_{\tilde \gG_2\in \tilde \cK^{\ell}_{\ext}(\gL\setminus \bar \gG_1)} 
 e^{-2^{-\ell} |(\gL\setminus \bar \gG_1)\setminus \bar \gG_2|} e^{-(\gb/2) L(\tilde\gG_2)}
 \\ \le  \sum_{\tilde \gG_1\in \tilde \cK^{(\ell+1)+}_{\ext}(\gL)} e^{-(\gb/2) |\tilde\gG|} e^{-2^{-(\ell+1)} |\gL\setminus \bar \gG_1|} \le 1,
\end{multline}
where in the first inequality  uses \eqref{industaps} for the domain $\gL\setminus \bar \gG_1$ and in the second one the induction hypothesis.

\medskip

Let us now prove \eqref{industaps}.
We have to distinguish between two sorts of contributions, according to the number of contours.
Let us first consider the contribution where the number of contour is smaller than $m:=\lfloor |\gL|\ell^{-2} \rfloor$.
Keeping in mind that, from isoperimetrical inequalities,  a contour  encloses at most  $(\ell/2)^2$ sites, we have 
$$|\gL\setminus \bar \gG|\ge |\gL|-(\ell/2)^2m \ge   \frac{3}{4}|\gL|$$
and hence 
\begin{equation}
\sum_{\{ \tilde \gG\in \tilde \cK^{\ell}_{\ext}(\gL)\ : \ \# \tilde \gG \le m \}} e^{-2^{-\ell} |\gL\setminus \bar \gG|-(\gb/2) L(\tilde\gG)}
\le \exp\left(-\frac{3}{2} 2^{-(\ell+1)}|\gL|\right) \sum_{\tilde \gG\in \tilde \cK^{\ell}_{\ext}(\gL)}e^{-(\gb/2)L(\tilde\gG)}.
\end{equation}
Now for each site $x\in  \bbZ^2$, we let $\tilde\cC^{\ell}(x,\gL)$ denote the set of geometric contour longer than $2\ell$ such that 
$\bar \gamma \subset \gL$,  for which $x$ is the 
smallest vertex in $\bar \gamma$ for the lexicographical order. We write $\tilde\cC^{\ell}(x)$ for the set corresponding to $\gL=\bbZ^2$.
We have
\begin{equation}
 \sum_{\tilde \gG\in \tilde \cK^{\ell}_{\ext}(\gL)}e^{-(\gb/2) L(\tilde\gG)}
 \le \prod_{x\in \gL} \left(1+\sum_{\tilde \gamma \in \tilde\cC^{\ell}(x,\gL)} e^{-\gb\ell}\right) \le 
 \left(1+\#\tilde\cC^{\ell}({\bf 0}) e^{-\gb\ell}\right)^{|\gL|}.
\end{equation}
where the first inequality is obtained by summing over all collections of contours instead instead of externally compatible ones and 
the second one by extending the sum to $\tilde\cC^{\ell}(x)$ and using translation invariance.
Using the fact that, by a classic counting argument, we have that for $\gb$ sufficiently large 
$$\#\tilde\cC^{\ell}({\bf 0}) e^{-\gb\ell}\le 9^{\ell} e^{-\gb\ell}\le e^{-\gb/2},$$
which implies in particular that, provided $\gb$ is sufficiently large,
\begin{multline}\label{zapi}
 \sum_{\{ \tilde \gG\in \tilde \cK^{\ell}_{\ext}(\gL)\ : \ \# \tilde \gG \le m \}} e^{-2^{-\ell} |\gL\setminus \bar \gG|}e^{-(\gb/2) L(\tilde\gG)}
\\ \le \exp\left(|\gL|\left(e^{-(\gb /2) \ell}-\frac{3}{2} 2^{-(\ell+1)}\right) \right)\le \exp\left(-\frac{5}{4}2^{-(\ell+1)} |\gL|\right).
\end{multline}
Now concerning the contribution of collections of cardinality larger than $m$, we neglect the penalty for uncovered area 
 \begin{equation}
 \sum_{\{ \tilde \gG\in \tilde \cK^{\ell}_{\ext}(\gL)\ : \ \# \tilde \gG > m \}} e^{-2^{-\ell} |\gL\setminus \bar \gG|}e^{-(\gb/2) L(\tilde\gG)} 
 \le \sum_{\{ \tilde \gG\in \tilde \cK^{\ell}_{\ext}(\gL)\ : \ \# \tilde \gG > m \}} e^{-(\gb/2) L(\tilde\gG)}.
 \end{equation}
 
To estimate the sum in the r.h.s.\, we consider that if $\# \tilde \gG=k$ then to select $k$ contours, we must first chose
$k$ vertices to be the minimal (for the lexicographical order) vertices enclosed by each contour 
(there are $\binom{|\gL|}{k}\le \left(\frac{e |\gL|}{k}\right)^k$ ways to do this)
and then ignoring further compatibility condition, consider that for each vertex, there are at most $9^{\ell}$ eligible contours of length $\ell$.
Thus we obtain that for $\gb$ sufficiently large we have
 
 \begin{multline} \label{zapa}
 \sum_{\{ \tilde \gG\in \tilde \cK^{\ell}_{\ext}(\gL)\ : \ \# \tilde \gG > m \}} e^{-(\gb/2) L(\tilde\gG)}\le
  \sum_{k> m} e^{-\gb\ell}  \binom{|\gL|}{k} 9^{k\ell}
  \le   \sum_{k> m} e^{-\frac{\gb k\ell}{2}} \left(\frac{e |\gL|}{k}\right)^k\\
  \le   \sum_{k> m}\left( e^{-\frac{\gb \ell}{2}+1} \ell^2 \right)^k
  \le  \sum_{k> m} e^{-\frac{\gb k\ell}{4}}\le 2 e^{-\frac{\gb \ell(m+1)}{4}} \le 2e^{-\frac{\gb}{4}\min \left( \ell, |\gL|/\ell \right)}.
 \end{multline}
Overall combining \eqref{zapi} and \eqref{zapa} we obtain that
\begin{multline}
\sum_{\tilde \gG\in \tilde \cK^{\ell}_{\ext}(\gL)}e^{-(\gb/2) L(\tilde\gG)-2^{-\ell} |\gL\setminus \bar \gG|}\\
\le
\exp\left(-\frac{5}{4}2^{-(\ell+1)} |\gL|\right)\ind_{\{m\ge 1\}}
+ 2e^{-\frac{\gb}{4}\min \left( \ell, |\gL|/\ell \right)} \le e^{-2^{-(\ell+1)} |\gL|},
\end{multline}
where to check the last inequality we have to check separately the cases $m\ge 1$ and $m=0$. 
 \end{proof}

\subsection{The large contour case for Proposition \ref{mejor}}\label{mejorplus}

We conclude this section by extending the stability result on a larger interval.
This can be treated in a relatively simple fashion by induction if we rely on cluster expansion estimates (Lemma \ref{finalfrontier}).

\medskip

We only have to check the stability for large contour since that of small contour has been checked in Proposition \ref{lesmall}. 
The proof works using the same induction as for Proposition \ref{truta}
We define the property
\begin{equation}
 \cP(n,\gamma):=\begin{cases} 
                 \big[ \text{ the contour } \gamma \text{ is stable at level }  n \text{ for } u\ge u^*_{n+1} \big] \quad &\text{ if } \gep(\gamma)=+,\\
                 \big[ \text{ the contour } \gamma \text{ is stable at level }  n \text{ for } u\le u^*_{n} \big] \quad &\text{ if } \gep(\gamma)=-.\\
                \end{cases}
 \end{equation}
and 
\begin{equation}
 \cP(\gamma):=\big[ \ \cP(n,\gamma)  \text{ is satisfied for all } n\ge 0 \big].
 \end{equation}
As for Proposition \ref{truta},
we need to prove $\cP(n,\gamma)$ assuming $\cP(n+\gep(\gamma),\gamma)$ and $\cP(\gamma')$ for $\bar \gamma'\subset \bar\gamma$, $\bar \gamma'\ne \bar \gamma$.
After fixing $n$ and $\gamma$ we assume below that 

\begin{equation}\label{redevu}
 u=u(n,\gamma):= \begin{cases} u^*_{n+1} \quad &\text{ if } \gep(\gamma)=+,\\
 u^*_{n} \quad &\text{ if } \gep(\gamma)=-.
                 \end{cases}
\end{equation}
Using Lemma \ref{reduc} (or rather its proof) we can reduce ourselves to proving
\begin{equation}
 \frac{\bar z^u_{n+\gep(\gamma)}}{\bar z^u_n}(\gamma)\le \frac{1}{2}e^{|\tilde\gamma|}.
\end{equation}
By induction hypothesis, all contours involved in the partition function $\bar z^u_{n+\gep(\gamma)}(\gamma)$, $\bar z^u_n(\gamma)$
are stable for $u$ as in \eqref{redevu}.
Recalling the definition for the truncated potential \eqref{trunk} and the definitions of Section \ref{fffa}, 
this stability implies that 
\begin{equation}
\bar z^u_{n+\gep(\gamma)}(\gamma)=Z\left[\gamma, w^{u,\tr}_{n+\gep(\gamma)}\right] \text{ and } \bar z^u_{n}(\gamma)=Z\left[\gamma, w^{u,\tr}_{n}\right].
\end{equation}
As the truncated potentials satisfy \eqref{usualchoice}, Lemma \ref{finalfrontier} allows to deduce that
\begin{equation}
\begin{split}
|\log \bar z^u_{n+\gep(\gamma)}(\gamma)- |\bar \gamma| \tf^{\tr}_{n+\gep(\gamma)}(\gb,u) |&\le |\tilde \gamma|/4,\\
|\log \bar z^u_{n}(\gamma)- |\bar \gamma| \tf^{\tr}_{n}(\gb,u) |&\le |\tilde \gamma|/4.
\end{split}
\end{equation}
Using the definition of $u^*_{n+1}$ or $u^*_n$ (depending on the sign of the contour)
we have  $\tf^{\tr}_{n+\gep(\gamma)}(\gb,u)=\tf^{\tr}_{n}(\gb,u)$ and thus we can deduce that 
\begin{equation}
\frac{ \bar z^u_{n+\gep(\gamma)}}{\bar z^u_{n}}(\gamma)\le \exp\left(|\tilde \gamma|/2\right).
\end{equation}

\section{Estimates for restricted partition functions: the proof of Lemma \ref{dzip} and \ref{dalo}}\label{secrestrict}

In this section we prove the two remaining technical lemmas from Section \ref{dcimo}, and fully complete the proof of Theorem \ref{converteo}.

\subsection{Proof of Lemma \ref{dzip}}
We first prove  upper bound results which are easier. The idea is to write a contour decomposition and to relax the compatibility assumption in the sum to obtain an upper bound.
Let us start with the case of zero boundary condition.
Recall that 
\begin{equation}\label{sumi}
 Z^{u,\sm}_{0}[\gamma,\gG_1,1]:= \sum_{\{ \gG_2 \in \cK^{\sm,+}_{\ext}(\gamma,1) \ : \ \gG_2 \parallel  \gG_1\}}   e^{h|\bar \gamma \setminus \bar \gG_1\cup \bar \gG_2|} \prod_{\gamma_2\in \gG_2 } 
 e^{-\gb |\tilde\gamma_2|} 
z^u_{1}(\gamma_2)
\end{equation}
where  (recall \eqref{smallandbig}) 
$$ \cK^{\sm,+}_{\ext}(\gamma,1)=\{ \gG_2\in  \cK^{\sm}_{\ext}(\gamma,1) \ : \ \forall \gamma\in \gG_2, \gep(\gamma_2)=+\}.$$

\medskip

To obtain an upper bound, we replace the first exponential term by $e^{h|\bar \gamma \setminus \bar \gG_1|}$, and we let the sum range over all collections of small
positive $1$-small contours 
$\cC^{\sm,+,1}_{\bar \gamma\setminus \bar \gG_1}$ without imposing any compatibility condition. 
The sum factorizes and we obtain

\begin{equation}
 Z^{u,\sm}_{0}[\gamma,\gG_1,1]\le 
   e^{h|\bar \gamma \setminus \bar \gG_1|} \prod_{\gamma_2 \in \cC^{\sm,+,1}_{\bar \gamma\setminus \bar \gG_1}}  
\left(1+  e^{-\gb |\tilde\gamma_2|} 
z^u_{1}(\gamma_2) \right).
\end{equation}
Now given $x\in \bbZ^2$, we let $\cC^{\sm,+,1}_x$ be the set of positive small contours $\gamma_2$ for which $x$ is the minimal point in $\bar \gamma_2$ for the lexicographical order.
We have 
\begin{multline}
 \sum_{ \gamma_2 \in \cC^{\sm,+,1}_{\bar \gamma\setminus \bar \gG_1}} \log \left(1+  e^{-\gb |\tilde\gamma_2|} 
z^u_{1}(\gamma_2)\right)\le \sum_{x\in |\bar \gamma\setminus \bar \gG_1|} \sum_{\gamma_2  \in \cC^{\sm,+,1}_x} \log \left(1+  e^{-\gb |\tilde\gamma_2|} 
z^u_{1}(\gamma_2)\right)\\
=|\bar \gamma\setminus \bar \gG_1|  \sum_{\gamma_2  \in \cC^{\sm,+,1}_{\bf 0}} \log \left(1+  e^{-\gb |\tilde\gamma_2|} 
z^u_{1}(\gamma_2)\right),
\end{multline}
because the right hand side includes all the term of the left hand side, plus a few extra contours that are not contained in $\bar \gamma\setminus \bar \gG_1$.
We observe that
\begin{equation}\label{night}
 \begin{cases}
  z^u_{1}(\gamma_2)&=(1-J^2)^{-1} \quad \text{ when } \bar \gamma_2=\{x\},\\
   z^u_{1}(\gamma_2)&=(1-J^3)^{-1}(1+J^2)(1-J^2) \quad \text{ when } \bar \gamma_2=\{x,y\} \text{ with } x\sim y,\\
   z^u_{1}(\gamma_2)&\le e \quad \text{ for other small contours }.
 \end{cases}
\end{equation}
For the two first line, the full computation is performed in \cite[Proof of Lemma 3.2]{cf:part1} and the last one can be derived like \eqref{truxc}.
Noting that in $\cC^{\sm,+,1}_{\bf 0}$ there is one contour of length $4$ and two of length $6$, we obtained that
\begin{equation}
\sum_{\gamma_2  \in \cC^{\sm,+,1}_x}\log \left(1+  e^{-\gb |\tilde\gamma_2|} 
z^u_{1}(\gamma_2)\right)\le J^2+2J^3+  C J^4,
\end{equation}
where the term $CJ^4$ includes  the contribution of all contours of length $8$ or more.

\medskip

\noindent For the case of boundary condition equal to one we observe similarly that 
\begin{equation}
 Z^{u,\sm}_{1}[\gamma,\gG_1,1]\le 
\prod_{\gamma_2 \in \cC^{\sm,1}_{\bar \gamma\setminus \bar \gG_1}}  
\left(1+  e^{-\gb |\tilde\gamma_2|} 
z^u_{1+\gep(\gamma_2)}(\gamma_2) \right),
\end{equation}
where $\cC^{\sm,1}_{\bar \gamma\setminus \bar \gG_1}$ denote the set of $1$-small contours  in $\bar \gamma\setminus \bar \gG_1$. 
 Similarly to \eqref{night}, one can check that (recall $h\le 2J^2$)
\begin{equation}\label{night2}
 \begin{cases}
  z^u_{1+\gep(\gamma_2)}(\gamma_2)\le 1+CJ^2 \quad &\text{ when } L(\tilde\gamma_2)\le 6,\\
   z^u_{1}(\gamma_2)\le e \quad &\text{ for other small contours },
 \end{cases}
\end{equation}
and the results follows as for zero boundary condition case.

\medskip

\noindent We obtain the lower bound results by restricting our sums  to contours of length $4$ and $6$.
We set 
\begin{equation}\label{csta}
\cC^{*,+}(\gamma,\gG_1):=\{ \gamma_2 \ : \  \gep(\gamma_2)=+, \ |\tilde \gamma_2|\le 6,\ \bar \gamma_2\subset \bar \gamma \setminus \bar \gG_1,\ 
d(\gamma_2, \{\gamma\}\cup \gG_1)>0 \},
\end{equation}
where $d(\gamma_2, \{\gamma\}\cap \gG_1)$ denotes the minimal distance between the geometric contours $\tilde \gamma_2$ and those in the set $\{\gamma\}\cup \gG_1$. This condition ensures in particular compatibility. We set 
$$v(\gamma_2):=e^{-\gb |\tilde\gamma_2|-h|\bar \gamma_2|} 
z^u_{1}(\gamma_2).$$
we obtain 
\begin{multline}\label{drac}
 Z^{u,\sm}_{0}[\gamma,\gG_1,1]
 \ge e^{h|\bar \gamma \setminus \bar \gG_1|} \sum_{\{ \gG_2 \in \cK^{\sm,+,1}_{\ext}(\gamma) \ : \ \gG_2 \parallel  \gG_1\}}    \prod_{\gamma_2\in \gG_2 } 
v(\gamma_2)\ind_{\{\gamma_2\in \cC^{*,+}(\gamma,\gG_1)\}}\\
=:e^{h|\bar \gamma \setminus \bar \gG_1|}
 Z[v,\cC^{*,+}(\gamma,\gG_1)].
 \end{multline}
 We can apply cluster expansion results for the relation of external compatibility with the weights being given by $v$ (see Remark \ref{laremarc}).
 We set for $\bC\subset  \cC^{*,+}(\gamma,\gG_1)$
 \begin{equation}\label{modpot}
 v^{R}(\bC):=\sum_{\bB\subset \bC} (-1)^{|\bB|+|\bC|} \log Z[v,\bB]
  \end{equation} 
where $Z[v,\bB]$ is defined as $Z[v,\cC^{*,+}(\gamma,\gG_1)]$ in \eqref{drac} but with the indicator $\ind_{\{\gamma_2\in \bB\}}$.
We let $\cR^+(\gamma,\gG_1)$ denotes the set of clusters associated with external compatibility in   $\cC^{*,+}(\gamma,\gG_1)$
that is $\bC\subset \cC^{*,+}(\gamma,\gG_1)$ is in  $\cR^+(\bar \gamma\setminus\bar \gG_1)$ if (recall \ref{noncluster})

\begin{equation}
\forall \bB \subset \bC, \ \bB\nparallel \bC\setminus \bB.
\end{equation}
where  $\bB\nparallel \bC$ is the negation of  $\bB\parallel \bC$.
After observing that $v^{R}(\bC)=0$ for $\bC\notin \cR^+(\bar \gamma\setminus\bar \gG_1)$
we obtain
 \begin{multline}\label{dox}
 \log Z[v,\cC^{*,+}\gamma,\gG_1)]= \sum_{\{ \bC\in \cR^+(\bar \gamma\setminus\bar \gG_1)\}} v^R(\bC)
 \\
 \ge \sum_{\{ \bC\in \cR^+(\gamma, \gG_1) \ : \ L(\bC)\le 6\}} v^R(\bC)
 -\sum_{\{ \bC\in \cR^+(\gamma, \gG_1) \ : \ L(\bC)\ge 8\}} \left|v^R(\bC)\right|.
 \end{multline}
Note that the clusters in the first sum in \eqref{dox} consist in only one contour. Using \eqref{night} again, and the assumption $h\le J^2$, the reader can check that for some appropriate constant $C$ we haves
\begin{equation}\label{speel}
v^R(\{\gamma_2\})=\log\left( 1+  v(\gamma_2)\right)\ge \begin{cases} J^2-C J^4, \quad \text{ if } \bar \gamma_2=\{x\},\\
J^3-C J^4, \quad  \text{ if }  \bar \gamma_2=\{x,y\} \text{ with } x\sim y.
\end{cases}
\end{equation}
We let the reader check that number of contours of length $4$ and $6$ in $\cC^{*,+}(\gamma,\gG_1)$ satisfies 
\begin{equation}
 \begin{split}
         \#\{ \gamma_2\in    \cC^{*,+}(\gamma,\gG_1) \ : \ |\tilde \gamma_2|=4\}&\ge  |\bar \gamma\setminus\bar \gG_1|-(|\tilde \gamma|+L(\gG_1)),\\    
        \#\{ \gamma_2\in    \cC^{*,+}(\gamma,\gG_1) \ : \ |\tilde \gamma_2|=6\}&\ge  2|\bar \gamma\setminus\bar \gG_1|-6(|\tilde \gamma|+L(\gG_1)).            
                \end{split}
\end{equation}
the second term being caused by boundary effects.
Hence we have for $\gb$ sufficiently large
\begin{equation}\label{wick}
 \sum_{\{ \bC\in \cR^+(\bar \gamma,\bar \gG_1) \ :  \ L(\bC)\le 6\}} v^R(\bC)
 \ge |\bar \gamma\setminus\bar \gG_1|(J^2+2J^3-CJ^4)- \frac{1}{2}(L(\gG_1)+|\bar \gamma|),
\end{equation}
where the second term is present to account for the fact that the number of contour is not exactly proportional to the volume $|\bar \gamma\setminus\bar \gG_1|$.

\medskip

\noindent To control the second term in \eqref{dox}, we use Theorem \ref{superclust} for external compatibility with 
$$a(\gamma_2)=|\tilde \gamma_2| \text{ and } d(\gamma_2)=(\gb-5)|\tilde \gamma_2|.$$
In that case \eqref{condit} holds and one deduces from \eqref{dabound} that
\begin{equation}
\sum_{\{ \bC\in \cR^+(\bar \gamma,\bar \gG_1) \ : \ L(\bC)\ge 8\}}  \left|v^R(\bC)\right|\le CJ^4|\bar \gamma\setminus\bar \gG_1|,
  \end{equation}
 which together with \eqref{wick} yields 
 \begin{equation}
 \log  Z[v,\cC^{*,+}(\gamma,\gG_1)]\ge |\bar \gamma\setminus\bar \gG_1|(J^2+2J^3-C'J^4)- \frac{1}{2}(L(\gG_1)+|\bar \gamma|),
 \end{equation}
and allows to conclude (recall \eqref{drac}).
 
\medskip

\noindent For the case with boundary condition equal to one, we have
\begin{equation}\label{sumikr}
 Z^{u,\sm}_{1}[\gamma,\gG_1,1]:= \sum_{\{ \gG_2 \in \cK^{\sm}_{\ext}(\gamma,1) \ : \ \gG_1 \parallel  \gG_2\}}   
\prod_{\gamma_2\in \gG_2 } 
 e^{-\gb |\tilde\gamma_2|} 
z^u_{1+\gep(\gamma_2)}(\gamma_2).
\end{equation}
The same argument as the one used at level $0$ allows to obtain a lower bound.
Restricting the sum to $\gamma_2\in \cC^{*}(\gamma,\gG_1)$ where 
$$\cC^{*}(\gamma,\gG_1):=\{ \gamma_2 \ : \  \ |\tilde \gamma_2|\le 6,\ \bar \gamma_2\subset \bar \gamma \setminus \bar \gG_1,\ 
d(\gamma_2, \{\gamma\}\cup \gG_1)>0 \}$$
and setting 
$$v'(\gamma_2):=e^{-\gb |\tilde\gamma_2|} 
z^u_{1+\gep(\gamma_2)}(\gamma_2),$$
we have
\begin{equation}
 Z^{u,\sm}_{1}[\gamma,\gG_1,1]
 \ge  \sum_{\{ \gG_2 \in \cK^{\sm,+,1}_{\ext}(\gamma) \ : \ \gG_2 \parallel  \gG_1\}}    \prod_{\gamma_2\in \gG_2 } 
v'(\gamma_2)\ind_{\{\gamma_2\in \cC^*(\gamma,\gG_1)\}}.
\end{equation}
We can then check (using the fact that $h\le 2J^2$) that \eqref{speel} is satisfied for $(v')^R$ defined as in Equation \ref{modpot} 
and conclude in a similar manner (the coefficient are multiplied by two because $\cC^{*}(\gamma,\gG_1)$ contains twice as many contours).

\medskip

\qed

\subsection{Proof of Lemma \ref{dalo}}

We start our proof with the lower bound which are easier to achieve since they are a consequence of Jensen's inequality. 
The upper bound results require a more delicate analysis and are treated afterwards.

\subsubsection{Lower bounds, the proof of $(ii)$}

From  Jensen's inequality, we have
\begin{equation}
  \bE^{n,\sm,n}_{\gamma,\gG_1}\left[e^{u|\phi^{-1}(\bbZ_-)|-\bar H(\phi^{-1}(\bbZ_-))}\right]
  \ge  \bE^{n,\sm,n}_{\gamma,\gG_1}\left[u|\phi^{-1}(\bbZ_-)|-\bar H(\phi^{-1}(\bbZ_-))\right],
\end{equation}
and thus Equation \eqref{ckonveu} 
can be deduced from the following inequality
\begin{multline}
 \bE^{n,\sm,n}_{\gamma,\gG_1}\left[u|\phi^{-1}(\bbZ_-)|-\bar H(\phi^{-1}(\bbZ_-))\right]\\
 \ge 
 |\bar \gamma \setminus \bar \gG_1|\left(\frac{1}{2}u J^{2n}-40J^{3n+3}\right)-\frac{1}{2}(|\tilde \gamma|+L(\gG_1)).
 \end{multline}
We are going to bound terms separately and show that 
  \begin{equation}\label{lixiout}\begin{split}
 \bE^{n,\sm,n}_{\gamma,\gG_1}\left[|\phi^{-1}(\bbZ_-)| \right]&\ge 
 \frac{1}{2}J^{2n}\left( |\bar \gamma\setminus\bar \gG_1|-3(|\tilde \gamma|+L(\gG_1))\right),\\
 \bE^{n,\sm,n}_{\gamma,\gG_1}\left[\bar H(\phi^{-1}(\bbZ_-)) \right]&\le 40 J^{3n+3}|\bar \gamma\setminus\bar \gG_1|.
  \end{split}\end{equation}
The first inequality can directly be deduced from \eqref{infin}: the equality is valid as soon as $x$ is not constrained by the boundary condition,
which might happen only if $x$ lies in $\Delta^-_{\gamma}$ or in $\Delta^+_{\gamma_1}$ for $\gamma_1\in \gG_1$,
hence the number of $x$ such that \eqref{infin} does not apply is proportional to $|\tilde \gamma|+L(\gG_1)$.

\medskip

Concerning the second line in \eqref{lixiout}, we let  $f_{i}(\phi)$ denote the number 
of points which lie in a connected component of $\phi^{-1}(0)$ of size $i$ or larger,
\begin{equation}
 f_{i}(\phi):=\{ x\in \bar\gamma \setminus \bar \gG_1 \ :
 \exists \gL\subset \bar\gamma \setminus \bar \gG_1,\ \gL \text{ connected };\  x\in \gL \ ; \ \forall z\in \gL, \ \phi(z)\le 0  \}.
\end{equation}
Using Lemma \ref{zoomats}, (we observe  \eqref{sitmats} implies that $\bar H \{ x,y\} \le J^3$ for $\gb$ large ) we obtain that
\begin{equation}\label{lamine}
 \bar H(\phi^{-1}(\bbZ_-))\le \frac{J^{3}}{2}f_2(\phi)+ 2 J^{2}f_{3}(\phi).
\end{equation}
From Proposition \ref{rouxrou},  we have
\begin{equation}
 \bE^{n,\sm,n}_{\gamma,\gG_1}[f_2(\phi)]\le |\bar \gamma \setminus \bar \gG_1|4\max_{\{x,y \in \bar \gamma \setminus \bar \gG_1 \ : \  x\sim y\}}\bE^{n,\sm,n}_{\gamma,\gG_1}[ \max(\phi(x),\phi(y))\le 0]
 \le 8 J^{3n}.
\end{equation}
Similarly we obtain also using Proposition \ref{rouxrou}
\begin{multline}\label{diatta}
 \bE^{n,\sm,n}_{\gamma,\gG_1}[f_3(\phi)]\le 18
  |\bar \gamma \setminus \bar\gG_1| \max_{\{x\sim y \sim z\}} \bE^{n,\sm,n}_{\gamma,\gG_1}[ \max(\phi(x),\phi(y),\phi(z))\le 0]\\
  \le 36n J^{4n}\le 36 J^{3n+1}.
 \end{multline}
 The coefficient $18$ corresponds to the number of ways of choosing a connected set of size three which contains a given $x$.
Combining \eqref{lamine}-\eqref{diatta} we conclude that the second inequality in \eqref{lixiout} holds.

\qed

\subsubsection{Upper bounds-Proof of $(i)$ and $(iii)$}

The upper bound is a bit more delicate since the proof relies on some decorrelation property of the measure $\bP^{m,\sm,n}_{\gamma,\gG_1}$.
We are going to use the following technical statement, which ensures that the bounds from 
Proposition \ref{rouxrou} remain valid after conditioning to the realization of field outside a large ball.
For the rest of the proof one sets $r=3n^2\gb^2$.

\begin{lemma}\label{locall}
We have for any $m\ge 1$, for all of $\phi$, and any $x\in\bar \gamma \setminus \bar \gG_1|$, we have 
   \begin{equation}\label{leff}
   \bP^{m,\sm,n}_{\gamma,\gG_1}[\phi(x)\le 0\  | \ \phi(z),\ |z-x|\ge r]\le 2 J^{2m}.
   \end{equation}
   If moreover 
   \begin{equation}\label{badcond}
   \{x,y\}\cap \bigg[\Delta^-_{\gamma}\cup \Big(\bigcup_{\gamma_1\in \gG_1} \Delta^+_{\gamma_1}\Big)\bigg]=\emptyset,
   \end{equation}
   then
   \begin{multline}\label{leffe}
      \bP^{m,\sm,n}_{\gamma,\gG_1}\big[\min(\phi(x),\phi(y))\le 0  \ ; \   \forall w\in \partial \{x,y\}, \ \phi(w)>0   \\
      | \ \phi(z),\ |z-x|\ge r\big]\ge \frac{1}{2} J^{3m}.
\end{multline}
\end{lemma}
The condition \eqref{badcond} is present to ensure that both $\phi(x)$ and $\phi(y)$ are allowed to take negative value under
$\bP^{m,\sm,n}_{\gamma,\gG_1}$. The condition is sufficient but not always necessary since only positive contours might forbid to have $\phi(x)\le m$.

Let us now prove $(i)$, which we choose to replace by a slightly more general statement. For $\gb$ sufficiently large one, for every $n\ge 1$,
$m\ge \max(n-1,1)$ and $v\le J^n$ we show that
\begin{equation}\label{lecarlat}
   \log \bE^{m,\sm,n}_{\gamma,\gG_1}\left[ e^{v |\phi^{-1}(\bbZ_-)|} \right] \le |\bar \gamma\setminus \bar \gG_1|3 v J^{2m}.
\end{equation}
Recall that $\delta_x:=\ind_{\{\phi_x\le 0\}}$.
In order to control the effect of correlation we choose to split $\gL$
according to the value of $\lfloor x/r \rfloor$ modulo $2$, by setting for $i \in \lint 1, 4\rint$, $z\in\bbZ^2$
\begin{equation}
\label{ztook}
\cB_i(z):=\{ x\in \bar \gamma \setminus \bar \gG_1 \ : \ \lfloor x_j/r \rfloor=\alpha_j(i)+2z_j,\   j=1,2\ \},
\end{equation}
where $\alpha_j(i)$ is the $j$-th digit in the dyadic development if $i-1$.
We have from H\"older's inequality
\begin{equation}\label{zolder}
\left( \bE^{m,\sm,n}_{\gamma,\gG_1}\left[ e^{v |\phi^{-1}(\bbZ_-)|} \right]\right)^4\le  \sum_{i=1}^4
\bE^{m,\sm,n}_{\gamma,\gG_1}\left[ e^{4v \sum_{z\in \bbZ}\sum_{x\in \cB_i(z)} \delta_x} \right].
\end{equation}
Using the fact that $e^x\le 1+e^K x$ for $x\in[0,K]$ and  Lemma \ref{locall}, we obtain that for each $z\in \bbZ$
\begin{multline}
\bE^{m,\sm,n}_{\gamma,\gG_1}\left[ e^{4v \sum_{x\in \cB_i(z)} \delta_x} \ | \ \phi(y), y\in \bigcup_{z'\ne z}\cB_i(z') \right]\\
\le  1+ 4ve^{4vr^2}\bE^{m,\sm,n}_{\gamma,\gG_1}\left[ \sum_{x\in \cB_i(z)} \delta_x \ | \ \phi(y), y\in \bigcup_{z'\ne z}\cB_i(z') \right]\\
\le 1+ 8ve^{4vr^2} J^{2m}|\cB_i(z)|. 
\end{multline}
Using this inequality iteratively and combining it with \eqref{zolder} we obtain that 
\begin{equation}
    \log \bE^{m,\sm,n}_{\gamma,\gG_1}\left[ e^{v |\phi^{-1}(\bbZ_-)|} \right]\le 2ve^{4vr^2} J^{2m}|\bar \gamma\setminus \bar \gG_1|.
\end{equation}
The inequality \eqref{lecarlat} follows provided $e^{4vr^2}\le 3/2$, which is the case under our assumption provided $\gb$ is sufficiently large.

\medskip

Let us now turn to the more delicate case $(iii)$.
Using Cauchy Schwartz inequality  we have 
\begin{multline}
 \bE^{n-1,\sm}_{\gamma,\gG_1}\left[ e^{u |\phi^{-1}(\bbZ_-)|-\bar H(\phi^{-1}(\bbZ_-))} \right]^2\\
 \le 
 \bE^{n-1,\sm}_{\gamma,\gG_1}\left[ e^{2u |\phi^{-1}(\bbZ_-)|} \right] \bE^{n-1,\sm}_{\gamma,\gG_1}\left[ e^{-2\bar H(\phi^{-1}(\bbZ_-))} \right].
 \end{multline}
 To evaluate the first term we can rely on \eqref{lecarlat} and conclude that 
 \begin{equation}
  \log \bE^{n-1,\sm}_{\gamma,\gG_1}\left[ e^{2u |\phi^{-1}(\bbZ_-)|} \right] \le 6 u J^{2(n-1)}.
 \end{equation}
The final step is to prove that 
\begin{equation}
  \log \bE^{n-1,\sm}_{\gamma,\gG_1}\left[ e^{-\bar H(\phi^{-1}(\bbZ_-))} \right] \le -
  \frac{1}{4} J^{3n}\left(|\bar \gamma\setminus\bar \gG_1|-10\left(|\tilde \gamma|+L(\gG_1)\right)\right).
\end{equation}
As $\bar H$ is positive, this yields the same upper bound for the expectation of $e^{-2\bar H(\phi^{-1}(\bbZ_-))}$.
We set for $\gL\subset\bbZ^2$ finite 
 \begin{equation}
 G^{k}(\gL):= \log  \bE^+_{\gL,\gb} \left[e^{-\bar H(\phi^{-1}[k,\infty))}\right].
 \end{equation}
Conditioning to the level set at level $n-1$ we obtain
\begin{multline}\label{rouge}
   \bE^{n-1,\sm,n}_{\gamma,\gG_1}\left[ e^{-\bar H(\phi^{-1}(\bbZ_-))} \right]
   \\=\bE^{n-1,\sm,n}_{\gamma,\gG_1}\left[ e^{-G^{n-2}(\phi^{-1}(-\infty,n-2])} \right]=
   \bE^{1,\sm,n}_{\gamma,\gG_1}\left[ e^{G^{n-2}(\phi^{-1}(\bbZ_-))} \right].
  \end{multline}
 We refer to \cite[Equation (6.7)]{cf:part1} for details concerning the first equality, the second one is simply using vertical translation invariance.
 Note that as this is the case for $\bar H$ (recall \eqref{maxdecomp}), $G^k(\gL)=\sum_I  G^k(\gL_i)$ where $(\gL_i)_{i\in I}$ is the decomposition
 into maximal connected components of $\gL$. 
 Note that we have, according to \cite[Lemma 6.2]{cf:part1}, for $\gb$ sufficiently large, for any pair of neighbors $\{x,y\}$ in $\bbZ^2$
  \begin{equation}
    G^{n-2}(\{x,y\})=\log \left( 1- \frac{J^3-J^4}{1-J^4}J^{3(n-2)} \right)\le -\frac{1}{2}J^{3n-3}.
  \end{equation}
  Hence setting
\begin{equation}
 \delta_2(x):=\ind\{\phi(x)\le 0 \text{ and the connected component of $x$ in $\phi^{-1}(\bbZ^-)$ has size two}  \},
\end{equation}
and ignoring the contribution to $G^{n-2}$ of connected components of larger size (for singletons, not that $G^{k}\{x\}=0$) we have 
\begin{equation}
   G^{n-2}(\phi^{-1}(\bbZ_-))\le -\frac{1}{8}J^{3n-3} \sum_{x\in \bar \gamma \setminus \bar \gG_1}\delta_2(x).
\end{equation}
As we are going to use Lemma \eqref{locall}, we are going to consider the sum restricted to 
\begin{equation}
 \gL(\gamma,\gL_1):=\left\{ x\in \bar \gamma \setminus \bar \gG_1 \ : \ 
 d\left(x, \Delta^-_{\gamma}\cup \Big(\cup_{\gamma_1\in \gG_1} \Delta^+_{\gamma_1}\Big)\right)\ge 2\right\}.
\end{equation}
Note that 
\begin{equation}
 \gL(\gamma,\gL_1)\ge |\bar \gamma \setminus \bar \gG_1|- 10 \left( |\tilde \gamma| + L(\gG_1) \right).
 \end{equation}
To conclude we thus need to prove that 
\begin{equation}\label{zoyer}
   \log  \bE^{1,\sm,n}_{\gamma,\gG_1}\left[ e^{-\frac{1}{2}J^{3n-3} \sum_{x\in \gL}\delta_2(x)} \right]\le \frac{1}{4}J^{3n}|\gL|.
\end{equation}

Now we choose to proceed as in the proof of \eqref{lecarlat} to deal with the correlation
between the variables $\delta_2(x)$.
We set with the same notation as in \eqref{ztookit} for $i\in \lint 1, 4 \rint$, $z\in \bbZ^2$
\begin{equation}\label{ztookit}
\cB_i(z):=\{ x\in \gL \ : \ \lfloor x_j/r \rfloor=\alpha_j(i)+2z_j,\   j=1,2\ \}.
\end{equation}
Note that from Lemma \ref{locall} (note that to have $\delta_2(x)$ there are four ways to chose a neighbor of $x$ for which the field is negative) we have 
\begin{equation}\label{wiou}
 \bE^{1,\sm,n}_{\gamma,\gG_1}\left[ \sum_{x\in \cB_i(z)} \delta_2(x) \ | \ \phi(y), \ y \in \bigcup_{z'\ne z}\cB_i(z') \right]\ge 2 J^3.
\end{equation}
Like for \eqref{zolder} we have from H\"older's inequalities 
\begin{equation}\label{miax}
  \left( \bE^{1,\sm,n}_{\gamma,\gG_1}\left[ e^{ -\frac{1}{4}J^{3n-3} \sum_{x\in \gL}\delta_2(x)} \right]\right)^4
  \le \prod_{i=1}^4 \bE^{1,\sm,n}_{\gamma,\gG_1}\left[ e^{-J^{3n-3} \sum_{z\in \bbZ}\sum_{x\in \cB_i(z)} \delta_2(x)} \right].
\end{equation}
Using $e^{-x}\le 1-e^{-K}x$ for $x\in[0,K]$ and \eqref{wiou}, we obtain that  for any $z$ and $i$
\begin{multline}
  \bE^{1,\sm,n}_{\gamma,\gG_1}\left[ e^{-J^{3n-3} \sum_{x\in \cB_i(z)} \delta_2(x)} \ | \ \phi(y), \ y \in \bigcup_{z'\ne z}\cB_i(z') \right]\\
  \le 1-e^{-J^{3n-3}r^2}J^{3n-3}\bE^{1,\sm,n}_{\gamma,\gG_1}\left[ \sum_{x\in \cB_i(z)} \delta_2(x) \ | \ \phi(y), \ y \in \bigcup_{z'\ne z}\cB_i(z') \right]
  \\ \le 1- J^{3n}|\cB_i(z)|.
\end{multline}
Using this inequality to evaluate each term factor in the l.h.s of  \eqref{miax} 
\begin{equation}
\prod_{i=1}^4 \bE^{1,\sm,n}_{\gamma,\gG_1}\left[ e^{-J^{3n-3} \sum_{z\in \bbZ}\sum_{x\in \cB_i(z)} \delta_2(x)} \right]\le
\left(1-J^{3n}\right)^{|\gL|}.
\end{equation}
which in view of \eqref{miax} concludes the proof of \eqref{zoyer}. \qed

\begin{proof}[Proof of Lemma \ref{locall}]
We are going to show that the inequalities are valid when a stronger conditioning is considered.
We set 
$$\cC[r,x]:=\left\{ \gamma \in \cC[\gamma,\gG_1,m] \ : \ \bar \gamma\cap B(x,r/2)\ne \emptyset \right\}$$
where $B(z,l)$
the Euclidean ball of center $z$ radius $l$. We let $\hat \cC[r,x]$ be the corresponding set of cylinders.

\medskip

Conditioned to  $\hat\Upsilon_{m}(\phi)\cap (\hat \cC[r,x])^{\cc}$ (the set of cylinders of $\phi$ which do not intersect $B(x,r/2)$), which, due to smallness of contours,
is a stronger conditioning than 
$(\phi(z),\ |z-x|\ge r)$,
the distribution of 
$\phi$ restricted to $B(x,r/2)$
is of the type $\bP_{\bL,\gL,\gb}$ described at the beginning of Section \ref{peaksec}, where 
$$\bL:=\{\gamma_3 \in \cC[r,x] \ : \ \forall \gamma_2\in  \hat\Upsilon_{m}(\phi)\cap (\hat \cC[r,x])^{\cc},\  \gamma_3 \mid \gamma_2  \}.$$
The results follows then by applying Proposition \ref{rouxrou}.

\end{proof}

\section{First consequences of the cluster expansion convergence}\label{dalast}

In this section, we exploit contour stability to obtain 
the convergence of Gibbs measure $\bP^{n,h}_{\gL,\gb}$ for appropriate values of $h$ in Section \ref{moots}.
We also exploit the decay of correlation to prove the regularity of the free energy in Section \ref{laprov}.
Finally we close the proof of Theorem \ref{main} in Section \ref{lapriv} where we prove Proposition \ref{convirj}.

\subsection{Convergence of the Gibbs measure}\label{moots}

In view of the results of Section \ref{clusexp}, the stability of the contour implies the convergence of some measure when the size of the box grows.
We need however some work to convert this result into a convergence result for the SOS measure.
The proof relies on the following observation: the distribution of the set of external contour $\Upsilon^{\ext}_n(\phi)$
under $\bP^{n,h}_{\gL,\gb}$ is the same that the distribution of external contour for the measure $\bbP^{w^{h}_{n}}_{\gL}$
(this is simply a by-product of the proof of Proposition \ref{trax}, recall Equation \ref{tapouz}).
The second observation is that, conditioned to the set of external contours, the distribution of the field inside each contour is independent from the 
rest and with an explicit distribution.

\begin{proposition}\label{conviark}
Let $f$ and $g$ be local functions $\gO_{\infty}\to [0,1]$ with respective supports $A$ and $B$.
 For $n\ge 0$ and $h\in [h^*_{n+1},h^*_n]$, we have for any simply connected $\gL,\gL'\subset \bbZ^2$ which satisfies  $A\subset \gL\cap \gL'$ 
 and  $d(A,\gL)\le d(A,\gL')$, there exist a  positive constant
 $C$ such that for all $\gb$ sufficiently large 
 \begin{equation}\label{clox}
\left| \bP^{n,h}_{\gL,\gb}[f(\phi) ]-\bP^{n,h}_{\gL',\gb}[f(\phi)]\right|\le C|A|e^{\frac{\gb}{100}d( A,\gL)}.
\end{equation}
As a consequence  the sequence of measure $\bP^{n,h}_{\gL,\gb}$ converges when $\gL$ exhausts $\bbZ^2$
to a  limit $\bP^{n,h}_{\gb}$ which satisfies
 \begin{equation}\label{cloxf}
\left| \bE^{n,h}_{\gL,\gb}[f(\phi) ]-\bE^{n,h}_{\gb}[f(\phi)]\right|\le C|A|e^{\frac{\gb}{100}d( A,\gL)}.
\end{equation}
Moreover for we have for all $\gL$  and $h\in [h^*_{n+1},h^*_n]$,
\begin{equation}\label{clax}
\left| \bE^{n,h}_{\gL,\gb}[ f(\phi)g(\phi) ]- \bE^{n,h}_{\gL,\gb}[ f(\phi)]  \bE^{n,h}_{\gL,\gb}[g(\phi) ] \right| \le C|A|e^{\frac{\gb}{100}d( A,B)}.
\end{equation}

\end{proposition}

Note that a consequence of \eqref{cloxf} is that the measure $\bE^{n,h}_{\gb}$ is translation invariant.
Note that  \eqref{clox} ensures that the convergence of $\bP^{n,h}_{\gL,\gb}[f(\phi)]$ (which is a sequence of continuous functions in $h$)
is uniform in $h\in [h^*_{n+1},h^*_n]$ and thus it implies that $\bP^{n,h}_{\gb}[f(\phi)]$ is continuous on that interval.

\begin{proof}
Similarly to Proposition \ref{propinfi}, the proof of \eqref{cloxf} amounts to evaluating total variation distances for the distribution of $\phi\restrict_A$.
We are going to show that in fact the result can be deduced as a consequence of Proposition \ref{propinfi} and Theorem \ref{converteo}.
\medskip

Firstly, we notice that conditionally on the set of external contour $\Upsilon^{\ext}_n(\phi)$, 
the conditional distribution of the field in $\bar \gamma$ for $\gamma \in \Upsilon^{\ext}_n(\phi)$ is independent of that in $\gL\setminus \bar \gamma$
and given by $\bP_{\gamma}^{n,h}$ (recall \eqref{defpnu}).
Hence the distribution of $\phi\restrict_{A}$ is completely determined by the set of external contours which intersects $A$ (recall the definition of $\cC'_A$ below 
\eqref{locallyfinite}). For this reason one has 
\begin{multline}\label{dafa}
\| \bP^{n,h}_{\gL,\gb}(\phi\restrict_A\in \cdot )- \bP^{n,h}_{\gL',\gb}(\phi\restrict_A\in \cdot)\|_{TV}\\\le 
\| \bP^{n,h}_{\gL,\gb}(\Upsilon^{\ext}_n(\phi) \cap \cC'_{A}\in \cdot )- \bP^{n,h}_{\gL',\gb}(\Upsilon^{\ext}_n(\phi) \cap \cC'_{A}\in \cdot)\|_{TV}\\
\le \bbP^{w^{h}_{n}}_{\gL}(\Upsilon \cap \cC'_{A}\in \cdot )-  \bbP^{w^{h}_{n}}_{\gL}(\Upsilon \cap \cC'_{A}\in \cdot)\|_{TV}.
\end{multline}
The second inequality is due to the fact that the $\Upsilon^{\ext}_n$ is the same as that of $\Upsilon^{\ext}$ 
the set of external contours associated with $\Upsilon$ under $\bbP^{w^{h}_{n}}_{\gL}$.
We can conclude using Proposition \ref{propinfi} and Theorem \ref{converteo}.
For the proof of \eqref{clax}, we condition the expectation according to the realization of the set of external contour.
Let us start by observing that when $\Upsilon^{\ext}_n(\phi)\cap \cC'_A \cap \cC'_B=\emptyset$, we have
\begin{equation}\label{troum}
 \bE^{n,h}_{\gL,\gb}[f(\phi)g(\phi) \ | \ \Upsilon^{\ext}_n(\phi)]= \bE^{n,h}_{\gL,\gb}[f(\phi) \, | \, \Upsilon^{\ext}_n(\phi)]  \bE^{n,h}_{\gL,\gb}[g(\phi) \, | \, 
 \Upsilon^{\ext}_n(\phi)].
\end{equation}
Now applying Proposition \ref{propinfi} for the distribution of  $\Upsilon^{\ext}_n(\phi)$ (which is distributed like the set of external contour under 
$\bbP^{w^{h}_{n}}_{\gL}$) and for the functions $\tilde f$ and $\tilde g$ which are the conditional expectation of $f$ and $g$ given the external set of contour,
we obtain that 
\begin{multline}
 \left| \bE^{n,h}_{\gL,\gb}\left[ \bE^{n,h}_{\gL,\gb}[f(\phi) \, | \, \Upsilon^{\ext}_n(\phi)]  \bE^{n,h}_{\gL,\gb}[g(\phi) \ | \ 
 \Upsilon^{\ext}_n(\phi)]\right] - \bE^{n,h}_{\gL,\gb}[f(\phi)]  \bE^{n,h}_{\gL,\gb}[g(\phi)]   \right|
 \\ \le C|A|e^{-(\gb/100) d(A,B)}.
\end{multline}
And we conclude by noticing that from  \eqref{troum} we have 

\begin{multline}
 \left| \bE^{n,h}_{\gL,\gb}[f(\phi)g(\phi)]- \bE^{n,h}_{\gL,\gb}\left[ \bE^{n,h}_{\gL,\gb}[f(\phi) \, | \, \Upsilon^{\ext}_n(\phi)]  \bE^{n,h}_{\gL,\gb}[g(\phi) \ | \ 
 \Upsilon^{\ext}_n(\phi)]\right]\right|\\ \le 
 \bP^{n,h}_{\gL,\gb}[\Upsilon^{\ext}_n(\phi)\cap \cC'_A \cap \cC'_B=\emptyset]\le |A|e^{-(\gb/2) d(A,B)}.
\end{multline}
where in the last inequality we used \eqref{pignouf}.

 \end{proof}

\subsection{Exponential decay for Ursell function and proof of Proposition \ref{infinitz}}\label{laprov}

The fact that exponential decay of correlation implies differentiability relies 
on a well established theory exposed e.g.\ in \cite{cf:vDK} in the case of the Ising model.
The argument displayed in \cite{cf:vDK} adapts verbatim to our problem.
For the sake of completeness, we provide here the main steps.

\medskip

In order to prove Proposition \ref{infinitz} we are going to show that 
for every $h\in [h^*_{n+1},h^*_n]$, $\gL\subset \bbZ^2$, $k\ge 1$, there exists a constant $C_k$ such that 
\begin{equation}\label{bootk}
 \left|\partial^k_h\left( \log \cZ^{n,h}_{\gL,\gb}\right)\right|\le C_k |\gL| 
\end{equation}
Then using Arzela-Ascoli's Theorem,
we can deduce from \eqref{bootk} that $\tf(\gb,h)$ is infinitely differentiable on  $(h^*_{n+1},h^*_n)$ and that 
\begin{equation}\label{koomits}
\lim_{N\to \infty}\frac{1}{N^2}\partial^k_h\left(\log  \cZ^{n,h}_{N,\gb}\right)= \partial^k_h \tf(\gb,h).
\end{equation}
To prove \eqref{bootk}, setting $\delta_x:=\ind_{\{\phi(x)=0\}}$, we use the fact that
\begin{equation}\label{mood}
\partial^h_h\left(\log \cZ^{n,h}_{\gL,\gb}\right)= \sum_{(x_1,\dots,x_k)\in \gL^k} \bE^{n,h}_{\gL,\gb}\left[ \delta_{x_1} \ ; \ \dots  \ ; \ \delta_{x_k} \right]
\end{equation}
where  $\bE^{n,h}_{\gL,\gb}\left[ \delta_{x_1} \ ; \ \dots  \ ; \ \delta_{x_k} \right]$ denotes the $k$ point truncated correlation function (or Ursell function)
defined by
\begin{equation}
\bE^{n,h}_{\gL,\gb}\left[ \delta_{x_1} \ ; \ \dots  \ ; \ \delta_{x_k} \right]:=
\sum_{\cP} (-1)^{|\cP|+1}(|\cP|+1)! \prod_{P\in \cP}\bE^{n,h}_{\gL,\gb}\left[\prod_{i \in P} \delta_{x_i}\right],
\end{equation}
where the sum in $\cP$ ranges over the set of partitions of $\lint 0,k \rint$, and $|\cP|$ denotes the cardinal of the partition.

The identity \eqref{mood} can be obtained by induction on $k$ and is a particular case of \cite[Equation (1.9)]{cf:vDK}.
We can conclude by using some decay estimates for these correlation functions.

\begin{proposition}\label{mickey}
For any $k$ there exists a positive constant $m_k$ such that for every $h\in (h^*_{n+1},h^*_n)$  and $x_1,\dots,x_k\in \gL$
 \begin{equation}
 \bE^{n,h}_{\gL,\gb}\left[ \delta_{x_1} \ ; \ \dots  \ ; \ \delta_{x_k} \right]\le C_k e^{-m_k d(x_1,\dots,x_k)}
 \end{equation}
 where $d(x_1,\dots,x_k)$ is the smallest  cardinal of a connected $\bbZ^2$ set containing $\{x_1,\dots,x_k\}$.
\end{proposition}
To deduce that \eqref{bootk} holds, one only needs to check that we have
\begin{equation}
\sum_{x_2,\dots,x_k\in \bbZ^d} e^{-m_k d(x_1,\dots,x_k)}<\infty. 
\end{equation}
This is an obvious consequence of
\begin{equation}
d(x_1,\dots,x_k)\ge \max_{i\in \lint 2,k\rint} |x_1-x_i|\ge \frac{1}{k-1} \sum_{i=2}^k |x_1-x_i|.
\end{equation}
\begin{proof}[Proof of Proposition \ref{mickey}]
 We can follow line by line the proof of \cite[Lemma 3.1]{cf:vDK} where  we replace \cite[Assumption (3.2)]{cf:vDK} by \eqref{clax}.
\end{proof}

\subsection{Proof of Proposition \ref{convirj}}\label{lapriv}

Recalling Equation \eqref{koomits}-\eqref{mood} for $k=1$, we only need to prove that 
for $h\in [h^*_{n+1},h^*_n]$ we have 
\begin{equation}
\frac{1}{2} J^{2n}\le \frac{1}{N}\bE^{n,h}_{N,\gb}\left[ \sum_{x\in \lint 1, N\rint}\delta_x\right]\le 2J^{2n}.
\end{equation}
We can in fact prove the inequality holds for $\bP^{n,h}_{\gL,\gb}\left[ \phi(x)=0 \right]$, uniformly in $\gL$ and $x$.
Let us start with the special case $n=0$, for which the upper bound is trivial.
For the lower bound we apply \eqref{godstim} to $w^h_n$ and $A=\{x\}$,
we obtain that 
\begin{multline}
 \bP^{n,h}_{\gL,\gb}[\phi(x)=0]\ge \bP^{n,h}_{\gL,\gb}[ \text{ No contour in } \Upsilon_n(\phi) \text{ encloses } x]
 =\bbP^{w^h_n}_{\gL}\left[ \Upsilon \cap \cC'_{\{x\}}=\emptyset \right]\\
 =\exp\left(-\sum_{\bC\in \cQ(\gL,x,\emptyset)} (w^h_n)^T(\bC)\right),
\end{multline}
where 
\begin{equation}
 \cQ(\gL,\{x\},\emptyset):=\{ \bC \in \cQ(\cC_{\gL}) \ : \ \exists \gamma \in \bC, x\in \bar \gamma\}.
 \end{equation}
The estimate \eqref{dabound2} is then sufficient to conclude that 
 \begin{equation}
\sum_{\bC\in \cQ(\gL,x,\emptyset)} |(w^h_n)^T(\bC)|\le 1/2,
\end{equation}
which is sufficient for our purpose.

\medskip

When $n\ge 1$, we let  $\bg_x$ be the unique external contour enclosing $x$ whenever it exists.
As a vertex not enclose in any contour is by definition at level $n$ we have  
\begin{multline}
 \bE^{n,u}_{N,\gb}[ \phi(x)=0]= \sum_{\gamma\in \cC_N : x\in \bar \gamma} \bP^{n,u}_{N,\gb}\left[ \phi(x)=0\, ; \, \bg_x=\gamma \right]\\
 =\sum_{\gamma\in \cC_N : x\in \bar \gamma} \bP^{n,u}_{N,\gb}\left[ \bg_x=\gamma \right]\bP^{n,u}_{\gamma}\left[ \phi(x)=0\right],
\end{multline}
where in the last equality, we used that the conditional distribution of $\phi$ restricted to $\bar \gamma$ is given by
$\bP^{n,u}_{\gamma}$ defined in \eqref{defpnu}. We are going to show that most of the contribution to the sum 
is given by he negative contour $\gamma^{x,-}$ which satisfies $\bar \gamma^{x,-}=\{x\}$ (there are two contours of length $4$ such that $x\in \bar \gamma$ but 
$\gamma^{x,-}$ is the only one which contributes to the sum).
More precisely we are going to show that for $\gb$ sufficiently large and $u$ satisfying the assumption we have
\begin{equation}\label{precizion}
 \left|  \bP^{n,u}_{N,\gb}\left[ \phi(x)=0\, ; \, \bg_x=\gamma^{x,-} \right]- e^{-4\gb n} \right|\le  C e^{-4\gb n}\left( u+ e^{-4\gb}\right). 
\end{equation}
and 
\begin{equation}\label{pro}
\sum_{\gamma\in \cC_N : x\in \bar \gamma,\ |\tilde \gamma|\ge 6} \bP^{n,u}_{N,\gb}\left[ \phi(x)=0\, ; \, \bg_x=\gamma \right] \le C e^{-(4n+2)\gb},
\end{equation}
which yields a sharper result than required.
For \eqref{precizion}
using \eqref{godstim} for $A=\{x\}$ again, we obtain that
\begin{multline}
 \bP^{n,u}_{N,\gb}\left[ \bg_x=\gamma^{x,-}\right]= \bP^{w^h_n}_{\gL}[\Upsilon \cap \cC'_{\{x\}}=\gamma^{x,-}]\\
 =w^h_n(\gamma^{x,-})\exp\left(-\sum_{\bC\in \cQ(\gL,x,\{\gamma^{x,-}\})} (w^h_n)^T(\bC)\right),
\end{multline}
where 
\begin{equation}
 \cQ(\gL,\{x\},\{\gamma^{x,-}\}):=\{ \bC \in \cQ(\cC_{\gL}\setminus \gamma^{x,-}) \ : \ \exists \gamma \in \bC,\ x\in \bar \gamma \text{ or } \bC \perp \gamma^{x,-} \}.
 \end{equation}
We have 
\begin{equation}
w^h_n(\gamma^{x,-})=e^{-4\gb}\frac{1+(e^{u-1})e^{-4\gb n}}{1+(e^{u-1})e^{-4\gb{n-1}}},
\end{equation}
and as a consequence of \eqref{dabound2},
\begin{equation}
\left| \sum_{\bC\in \cQ(\gL,x,\{\gamma^{x,-}\})} (w^h_n)^T(\bC) \right| \le C e^{-4\gb},
\end{equation}
which is sufficient to yield \eqref{precizion}.

%
  \medskip
  
  For \eqref{pro}
  first, let us notice that we can discard the contribution of contours longer than $100 n$, because using a union bound and \eqref{pignouf} we have  
  \begin{equation}
 \bP^{n,u}_{N,\gb}\left[ \Diam(\bg_x)\ge 100 n \right]\le e^{-(4n+2)\gb}.
\end{equation}
For smaller contours,
we need an estimate $\bP^{n,u}_{\gamma}\left[ \phi(x)=0\right]$.
Recalling the definition of $\bP^{n}_{\gamma}$ below \eqref{tambem} we have
\begin{multline}\label{zipa}
\bP^{n,u}_{\gamma}\left[ \phi(x)=0\right]= 
 \frac{\bE^{n+\gep(\gamma)}_{\gamma}\left[\delta_x e^{u|\phi^{-1}(0)|-H(\phi^{-1}(\bbZ_-))} \right]}{\bE^{n+\gep(\gamma)}_{\gamma}\left[e^{u|\phi^{-1}(0)|-H(\phi^{-1}(\bbZ_-))} \right]}
 \\
 \le  \frac{e^{|\bar \gamma| u}\bE^{n+\gep(\gamma)}_{\gamma}\left[\delta_x \right]}{\bP^{n+\gep(\gamma)}_{\gamma}[\forall x\in \bar \gamma, \phi(x)\ge 1]},
 \le 4e^{- 4\gb(n+\gep(\gamma))}.
 \end{multline}
 where in the last inequality, we used Proposition \ref{rouxrou} to estimate both the numerator and the denominator.
 Now using \eqref{pignouf}, one has 
\begin{equation}\label{zipi}
 \bP^{n,u}_{N,\gb}\left[ \bg_x=\gamma \right]\le w^h_n(\gamma) \le e^{-(\gb-1)|\tilde \gamma|} , 
\end{equation}
where the last inequality is a consequence of Proposition \ref{lesmall}.
Combining \eqref{zipa} and \eqref{zipi} as well a standard bound for the number of contour of a given length enclosing $x$ we find that
\begin{multline}
\sum_{\gamma\in \cC_N : \ x\in \bar \gamma,\ |\tilde \gamma|\ge 6,\ \Diam(\gamma)\le 100 n}
\bP^{n,u}_{N,\gb}\left[ \bg_x=\gamma \right]\bP^{n,u}_{\gamma}\left[ \phi(x)=0\right]\\\le  4e^{- 4\gb(n-1)}
\sum_{\gamma\in \cC_N : x\in \bar \gamma,\ |\tilde \gamma|\ge 6}e^{-(\gb-1)|\tilde \gamma|}\le C e^{- 4\gb(n+2)},
\end{multline}
which is sufficient to conclude the proof of \eqref{pro}.
\qed

\section{Properties of Gibbs measure: the proof for Theorem \ref{Gibbs}}\label{daverylast}

In this final section, we prove the remaining unproved statements from Theorem \ref{Gibbs}.
First in Section \ref{start}, we prove our statement concerning $\star$-connectivity of the level sets.
Then  in Section \ref{noon} we prove that there exists  no Gibbs states for $h\le h_w(\gb)$.
In Section \ref{cf}, we identify the contact fraction for each Gibbs states which has been obtained in Proposition \ref{conviark}. 
This yields in particular \eqref{variousgibbs}.
In Section \ref{zerob} we identify the minimal Gibbs states, which is the one obtained by taking the limit of zero boundary condition.
In Section \ref{unik} we prove uniqueness of Gibbs states at differentiability points, and  in Section \ref{sanduba}, we prove that at 
  $\bP^{n,h^*_n}_{\gb}$ and   $\bP^{n-1,h^*_n}_{\gb}$ are respectively the maximal and minimal Gibbs states corresponding to $h^*_n$ proving \eqref{ssanduba}.

\subsection{Percolative properties of level sets}\label{start}

Let us check that for all $h\in [h^*_n,h^*_{n+1}]$, the random field $\phi$ percolates at level $n$ under  $\bP^{n,h}_{\gb}$.
The external contour lines of $\phi$ under  $\bP^{n,h}_{\gb}$ can be obtained by considering the set of external contour of a sample of 
$\bbP^{\go^h_n}$. In particular, this implies that almost surely there are no infinite contour lines.

\medskip

As each maximal connected components of $\phi^{-1}[n+1,\infty)$ resp. $\phi^{-1}(-\infty,n-1]$ is enclosed in a positive contour resp. negative contour,
this implies that they are all finite. We can even prove using a union bound argument and Theorem  \ref{converteo} that the diameter of the 
largest such component in a box of side-length $N$ is of order $\log N$).

\medskip

Proving the existence of an infinite component for $\phi^{-1}(n)$ is more tricky as some points which are not enclosed by any contour can belong to finite clusters of 
$\phi^{-1}(n)$. Now our result will hold if we can prove that 
\begin{equation}\label{thatsallfolks}
\bbR^2\setminus \union_{\gamma\in \Upsilon^{\ext}_n(\phi)} \bar \gamma \quad  \text{ has a unique unbounded connected component },
\end{equation}
where in the equation above, with a small abuse of notation,  $\bar \gamma$ denotes the closed subset of $\bbR^2$ enclosed by $\gamma$.

\medskip

Our idea is to compare the set $\union_{\gamma\in \Upsilon^{\ext}_n(\phi)} \bar \gamma$ with the occupied set of a Poisson Boolean percolation process \cite{cf:ATT, cf:Hall, cf:G}. 
We know  that the set of external contour under $\bP^{n,h}_{\gb}$ can be obtained as a subset of a sample of $\bbP^{\go^h_n}$
which itself is dominated (e.g. by \cite[Lemma 4.4]{cf:part1}) by 
by a random collection of contours $\chi$ where each contour $\gamma$ is  present independently with probability 
$\frac{\go^h_n(\gamma)}{1+\go^h_n(\gamma)}\le e^{(1-\gb) |\tilde \gamma|}$ (the inequality being a consequence of Theorem \ref{converteo}).

\medskip

We let the reader check that for $\gb$ sufficiently small   $\union_{\gamma\in \chi} \bar \gamma$ is stochastically dominated by a continuum percolation process,
where obstacles are balls whose centers are distributed according to a Poisson point process with intensity $\gl(\gb)=e^{-\gb}$ and 
whose radius are IID with standard exponential distribution (see e.g. \cite{cf:ATT, cf:Hall, cf:G} for a more formal definition).

\medskip

It has been proved that for $\gl$ sufficiently small the vacant set for such a Boolean percolation process percolates, and  that the occupied set is
only composed of bounded connected components (a much stronger result is displayed  in \cite[Theorem 1]{cf:ATT} with optimal assumptions,
but the statement we need can also be extracted from earlier work
e.g. \cite{cf:G, cf:Hall}) . This proves \eqref{thatsallfolks} and conclude our reasoning. 

\qed

\subsection{Absence of Gibbs state for $h\le h_w(\gb)$ }\label{noon}

As by the DLR relation a Gibbs states can always be obtained as a limit of finite volume measures with random boundary condition,
we know that the limit obtained with zero boundary condition, is, if finite, the minimal Gibbs state.
We are going to prove the following result which implies divergence of the distribution of $\phi$.

\begin{proposition}
For $\gb$ sufficiently large, for $h\le h_w(\gb)$ we have for any $x\in \bbZ^d$ and any $K>0$

\begin{equation}
\lim_{\gL\to \bbZ^2} \bP^{h}_{\gL,\gb}(\phi(x)\le K)=0.
\end{equation}
\end{proposition}

Note that by monotonicity in $h$ (Corollary \ref{FKalt}), it is sufficient to check the statement for $h=h_w(\gb)$.

\begin{proof}
 Now using the DLR relation for the neighborhood of $x$, and the definition of the measure one obtains that for any $k\ge 0$ one has 
 \begin{equation}
  \bP^{h_w(\gb)}_{N,\gb}(\phi(x)=k+1 \ | \ \phi(y),\ y\sim x)\le e^{4\gb}  \bP^{h_w(\gb)}_{N,\gb}(\phi(x)=k+1 \  |  \ \phi(y),\ y\sim x).
 \end{equation}
 This readily implies that for an explicit constant $C(\gb,K)$ one has
 \begin{equation}
  \bP^{h_w(\gb)}_{\gL,\gb}(\phi(x)\le K)\le C(\gb,K)  \bP^{h_w(\gb)}_{\gL,\gb}(\phi(x)=0).
 \end{equation}
To conclude we just need to show that
 \begin{equation}\label{suit}
 \lim_{\gL \to \bbZ^2}\bP^{h_w(\gb)}_{\gL,\gb}(\phi(x)=0)=0.
 \end{equation}
 As a consequence of Theorem \ref{oldmain} and \eqref{difencial},
 we have 
 \begin{equation}
  \lim_{N\to \infty} \frac{1}{N^2}\sum_{x\in \lint 1,N\rint} \bP^{h_w(\gb)}_{N,\gb}(\phi(x)=0)=\partial_h \tf(\gb,h_w(\gb))=0.
 \end{equation}
 As by monotonicity (Corollary \ref{FKalt}), each term in the sum is larger than the limit one wants to compute, we obtain \eqref{suit}.

\end{proof}

\subsection{Identifying the contact fraction: the proof of \eqref{variousgibbs}}\label{cf}

Let us prove that for any $h\in [h^*_{n+1},h^*_n]$ we have 
\begin{equation}\label{wootz}
 \bP^{n,h}_{\gb}[\phi(x)=0]=\partial_h \tf(\gb,h),
\end{equation}
where the derivative as to be understood as the derivative on the right for $h=h^*_{n+1}$ and on the left for 
$h=h^*_{n}$.
We already know as a consequence of \eqref{koomits} that for $h\in(h^*_{n+1},h^*_n)$ we have
\begin{equation}\label{zeone}
 \lim_{N\to \infty} \frac{1}{N^2}\sum_{x\in \lint 1,N \rint^2 }\bP^{n,h}_{N,\gb}[\phi(x)=0]=\partial_h \tf(\gb,h).
\end{equation}
Note that the statement can be extended to the boundary of the interval $h\in\{h^*_{n+1},h^*_n\}$ using that 
the second derivative $\partial^2_h \log Z^{n,h}_{N,\gb}$ is uniformly bounded on the interval $[h^*_{n+1},h^*_n]$ (recall \eqref{koomits}).

\medskip

\noindent A consequence of the exponential decay of correlation \eqref{cloxf} is that 
\begin{equation}\label{zetwo}
\left|\sum_{x\in \lint 0,N \rint^2 }\bP^{n,h}_{\gb,N}[\phi(x)=0]- N^2\bP^{n,h}_{\gb}[\phi(x)=0]\right|\le C N,
\end{equation}
and thus \eqref{wootz} is deduced from the combination of \eqref{zeone} and \eqref{zetwo}.

\subsection{Limit with zero boundary condition}\label{zerob}

So far we have only shown the existence of translation invariant Gibbs measure and non-uniqueness at the phase transition points.
To conclude we need an argument to show that they are the only one.

\medskip

A first step is to show convergence of the measure when having zero boundary condition to a translation invariant limit which has the right contact fraction.
The proof uses essentially the same idea as those to prove a similar result for wetting of the harmonic crystal
\cite[Section 5]{cf:GL}.

\begin{proposition}
For any $h>h_w(\gb)$
 the sequence of measure $\bP^{h,0}_{\gL,\gb}$ converge to an infinite volume limit which we call $\tilde \bP^{h}_{\gb}$.
 We have for every $x\in 0$
 \begin{equation}\label{zioup}
  \tilde \bP^{h}_{\gb}[\phi(x)=0]=\partial^+_h \tf(\gb,h).
 \end{equation}
 
\end{proposition}

\begin{proof}
 A first observation is that $\bP^{h,0}_{\gL,\gb}$ restricted to $A\subset \gL$ increases stochastically when $\gL$ increases.
 This is a consequence of the DLR property (recall \eqref{DLR}) and  Corollary \ref{FKalt}: for $\gL'\supset \gL$ the restriction
  $\bP^{h,0}_{\gL',\gb}$ to $\gL$ corresponds to $\bP^{h,\Psi}_{\gL,\gb}$ with a random boundary condition $\Psi\ge 0$ which thus dominates $\bP^{h,0}_{\gL,\gb}$.
 
 \medskip
 
 The sequence of measure is tight because, from Corollary \ref{FKalt}, $\bP^{h,0}_{\gL,\gb}$ is dominated by $\bP^{n,h}_{\gL,\gb}$ which converges.
 Hence we obtain the existence of the limit.
 
 \medskip

 To prove the statement about the contact fraction,  let us set $\eta:= \tilde \bP^{h}_{N,\gb}[\phi({\bf 0})=0]$.  
 Note that by monotonicity for any $\gL$ and  $x\in \gL$ we have 
 \begin{equation}
  \tilde \bP^{h,0}_{\gL,\gb}[\phi(x)=0] \ge \eta,
 \end{equation}
and hence \eqref{difencial} implies that when the derivative exists  
\begin{equation}
  \partial_h\tf(\gb,h)\ge \eta. 
\end{equation}
Given $\gep>0$, using the definition of $\eta$ we can choose $K=K_\gep$ sufficiently large which satisfies 
$$ \bP^{h}_{\lint -K,K\rint,\gb}[\phi(x)=0]\le \eta+\gep.$$
By monotonicity we have also 
$$\bP^{h,0}_{N,\gb}[\phi(x)=0]\le \eta+\gep$$ 
for all $x$ such that $d(x,\partial \gL_N)\ge K+1$.
Hence we have 
\begin{equation}
 \sum_{x\in \gL_N}    \tilde \bP^{h}_{N,\gb}[\phi(x)=0]\le 4 (K+1) N+ N^2(\eta+\gep).
\end{equation}
Using \eqref{difencial} again, we conclude that  $\partial_h\tf(\gb,h)\le \eta+\gep$ and thus obtain that \eqref{zioup} holds at all differentiability points.
To conclude, we remark that  $\bP^{h}_{N,\gb}[\phi({\bf 0})=0]$ being the infimum (in $N$) of continuous non-decreasing function, its limit has to be right-continuous on the 
 at every point, from which we deduce that the results also holds where $\tf$ is not differentiable .
 
 \end{proof}

\noindent The next step is to show that the limit found above coincides with $\bP^{n,h}_{\gb}$.

\begin{proposition}\label{dabot}
For any $h\in [h^*_{n+1},h^*_n)$, we have 
\begin{equation}
\tilde \bP^{h}_{\gb}= \bP^{n,h}_{\gb}.
\end{equation}
\end{proposition}

 \begin{proof}
The proof relies on the fact that, by \eqref{wootz} and \eqref{zioup} the two measures have the same contact fractions,
while $\bP^{n,h}_{\gb}$ stochastically dominates $\tilde \bP^{h}_{\gb}$, 
because of the stochastic ordering induced by the boundary condition before taking the limit.

\medskip

\noindent Given $x\in \bbZ^2$, using the DLR equation for the neighborhood of $x$ we have
\begin{multline}\label{superdlr}
\bE^{n,h}[\phi(x)=0]=\bE^{h,0}[\phi(x)=0]=
\bE^{h,0}\left[ \frac{e^{h-\gb\sum_{y\sim x}\phi(y)}}{e^{h-\gb\sum_{y\sim x}\phi(y)}+ \sum_{k\ge 1} e^{-\gb\sum_{y\sim x}|\phi(y)-k| }} \right]\\
=\bE^{n,h}\left[ \frac{e^{h-\gb\sum_{y\sim x}\phi(y)}}{e^{h-\gb\sum_{y\sim x}\phi(y)}+ \sum_{k\ge 1} e^{-\gb\sum_{y\sim x}|\phi(y)-k| }} \right].
\end{multline}
Now the reader can check that the function on the right hand side is strictly decreasing in $\phi(y)$ for all $y$ neighboring $x$.
  As $\bP^{n,h}$ dominates $\bP^{h,0}$, the equality implies thus that $(\phi(y))_{y\sim x}$ must have the same distribution under the two measures.
By stochastic domination and translation invariance we conclude that the two measures are equal.
  
 \end{proof}

 \subsection{Uniqueness of Gibbs states}\label{unik}

 The key point is to prove the following
 \begin{lemma}\label{labon}
  If $h$ is a differentiability point of $\tf(\gb,h)$  then there is only one  translation invariant, finite mean Gibbs state. 
 \end{lemma}
 
 \begin{proof}
  Let $\nu^h$ be a translation invariant, finite mean Gibbs state. 
Let $\hat \phi$ be a boundary condition sampled according to $\nu^h$ (but for practical reason we write $\hat \nu^h$ for the distribution).
Then the DLR equation implies that the law of $\phi\restrict_{\gL}$ under $\bP^{h,\hat \phi}_{\gL,\gb}$ corresponds to the restriction of  $\nu^h$.
In particular this implies that for any increasing local  function

\begin{equation}
\nu^h(f(\phi))=\hat \nu^h\left(\bE^{h,\hat \phi}_{\gL,\gb}f(\phi)\right)\ge \bE^{h,0}_{\gL,\gb}[f(\phi)].
 \end{equation}
 Passing to the limit we conclude that $\nu^h$ dominates the limit obtained with zero boundary condition which has been identified in Lemma \ref{dabot}.
 Now from translation invariance we have 
 \begin{equation}\label{contrex}
\nu^h(\phi(x)=0)=\frac{1}{N^2}\hat \nu^{h}\bE^{h,\hat \phi}_{N,\gb}\left[\sum_{x\in \gL_N} \phi(x)=0\right].
 \end{equation}
 Now we are going to prove that for every $h'\in \bbR$ we have 
  \begin{equation}\label{evian}
  \lim_{N\to \infty}\frac{1}{N^2}\hat \nu^{h}\left( \log \cZ^{h',\hat \phi}_{N,\gb}\right)=\tf(\gb,h').
 \end{equation}
This follows from the fact that convergence holds with zero boundary condition 
and that the effect boundary conditions can be controlled via the following observation
$$|\log \cZ^{h,\hat \phi}_{N,\gb}-\log \cZ^{h}_{N,\gb}|\le \gb\max_{\phi}|\cH^{\hat \phi}_{N}(\phi)- \cH_{N}(\phi)|,$$
\begin{equation}
|\cH^{\hat \phi}_{N}(\phi)- \cH_{N}(\phi)|\le \sum_{x\in \partial \gL_N} \hat \phi(x).
\end{equation}
and we can conclude using the fact that from our assumptions  $\hat \nu^h\left(\sum_{x\in \partial \gL_N} \hat \phi(x)\right)\le CN$.
 
\medskip

Now we observe that \eqref{evian} and convexity imply that the right hand side of \eqref{contrex} converges to $\partial_h\tf(\gb,h)$,
and thus that $\hat \nu^h$ and $\tilde \bP^h_{\gb}$ have the same contact fraction. 
Using the same trick as in the proof of Lemma \ref{dabot} (recall \eqref{superdlr}) we prove that the two measures coincide. 
\end{proof}

\subsection{Stochastic sandwich at angular points: the proof of \eqref{ssanduba}}\label{sanduba}

Assume now that $h=h^*_n$. 
While the derivative of the free energy  does not exists the proof of Lemma \ref{labon} still implies that 
$$\nu^h(\phi(x)=0)\in [\partial^-_h \tf(\gb,h),\partial^+_h\tf(\gb,h)]$$ and that $\nu^h$ stochastically 
dominates the Gibbs states obtained in the limit with $0$ boundary condition which we know to be $\bP^{n-1,h}_\gb$.

\medskip

A last thing to prove is that $\nu^h$ is dominated by $\bP^{n,h}_{\gb}$.
Consider $\hat \phi$ being distributed according to the measure $\nu^h$ (we write $\hat \nu^h$)
and consider the finite volume measure corresponding to boundary condition $n\vee\hat \phi=\max(n,\hat \phi)$.

\medskip

If we allow $\phi(x)=\infty$ and make the set $\bbZ_+\cup\{\infty\}$ compact, then any sequence of measure is tight and thus admits a limit point.
We consider $\nu'$ a limit point of the following sequence of probability on $\gO_N$
\begin{equation}
\nu'_N(\cdot):=\frac{1}{(2N^2-2N+1)^2}\sum_{y\in \lint N-N^2,N^2-N\rint^2}\hat \nu^h \bP^{h,\hat \phi\vee n}_{\lint -N^2,N^2\rint^2,\gb}\left[ \left(\phi(x+y)\right)_{x\in \lint N,N\rint} \in \cdot \right].
\end{equation}
Note that by construction $\nu'$ is translation invariant and dominates both $\nu^h$ and $\bP^{n,h}_{\gb}$.
Moreover $\nu'$ also satisfies the following version of the DLR equation (recall \eqref{DLR}):
For
every finite subset $\gL$ of $\bbZ^2$ and for every local bounded continuous $g: (\bbZ_+\cup\{\infty\})^{\bbZ^d} \to \bbR$ - in 
particular, the limit of $g(\phi)$, when $\min_{x\in \gL}\phi(x)\to \infty$ , exists and we call it $g(\infty)$ -
then we have $\nu'$ almost surely
\begin{multline}\label{superdlr2}
\nu'\left[ g(\phi) \ | \ \phi(x)=\psi(x), \forall x\in \partial \gL\right]\\=
\begin{cases}
                 \frac{1}{\cZ^{\psi,h}_{\gL,\gb}}\sum_{\phi\in \gO^+_{\gL}}e^{-\beta\cH^{\psi}_{\gL}(\phi)+h|\phi^{-1}(0)|}, 
                 &\text{ if } \psi(y)<\infty \text{ for all } y \in \partial \gL,\\
               g(\infty)  &\text{ if } \psi(y)=\infty \text{ for some } y \in \partial \gL.                 
                \end{cases}
\end{multline}
The statement is valid for every measure in the sequence for $N$ sufficiently large and passes to the limit by continuity
(see \cite[Equation (5.3)]{cf:GL} for a similar argument).

\medskip

\noindent Similarly to \eqref{evian} we have 
  \begin{equation}
  \lim_{N\to \infty}\frac{1}{N^2}\hat \nu^{h}\left( \log \cZ^{h',n\vee \hat \phi}_{N,\gb}\right)\le \tf(\gb,h').
 \end{equation}
and thus \eqref{difencial} implies readily that
$\nu'(\phi({\bf 0})=0)\ge \partial^-_h\tf(\gb,h)$
and hence as from stochastic comparison $\nu'(\phi({\bf 0})=0)\le \bP^{n,h}_{\gb}(\phi({\bf 0})=0)$, we conclude from \eqref{variousgibbs} that
$$\nu'(\phi({\bf 0})=0)\ge \partial^-_h\tf(\gb,h).$$

\medskip

 To conclude we need to prove that $\nu'=\bP^{n,h}_{\gb}$.
By stochastic domination, there exists  a coupling $(\phi_1,\phi_2)$ of the measures $\nu'$ and $\bP^{n,h}_{\gb}$, 
such that almost surely  $\phi_1(x)\ge \phi_2(x)$ for all $x$.  As we have $\{\phi_1(x)=0\}\subset \{\phi_2(x)=0\}$,
the fact that the two event have equal probability implies that almost surely
$$\phi^{-1}_1\{0\}=\phi^{-1}_2\{0\}.$$

\medskip

In particular we have

\begin{equation}
 \nu'( \exists x\in \bbZ^2, \phi(x)=0)= \bP^{n,h}_{\gb}( \exists x\in \bbZ^2, \phi(x)=0)=1.
\end{equation}
Then replicating the argument in \cite[Section 5]{cf:GL}, we deduce that 
\begin{equation}
 \nu'( \forall x\in \bbZ^2, \phi(x)<\infty)=1.
\end{equation}
To conclude we use the DLR relation in a neighborhood of $x$ for $\nu'$ and $\bP^{n,h}_{\gb}$ like in Equation \eqref{superdlr} to prove that the two measures coincide.

\textbf{Acknowledgements:} We are grateful to Kenneth Alexander for enlightening discussions.
This work has been performed in part
during a visit at the Institut Henri Poincar\'e (2017, spring-summer trimester)  supported by the Fondation de Sciences
Math\'ematiques de Paris. The author thanks the Institute for its hospitality.
The author also acknowledges the support of a productivity grand from CNPq as well as a JCNE grant from FAPERJ.

\appendix

\section{Cluster Expansion estimates}\label{apclus}

\subsection{Proof of Lemma \ref{finalfrontier}}\label{pfinal}

From the expression \eqref{fraconv} and translation invariance we have 
\begin{equation}
 |\bar \gamma|\tf(w)= \sum_{x\in \bar \gamma }\sum_{\{\bC\in \cQ \ : \ x-(1/2,1/2) \in \bC\}} \frac{1}{|\bC|}w^T(|\bC|).
\end{equation}
In this sum, the coefficient of $w^T(|\bC|)$ for a  cluster in $\cQ(\cC_{\gga})$ is $1$ while the other clusters have a coefficient between zero and one.
thus we have 
\begin{equation}
 \left|Z[\gamma,w]- |\bar \gamma|\tf(w)| \right|\le \sum_{\{\bC\in \cQ\setminus \cQ(\cC_{\gamma}) \ : \ \bC\cap [\bar \gamma-(1/2,1/2)]\ne \emptyset   \}} w^T(|\bC|).
\end{equation}
Now note that for all clusters in the right hand side, there is at least one site in $\tilde \gamma$ that belongs to $\bC$. Thus using translation invariance,
given a fixed $x^*_0\in (\bbZ^2)^*$ we have 
\begin{equation}
  \left|Z[\gamma,w]- |\bar \gamma^*|\tf(w)| \right|\le |\tilde \gamma|\sum_{\bC\in \cQ \ : \ x_0^* \in \bC} w^T(|\bC|)\le \frac{1}{4}|\tilde \gamma|,
\end{equation}
where the last inequality is valid for $\gb$ sufficiently large as a consequence of \eqref{dabound2}.
\qed

\subsection{Proof of Proposition \ref{propinfi}}\label{ppropinf}

 The second line of \eqref{convergenz} can be deduced from the first, by considering a sequence $\bL'$ that exhausts $\cC$.
The first result corresponds to evaluating the total variation distance between the respective distribution of $\Upsilon\cap \cC'_A$ under
 $\bbP^w_{\bL}$ and $\bbP^w$, which is equal to
 \begin{equation}
\sum_{\gG \in \cK(\cC'_A)} \left(\bP^w_{\bL}\left( \Upsilon \cap \cC'_A=\gG \right)-   \bP^w_{\bL'}\left(\Upsilon \cap \cC'_A=\gG\right) \right)_+
 \end{equation}
We set for this proof $d:=d(A,\bL^{\cc})$.
We consider first the contribution to the above sum of $\gG$ which contains a contour $\gamma$ of large diameter  , that is such that 
$|\Diam(\gamma)|\ge d/3$.
Simply using the fact that by Peierls argument (see e.g.\ \cite[Lemma 4.4]{cf:part1}) we have for every $\gamma$
\begin{equation}\label{pignouf}
\bP^{w}_{\bL}[\gamma \in \Upsilon]\le \frac{w(\gamma)}{1+w(\gamma)}. 
\end{equation}
We deduce from the assumption \eqref{usualchoice}, summing over all possible such contours  that for $\gb$ sufficiently large 
\begin{equation}
 \bbP^w_{\bL}\left( \exists \gamma \in \Upsilon, \ \bar\gamma \cap A\ne \emptyset , |\tilde \gamma|\ge \frac{2d}{3}\right)\le 
 |A| e^{-\gb \frac{d}{2}}.
\end{equation}
Now from the definition of $d(A,\bL^{\cc})$, if $\gG$ does not contain a contour of large diameter then all contours in it belongs to $\bL$ and $\bL'$.
In that case we have, from \eqref{godstim}
\begin{multline}
\left( \bP^w_{\bL}\left( \Upsilon \cap \cC'_A=\gG \right)-   \bP^w_{\bL'}\left(\Upsilon \cap \cC'_A=\gG\right) \right)_+
\\= \bP^w_{\bL}\left( \Upsilon \cap \cC'_A=\gG \right) \left[ e^{ \left(\sum_{\bC'\in \cQ(\bL',A,\gG)} w^T(\bC')-
  \sum_{\bC \in \cQ(\bL,A,\gG)} w^T(\bC)\right)_+  }-1 \right].
\end{multline}
And we can conclude provided that we can show that the difference in the exponential is small.
We notice that under the assumption that $\Diam(\gamma)\le d/3$ for $\gamma\in \gG$ we have
 \begin{equation}
\cQ(\bL,A,\gG)\triangle\cQ(\bL',A,\gG)\subset \cQ_1.
 \end{equation}
 where $\triangle$ stands for the symmetric difference between sets and 
 \begin{equation}
\cQ_1:=\{ \bC \ : \ \exists \gamma_1 ,\gamma_2 \in \bC, \min_{x\in A, y\in \bar \gamma_1} |x-y|\le d/3, \max_{x\in A, y^*\in \gamma_2}|x-y^*|\ge d \}
 \end{equation}
 where the existence of $\gamma_1$ is justified by the fact that $\bC$ must contain a contour that either is connected to $\gG$ (and the inequality follows from
 the assumption that contours in $\gG$ have diameter smaller than $d/3$), or belong to $\cC'_A$
and that of $\gamma_2$ by the fact that 
 it must contain a contour in $\bL\triangle\bL'$.
This implies in any case that the diameter of $\bC$ is larger than $d/2/$ so that $L(\bC)\ge d$ and one can conclude using \eqref{dabound2} that 
 \begin{equation}
  \sum_{\bC\in \cQ_1} |w^T(\bC)|\le C d^2 |A|e^{-\gb d/2}\le C|A|e^{-\gb d/4},
\end{equation}
were $|A|d^2$ gives a bound for the number of points at distance $d/3$ or less from $A$.

\medskip

 Let us now move to the proof of \eqref{decayze}.
 This is simply a bound on the total variation distance between the distribution $\Upsilon \cap \cC'_{A\cup B}$ and what we obtain by considering the product distribution of  $\Upsilon \cap \cC'_A$ and $\Upsilon \cap \cC'_B$. We must prove 
 \begin{multline}
\sum_{\gG  \in \cK(\cC'_{A\cup B})}\left( \bbP^w_{\bL}( \cC'_{A\cup B} \cap \Upsilon= \gG)-  \bbP^w_{\bL}(\cC'_A\cap \Upsilon=  \cC'_A\cap \gG)
 \bbP^w_{\bL}(\cC'_B\cap \Upsilon=\cC'_B\cap \gG ) \right)_+\\
 \le |A| e^{-c\gb d(A,B)}.
\end{multline}
A first step is to discard the possibility of having contour that comes to the neighborhood in the two regions. More precisely using \eqref{pignouf}
one can check that if \eqref{usualchoice}
holds we have 
\begin{equation}\label{badseeds}\begin{split}
\bbP^w_{\bL}\left(\exists \gamma \in \Upsilon \cap \cC'_{A}, \ \Diam(\gamma)\ge d/4 \right)&\le |A| e^{\gb d(A,B)/8},\\
\bbP^w_{\bL}\left(\exists \gamma \in \Upsilon \cap \cC'_{B}, \ \min_{x\in A, y^*\in \gamma}|x-y^*|\le d/4 \right) &\le |A| e^{-c\gb d(A,B)/8}.
\end{split}
\end{equation}
Now if  $\gG$ is such that  $ \gG_1:= \cC'_A\cap \gG$ and $\gG_2:=\cC'_B\cap \gG$ are disjoint,
using \eqref{godstim}, and observing that  $\cQ(\bL,A,\gG_1)\cup \cQ(\bL,B,\gG_2)= \cQ(\bL,A\cup B,\gG)$
\begin{multline}
\left(  \bbP^w_{\bL}( \cC'_{A\cup B} \cap \Upsilon= \gG)-  \bbP^w_{\bL}(\cC'_A\cap \Upsilon=   \gG_1)
 \bbP^w_{\bL}(\cC'_B\cap \Upsilon=\gG_2 )  \right)_+
\\= \bbP^w_{\bL}( \cC'_{A\cup B} \cap \Upsilon= \gG)\left[e^{\left(\sum_{\bC\in \cQ(\bL,A,\gG_1)\cap\cQ(\bL,B,\gG_2) } w^T(\bC)\right)_+}-1\right].
\end{multline}
Now if $\gG$ does not include any contours of the type considered in \eqref{badseeds}, 
\begin{equation}
  \cQ(\bL,A,\gG_1)\cap\cQ(\bL,B,\gG_1) \subset \cQ_2 
 \end{equation}
 where, using the notation $d:=d(A,B)$, we define
 \begin{multline}
  \cQ_2 :=\{ \bC\in \cQ \ : \  \exists \gamma_1, \gamma_2 \in \bC, \min_{x\in A, y\in \bar \gamma_1} |x-y|\le d/4, \  \max_{x\in A, y\in \bar \gamma_2}\ge 3d/4\}.
   \end{multline}
   The existence of $\gamma_1$ is justified by the fact that some contour in $\bC$ must either be connected with $\gG_1$ or intersect $A$ and that of $\gamma_2$
   by the fact that some contour in $\bC$ must either be connected with $\gG_2$ or intersect $B$.
   This implies in particular that the diameter of $\bC$ is larger then $d/2$, so that one must have $L(\bC)\ge d$ and we can also conclude using \eqref{dabound2} that 
 \begin{equation}
  \sum_{\bC\in \cQ_2} |w^T(\bC)|\le C d^2 |A|e^{-\gb d/2}\le C|A|e^{-\gb d}.
\end{equation}
\qed
\subsection{Proof of \eqref{infin} and \eqref{infde}}\label{lazt}

The proof is very similar to the one of the lower bound for Proposition \ref{convirj} displayed in Section \ref{lapriv}.
Let us treat only the case of \eqref{infde} as \eqref{infin} is very similar.

\medskip

We need to consider only the contribution of configuration for which the only contour in $\cC'_{\{x,y\}}$ is given by $\gamma^+_{x,y}$ the positive contour
of length $6$  which enclose $x$ and $y$. 
Using the definition of $ \bP_{\bL,\gL,\gb}$ we obtain that 

\begin{equation}
 \bP_{\bL,\gL,\gb}[\min\phi(x),\phi(y))\ge n]\ge \bP^{w_{\gb}}_{\bL}\left[\Upsilon\cap \cC'_{\{x,y\}}=\gamma^+_{x,y}\right]e^{-6\gb(n-1)}.
\end{equation}
Applying \eqref{godstim} for the contour weight \eqref{specialweight}
\begin{equation}
  \bP^{w_{\gb}}_{\bL}\left[\Upsilon\cap \cC'_{\{x,y\}}=\gamma^+_{x,y}\right]= \frac{1}{e^{6\gb}-1}
  \exp\left(-\sum_{\bC\in \cQ(\bL,\{x,y\},\{\gamma^+_{x,y}\})} w^T(\bC)\right).
\end{equation}
The sum in the exponential can be seen to be small as a consequence of \eqref{dabound2}, and this allows to conclude by choosing $\gb$ sufficiently large.

\end{document}